\documentclass[11pt,american]{article}
\usepackage[T1]{fontenc}
\usepackage{babel}
\usepackage{csquotes}
\usepackage{amsmath}
\usepackage{graphicx}
\usepackage{amssymb}
\usepackage{tikz}
\usetikzlibrary{positioning,arrows.meta,calc}
\usepackage{pgfplots}
\usepackage{hyperref}
\usepackage[left=2cm, right=2cm]{geometry}

\DeclareMathOperator{\sign}{sign}

\DeclareMathOperator{\R}{\mathbb{R}}
\DeclareMathOperator{\rank}{\operatorname{rank}}
\DeclareMathOperator{\range}{range}
\DeclareMathOperator{\spanop}{span}

\usepackage{amsthm}
\newtheorem{theorem}{Theorem}
\newtheorem{lemma}{Lemma}
\newtheorem{proposition}{Proposition}
\theoremstyle{definition}
\newtheorem{definition}{Definition}
\newtheorem{assumption}{Assumption}
\theoremstyle{plain}
\newtheorem{corollary}{Corollary}
\theoremstyle{remark}
\newtheorem{remark}{Remark}

\usepackage{color,ifpdf,graphicx}
\ifpdf %if using PDFTeX in PDF mode
  \DeclareGraphicsExtensions{.pdf,.png,.mps}
\else %if using TeX or PDFTeX in TeX mode
  \DeclareGraphicsExtensions{.eps,.bmp}
  \usepackage{pstricks}%variant: \usepackage{pst-all}
\fi

\title{Tracking Through Decoupling Singularities: A Singularity-Robust Homotopy-Continuation Extension of Feedback Linearization}
\author{Alex Borisevich, \href{mailto:akpc806b@gmail.com}{akpc806b@gmail.com}}

\begin{document}

\maketitle

\begin{abstract}
Input--output feedback linearization fails at decoupling singularities, where the decoupling matrix loses rank, the relative degree is lost, and the linearizing control becomes unbounded. This paper develops a singularity-robust trajectory-tracking controller for square nonlinear control-affine systems that tracks through isolated decoupling singularities with bounded control. The method recasts tracking as real-time arc-length homotopy continuation, equivalently a continuous-time Newton/Davidenko flow, and replaces the inverse decoupling matrix by the least-norm Moore--Penrose solution of an augmented matrix $A=[\Lambda\mid b]$, where $b$ is the homotopy direction. A transversality condition $w^T b \ne 0$, with $w$ in the left null space of the decoupling matrix, keeps the augmented matrix full row rank through a generic rank-one loss. The resulting flow agrees with feedback linearization away from the singular set, tracks with $O(1/k)$ error, and re-locks after each crossing. The theory also characterizes the reflection-versus-branch-crossing dichotomy at Whitney folds and relates the reflection case to a Filippov sliding mode. Extensions cover dynamic relative-degree-one minimum-phase systems and arbitrary relative degree via filtered-error reduction. The core results are formalized and machine-checked in the Lean~4 proof assistant, including the multi-parameter (higher-corank) repair criterion, the uniqueness of the sliding motion at the fold, and the Filippov (differential-inclusion) analysis of the fold model. Simulations include a redundant 2-DOF manipulator, relative-degree-one and relative-degree-two plants, and a dual-active-bridge series-resonant DC/DC converter, where the method performs bounded inversion across buck/boost and resonance singularities while preserving zero-voltage soft switching.
\end{abstract}

\medskip
\noindent\textbf{Keywords:} feedback linearization; loss of relative degree; decoupling matrix; decoupling singularity; nonlinear trajectory tracking; homotopy continuation; arc-length continuation; continuous-time Newton (Davidenko) method; augmented Jacobian; least-norm inverse; transversality; singularity-robust control; damped least squares; singularity-robust inverse kinematics; kinematic singularity; redundant manipulators; resolved-rate control; Type-2 parallel-mechanism singularity crossing; Filippov sliding mode; Whitney fold; monodromy; minimum phase; series resonant converter; dual-active-bridge; DC--DC converter; zero-voltage switching (ZVS); soft switching; formal verification; interactive theorem proving; Lean; differential inclusions.

\section{Introduction}

Feedback linearization is the standard route to tracking control for nonlinear systems with a well-defined relative degree \cite{Isidori1995,SlotineLi1991,Khalil2002}. For a square control-affine plant of relative degree one it inverts the decoupling matrix $\Lambda(x) = Dh(x)\,g(x)$ to cancel the nonlinearity and impose linear output dynamics. The construction is, however, intrinsically \emph{partial}: the linearizing law $u = \Lambda^{-1}(v - a)$ is defined only off the \emph{decoupling singularity} (equivalently, the locus of relative-degree loss) $\mathcal D = \{\det\Lambda(x) = 0\}$, and the commanded control grows without bound as the trajectory approaches $\mathcal D$. The same obstruction appears kinematically as the manipulator-Jacobian singularity of resolved-rate control, and dynamically as the Type-2 (parallel) singularity of closed-chain mechanisms.

This paper develops a controller that tracks a reference \emph{through} isolated decoupling singularities with bounded effort, using one formula and no case analysis. The idea is to embed the tracking problem in an arc-length homotopy continuation and to replace the partial inverse $\Lambda^{-1}$ by the least-norm (Moore--Penrose) solution of the \emph{augmented} matrix $A = [\Lambda \mid b]$, where $b$ is the homotopy direction. Augmenting $\Lambda$ with the extra column $b$ repairs an isolated rank drop precisely when a transversality condition holds, and the resulting flow crosses the singularity at finite velocity while retaining the accuracy of feedback linearization elsewhere.

\paragraph{Contributions.} 
\begin{enumerate}
\item A kernel structure lemma (Lemma~\ref{lem:kernel}) characterizing exactly when the augmentation repairs a rank drop, through a \emph{transversality} condition $w^Tb \ne 0$ on the homotopy direction, and quantifying the collapse of parameter authority $\rho \to 0$ at the singularity.
\item A tracking theorem (Theorem~\ref{thm:crossing}) proving bounded crossing, bounded excursion of the locking error during the crossing, and exponential re-locking afterwards, together with a sharp failure statement when transversality is lost (Proposition~\ref{prop:failure}).
\item A resolution of the reflection-versus-crossing dichotomy: the single kernel direction cannot both advance the homotopy parameter and change branch, and the orientation choice selects between them. At a generic fold we prove reflection is the canonical least-norm lift realized as a Filippov sliding mode, crossing its $\mathbb Z/2$ monodromy image (Theorem~\ref{thm:reflection}).
\item Extensions to dynamic relative-degree-one systems (minimum phase) and to arbitrary relative degree via filtered-error reduction, and a proof that the controller reduces to feedback linearization up to $O(1/k)$ off the singular set.
\item Validation on a 2-DOF arm, dynamic relative-degree one and two plants, and a dual-bridge series resonant DC/DC converter exhibiting both a crossing (buck/boost fold) and a reflection (resonance) singularity.
\item A machine-checked formalization (Lean~4 / Mathlib) of the core results --- the kernel lemma with its multi-parameter generalization, the tracking bounds of Theorem~\ref{thm:crossing}(2)--(4), and the sliding-mode analysis of the fold, including a constructive treatment of Filippov (differential-inclusion) solutions (Section~\ref{sec:lean}).
\end{enumerate}

The rest of the paper is organized as follows. Section~\ref{sec:related} positions the method against the prior art. Section~\ref{sec:setup} recalls the system class and feedback linearization and gives the motivating example. Sections~\ref{sec:wellbehaved}--\ref{sec:tracking} develop the continuation construction, the kernel lemma, and the tracking theorem; Section~\ref{sec:reflection} analyzes the reflection--crossing dichotomy and proves the reachability criterion that selects between them. Sections~\ref{sec:dynamic}--\ref{sec:reldeg} extend to dynamic and higher-relative-degree systems and relate the method to feedback linearization; Section~\ref{sec:adaptive} states the multi-parameter continuation results and their reading for adaptive control; Section~\ref{sec:lean} reports the machine-checked formalization of the core results; Sections~\ref{sec:sims}--\ref{sec:converter} present the simulations and the converter application, and Section~\ref{sec:conclusion} concludes.

\section{Related work}\label{sec:related}

\paragraph{Singularity-robust inverse kinematics.} The closest antecedent is the damped-least-squares (DLS), or singularity-robust, inverse of the manipulator Jacobian introduced by Nakamura and Hanafusa \cite{Nakamura1986} and Wampler \cite{Wampler1986}, and refined into task-priority and variable-damping schemes by Chiaverini and co-workers \cite{Chiaverini1997,ChiaveriniHandbook2008}; see \cite{Deo1995} for a review. DLS replaces $J^{-1}$ by $J^T(JJ^T + \eta^2 I)^{-1}$, trading tracking accuracy for boundedness in a neighbourhood of the singularity, with the damping $\eta$ scheduled on the smallest singular value. Our construction differs in three respects. First, boundedness comes from the \emph{geometry} of the augmented map $A = [\Lambda \mid b]$ --- the extra homotopy column, not an added damping term --- and is exact rather than approximate at the singular instant. Second, away from the singularity the controller equals feedback linearization up to $O(1/k)$ (Proposition~\ref{prop:consistency}), so unlike DLS it introduces no accuracy penalty in the regular region. Third, the augmentation supplies a checkable transversality criterion (Definition~\ref{def:transversal}) that decides in advance whether a given singularity is crossable, and a clean characterization of what the extra degree of freedom does (reflection versus crossing), which damping does not provide. The price is that the method addresses \emph{isolated} decoupling singularities along the path rather than dense singular regions. Beyond uniform damping, refinements include selectively damped least squares, which adapts the damping per singular direction \cite{BussKim2005}, and hierarchical / task-priority formulations solved as (inequality-constrained) quadratic programs \cite{Escande2014,Moe2016}; more recently, predictive redundancy-resolution schemes attain singularity-free tracking by embedding the kinematics in a receding-horizon program \cite{HIMPC2024}. All of these \emph{avoid} or \emph{regularize} the singularity --- keeping the trajectory away from it or damping the response near it --- whereas the present method drives the output \emph{through} an isolated decoupling singularity while coinciding with the exact inverse away from it.

\paragraph{Continuation and continuous-time Newton methods.} Writing the root-finding flow $\dot x = -[Df]^{-1}\dot y$ as a differential equation is the continuous Newton (Davidenko) method \cite{Davidenko1953}, and following a deformed solution curve is the subject of predictor--corrector and arc-length continuation \cite{AllgowerGeorg1990}. In numerical algebraic geometry, generic-parameter homotopies avoid the discriminant locus with probability one (the ``$\gamma$-trick''), meeting singularities only at isolated points \cite{SommeseWampler2005}. Our controller is an arc-length continuation in which the homotopy parameter is driven by real time and the augmentation plays exactly the discriminant-avoidance role, but with two control-specific additions: the transversality analysis of the augmented Jacobian at a genuine rank drop, and the locking feedback that regulates the parameter error $e = t-\mu$ to $O(1/k)$. When a path does approach a singular endpoint, numerical algebraic geometry handles it with \emph{endgames} --- power-series or Cauchy-integral extrapolation --- and with certified predictor--corrector tracking \cite{Bertini2013,BeltranLeykin2012}. Our transversality condition plays the complementary role of keeping the flow \emph{off} the discriminant in the first place, so that a single least-norm step suffices at the crossing and no endgame is required.

\paragraph{Continuation in control.} Continuation has been used in control most prominently in Ohtsuka's Continuation/GMRES method for nonlinear model-predictive control \cite{Ohtsuka2004}, where the KKT system is tracked in time by a continuation flow. The present homotopy-for-tracking construction and its setpoint predecessor were surveyed by the author \cite{Borisevich2013}; that note develops the ``auxiliary controller as a drift term'' idea for regulation but leaves the trajectory-tracking theory, the transversality/crossing analysis, and the singularity guarantees incomplete. The present paper completes and supersedes that line: it gives the augmented-Jacobian kernel lemma, the bounded-crossing theorem, the reflection/crossing dichotomy, and the dynamic and higher-relative-degree generalizations that were missing there.

\paragraph{Crossing Type-2 singularities of parallel mechanisms.} For closed-chain mechanisms it is known that a Type-2 (parallel) singularity can be crossed if the applied wrench is reciprocal to the twist along the uncontrollable direction --- a dynamic criterion developed by Arakelian, Briot and co-workers, first with pre-planned optimal trajectories and later with workspace-enlarging and virtual-constraint controllers \cite{Briot2008,Pagis2015,Briot2016}. Our transversality condition $w^Tb \ne 0$ is the continuation-flow counterpart of that reciprocity criterion: the homotopy direction $b$ must have a nonzero component along the uncontrollable mode $w$. The difference is that our criterion arises from a purely kinematic rank argument on the augmented map, requires no inverse dynamic model, and comes with an explicit tracking-error guarantee rather than only a feasibility condition. This line has since moved from pre-planned optimal trajectories to controllers that cross without an off-line trajectory, through virtual-constraint and multi-model computed-torque designs \cite{Pagis2015,Pagis2017}, while a parallel body of work instead \emph{avoids} Type-2 configurations, detecting the responsible limbs via the angle between output twist screws \cite{OTSAvoid2024}. The present contribution is on the crossing side but purely kinematic, and the same transversality test carries over verbatim to serial decoupling singularities and to the dynamic control-affine setting of Section~\ref{sec:dynamic}.

\paragraph{Control through loss of relative degree.} Within nonlinear control proper, the decoupling singularity $\det\Lambda = 0$ is precisely the loss of (vector) relative degree, at which exact input--output linearization is undefined. Two classical responses are \emph{approximate} input--output linearization, which discards the ill-defined part of the dynamics in a neighbourhood of the singularity \cite{HauserSastry1992}, and \emph{switching} between control laws valid on either side of it \cite{TomlinSastry1998}. More recent work retains a prioritized normal form so that a singular sub-block of the input-gain matrix does not propagate to higher-priority channels \cite{PriorityFBL2023}. The continuation controller is a third route: rather than approximating, switching, or re-prioritizing, it augments the decoupling matrix with the homotopy direction and takes the least-norm solution, which remains bounded \emph{at} the singularity under transversality and reduces to exact feedback linearization away from it (Proposition~\ref{prop:consistency}).

\paragraph{Feedback linearization and flatness.} The method sits inside the input--output linearization framework \cite{Isidori1995,SlotineLi1991,Khalil2002}: it presupposes a well-defined relative degree and minimum-phase zero dynamics, and does not manufacture either where absent. In particular it is orthogonal to dynamic feedback linearization and flatness \cite{Fliess1995}, which enlarge the linearizable class by dynamic extension; the continuation idea wraps a (possibly extended) system and repairs its \emph{decoupling} singularities, a distinct obstruction from non-flatness. The relation is made precise in Section~\ref{sec:reldeg}: the controller is a singularity-robust extension of feedback linearization, identical to it off $\mathcal D$ and bounded on it.

\section{Problem setup and feedback linearization}\label{sec:setup}

We consider square control-affine MIMO systems
\begin{equation}\label{eq:sys}
\dot x = f(x) + g(x)\,u, \qquad y = h(x),
\end{equation}
with
\begin{itemize}
\item state $x \in \R^n$, input $u \in \R^m$, and output $y \in \R^m$ (square: $\dim u = \dim y = m$);
\item drift $f:\R^n\to\R^n$ and input matrix $g:\R^n\to\R^{n\times m}$ with columns $g_1,\dots,g_m$;
\item output map $h:\R^n\to\R^m$,
\end{itemize}
all assumed smooth on the domain of interest. Writing $Dh = \partial h/\partial x$, the (vector) relative degree is one on a set $S$ if the \emph{decoupling matrix}
\begin{equation}\label{eq:Lambda}
\Lambda(x) := Dh(x)\,g(x) \in \R^{m\times m}
\end{equation}
is nonsingular for $x\in S$; such systems are common where fast feedback is required, as in force-controlled robots and power electronics. Differentiating the output once,
\begin{equation}\label{eq:ydot}
\dot y = Dh(x)\,[f(x)+g(x)u] = a(x) + \Lambda(x)\,u, \qquad a(x) := Dh(x)\,f(x),
\end{equation}
so the input already appears in $\dot y$. Where $\Lambda$ is invertible, the feedback-linearizing law
\begin{equation}\label{eq:fl}
u = \Lambda(x)^{-1}\bigl(v - a(x)\bigr)
\end{equation}
reduces the output dynamics to $\dot y = v$, after which any linear design applies to $v$; this is input--output feedback linearization \cite{Isidori1995,SlotineLi1991,Khalil2002}. The construction is defined only off the \emph{decoupling singularity}
\begin{equation}
\mathcal D = \{\, x : \det\Lambda(x) = 0 \,\},
\end{equation}
where $\Lambda^{-1}$ ceases to exist and $\|u\|\to\infty$. Repairing exactly this failure --- tracking through $\mathcal D$ with bounded control --- is the subject of the paper.

\section{Motivating example}\label{sec:motivating}

Before developing the general method, we expose the obstruction on the simplest possible system. Consider the scalar system
\begin{equation}\label{toy_example}
\dot x = u, \qquad y = h(x) = x\,(x^2 - 1),
\end{equation}
i.e. \eqref{eq:sys} with $m = n = 1$, $f = 0$, $g = 1$, so the decoupling matrix is the scalar $\Lambda(x) = Dh(x) = 3x^2 - 1$. The cubic $h$ models, for instance, the current--voltage characteristic of a tunnel diode (Fig.~\ref{fig:toy_cubic}).

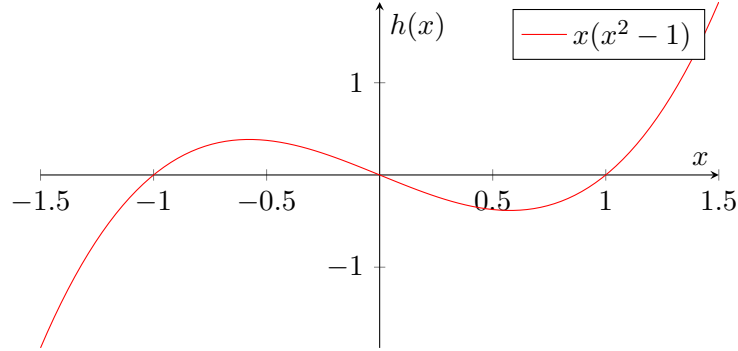
\begin{figure}[h]
\centering
\begin{tikzpicture}
\begin{axis}[
    axis lines = center,
    xlabel = $x$,
    ylabel = {$h(x)$},
    width = 0.6\textwidth,
    height = 0.35\textwidth,
]
\addplot [domain=-1.5:1.5, samples=100, color=red]{x*(x*x - 1)};
\addlegendentry{$x(x^2 - 1)$}
\end{axis}
\end{tikzpicture}
\caption{The cubic output map $h(x) = x(x^2-1)$ of the motivating example. Its slope $Dh(x) = 3x^2 - 1$ vanishes at $x = \pm 1/\sqrt3$, the decoupling singularities.}
\label{fig:toy_cubic}
\end{figure}

The system is output controllable in the elementary sense that every $y \in \R$ is attained by some $x$. Feedback linearization nonetheless fails at the two points where $\Lambda(x) = 3x^2 - 1 = 0$, i.e. $x = \pm 1/\sqrt3 \approx \pm 0.577$: the slope of $h$ is negative between them and positive outside, so $\Lambda$ vanishes and the linearizing law \eqref{eq:fl} diverges as $x \to \pm 1/\sqrt3$. This is the simplest instance of a decoupling singularity, and we return to it numerically in Section~\ref{sec:sims}.

\section{Background: continuous Newton and homotopy continuation}\label{sec:wellbehaved}

The controller rests on two classical ideas --- the continuous-time Newton flow and arc-length homotopy continuation --- which we recall in a form tailored to tracking. Throughout this section and Sections~\ref{sec:augmented}--\ref{sec:tracking} we study the \emph{kinematic} (velocity-controlled) problem
\begin{equation}\label{eq:kinematic}
\dot x = u, \qquad y = f(x), \qquad f:\R^n\to\R^n,
\end{equation}
in which the input is the state velocity and, for these sections only, $f$ denotes the \emph{output map} (so $Df$ is its Jacobian). This isolates the singularity mechanism in its simplest form; the general control-affine plant \eqref{eq:sys}, with $Df$ replaced by the decoupling matrix $\Lambda = Dh\,g$, is recovered in Section~\ref{sec:dynamic}. We begin with the regular case $\rank Df(x) = n$.

\subsection{Root-finding}

If the problem is to find a root of this function such that

\begin{equation}
f(x^*) = 0
\end{equation}

then a time derivative of $y = f(x)$ could be found:

\begin{equation}
\dot y = D f(x) \cdot \dot x
\end{equation}

and then the trajectory of $y$ could be designed such that $y(t) \to 0$ for example:

\begin{equation}
\dot y = -y
\end{equation}

or solving for $\dot x$:

\begin{equation}
\dot x = \Bigl[ D f(x) \Bigr]^{-1} \cdot \dot y = -\Bigl[ D f(x) \Bigr]^{-1} y
\end{equation}

The above is just well-known Newton's method written in continuous time.

\subsection{Trajectory tracking}

Generalizing the above approach, it is possible to follow an arbitrary smooth trajectory $y^*(t)$ by manipulating $\dot x$ through the invertible Jacobian matrix $D f(x)$ as 

\begin{equation}\label{newton_track}
\quad \dot x = \Bigl[ D f(x) \Bigr]^{-1} \cdot \dot y^*
\end{equation}

assuming $y(0) = y^*(0)$ without loss of generality (the initial point $x_0$ solving $f(x_0) = y^*(0)$ is found by the root-finding flow above).

Both the root-finding flow and the tracking law \eqref{newton_track} invert $Df$ directly, so they diverge wherever $Df$ loses rank. Homotopy continuation removes this limitation by embedding the problem in a one-parameter family, which we recall next.

\subsection{Newton homotopy and arc-length continuation}

Homotopy continuation deforms a system with a known solution into the target system and tracks the solution along the deformation; the solution curve may pass through singularities at which a naive Newton step breaks down \cite{AllgowerGeorg1990,SommeseWampler2005}. Consider the equation
\begin{equation}
f(x) = y^*
\end{equation}
with $x \in \R^n$ and $f:\R^n\to\R^n$.

A particular form of Newton homotopy for a problem $f(x) = y^*$ can be written as:

\begin{align}\label{hom_func}
H(x,\lambda) &= (1 - \lambda) \cdot \Bigl[ f(x) - f(x_0) \Bigr] + \lambda \cdot \Bigl[ f(x) - y^* \Bigr] \\
&= f(x) - f(x_0) + \lambda \cdot \Bigl[ f(x_0) - y^* \Bigr] 
\end{align}

where $x(0) = x_0$, and $\lambda \in \R$, $\lambda(0) = 0$, $\lambda \in [0,1]$.

The system $H = 0$ depends on a parameter $\lambda$ that evolves in time and the solution $x$ varies correspondingly to trace a curve $(x, \lambda)$ defined by $H(x, \lambda) = 0$.
For initial moment of time $t = 0$, $\lambda(0) = 0$ and $x(0) = x_0$ and $H(x_0,0) = 0$.

Typical assumption for the family of homotopy continuation methods is that the Jacobian of $H$ is a full rank matrix

\begin{equation}\label{full_rank_H}
\rank D H = n
\end{equation}

This is weaker than $\rank Df(x) = n$: the augmented Jacobian $DH$ is $n\times(n+1)$, so the extra degree of freedom $\lambda$ (column $n+1$) can restore full row rank even where $Df$ drops rank.

Stepping $\lambda$ directly from $0$ to $1$ can fail if the solution curve folds in $\lambda$. Arc-length (pseudo-arclength) continuation instead introduces a parameter $s$, the arc length of the curve, which may be identified with time $t$ when the curve is traversed at unit speed, and normalizes the tangent to the solution curve.

In arc‐length form, the solution curve

\begin{equation}
c(t) \;=\; (x(t),\,\lambda(t))
\end{equation}

lives in $\mathbb{R}^{n+1}$. The $c(t)$ is a solution of initial value problem defined by the following conditions (see \cite[Sec.~2.1]{AllgowerGeorg1990}):

\begin{enumerate}
\item[(i)] $\dot H\bigl(c(t)\bigr)\cdot \dot c(t) = 0$ (path-following),
\item[(ii)] $\|\dot c(t)\| = 1$ (normalization),
\item[(iii)] $\det\begin{pmatrix} DH \\ \dot c^T \end{pmatrix} > 0$ (orientation).
\end{enumerate}

By construction $H(x_0,0) = 0$; imposing $\dot H = 0$ along $(x(t),\lambda(t))$ therefore keeps $H \equiv 0$.

Direct calculation of time derivative $\dot H$ of \eqref{hom_func} gives:

\begin{align}\label{diff_hom_func}
\dot H(x,\lambda) &= D f(x) \cdot \dot x + \bigl[ f(x_0) - y^* \bigr] \cdot \dot \lambda \\
&= \bigl[ D f(x) \; \big| \; f(x_0) - y^* \bigr] \cdot \begin{pmatrix} \dot x \\ \dot \lambda \end{pmatrix}
\end{align}

\subsection{The regular case}

Suppose $Df(x)$ is invertible, $\rank Df(x) = n$. By Sard's theorem this holds for almost every $x$, the critical set $\{\rank Df < n\}$ having measure-zero image. Multiplying \eqref{diff_hom_func} by $[Df(x)]^{-1}$ solves for $\dot x$ in terms of $\dot\lambda$, and imposing the normalization $\|(\dot x,\dot\lambda)\| = 1$ with the positive orientation gives the arc-length flow
\begin{gather}
\dot x = \alpha(x)\,[Df(x)]^{-1}\bigl[y^* - f(x_0)\bigr], \qquad \dot\lambda = \alpha(x),\\
\alpha(x) = \Bigl(\, \bigl\|[Df(x)]^{-1}\bigl(y^*-f(x_0)\bigr)\bigr\|^2 + 1 \,\Bigr)^{-1/2} \in (0,1].
\end{gather}
The induced output dynamics, from $\dot y = Df(x)\,\dot x$, is
\begin{equation}
\dot y = \alpha(x)\,\bigl[y^* - f(x_0)\bigr],
\end{equation}
so $y$ moves along the straight segment from $f(x_0)$ to $y^*$ at the positive rate $\alpha$. This is the well-defined behaviour the augmentation must preserve when $Df$ loses rank, which we analyze next.

\section{The augmented Jacobian at a singularity}\label{sec:augmented}

The differentiated homotopy \eqref{diff_hom_func} already exhibits the object the continuation flow must invert: the \emph{augmented Jacobian}
\begin{equation}\label{eq:augmented}
A = \bigl[\, D f(x) \;\big|\; b \,\bigr] \in \R^{n\times(n+1)}, \qquad b = f(x_0) - y^*,
\end{equation}
with homotopy direction $b$ (which becomes time-dependent, $b(t) = f(x_0) - y^*(t)$, once we track a reference in Section~\ref{sec:tracking}). Its rank structure governs the entire method. We now make precise \emph{when} the augmentation repairs a rank drop of $Df(x)$, and what happens to the extra ($\lambda$- or $\mu$-) direction at a singular point.

\begin{definition}[Transversal singularity]\label{def:transversal}
A point $x$ with $\rank D f(x) = n-1$ is a \emph{transversal singularity} of the homotopy if the homotopy direction satisfies
\begin{equation}\label{transversality}
w^T b \ne 0,
\end{equation}
where $w \ne 0$ spans the (one-dimensional) left null space of $D f(x)$, i.e. $w^T D f(x) = 0$. Equivalently, $b \notin \range D f(x)$.
\end{definition}

The vector $w$ is the direction in output space along which control is instantaneously lost (the uncontrollable mode); condition \eqref{transversality} asks the homotopy direction $b$ to have a nonzero component along that mode.

\begin{lemma}[Rank and kernel of the augmented Jacobian]\label{lem:kernel}
Let $A = \bigl[ D f(x) \; \big| \; b \bigr] \in \R^{n\times(n+1)}$.
\begin{enumerate}
\item[(a)] If $\rank D f(x) = n$, then $\rank A = n$ and
\[
\ker A = \spanop\Bigl\{ \bigl( -[Df(x)]^{-1} b, \; 1 \bigr) \Bigr\}.
\]
The unit kernel vector oriented to a nonnegative last component has last ($\mu$-)component
\[
\rho \;=\; \bigl( 1 + \| [Df(x)]^{-1} b \|^2 \bigr)^{-1/2} \in (0,1].
\]
\item[(b)] If $\rank D f(x) = n-1$ and $x$ is a transversal singularity, then $\rank A = n$ (so the continuation tangent stays well defined) and
\[
\ker A = \ker D f(x) \times \{0\}.
\]
Consequently the last component of every kernel vector vanishes: $\rho = 0$.
\end{enumerate}
\end{lemma}

\begin{proof}
(a) With $Df(x)$ invertible, $A$ has $n$ independent columns already, so $\rank A = n$. Solving $Df(x)\,v + b\,w_0 = 0$ gives $v = -[Df(x)]^{-1} b\, w_0$, a one-dimensional space spanned by $(-[Df(x)]^{-1}b, 1)$; normalizing gives $\rho$.

(b) Take $(v, w_0) \in \ker A$, so $Df(x)\,v + b\,w_0 = 0$. If $w_0 \ne 0$ then $b = -Df(x)\,v / w_0 \in \range Df(x)$, contradicting transversality (recall $w^T Df(x) = 0$ implies $w^T \range Df(x) = 0$, whereas $w^T b \ne 0$). Hence $w_0 = 0$ and $Df(x)\,v = 0$, i.e. $v \in \ker Df(x)$. Since $\dim \ker Df(x) = 1$, we get $\ker A = \ker Df(x) \times \{0\}$. Finally $\range A = \range Df(x) + \spanop(b)$ has dimension $(n-1) + 1 = n$ because $b \notin \range Df(x)$, so $\rank A = n$.
\end{proof}

\begin{figure}[h]
\centering
% Transversality / augmented-Jacobian TikZ (only the pictures).
% Wrapped in a figure environment (with caption/label) in the main file.
% Requires: tikz (main preamble) and the operators \range,\rank.
% ---------- panel (a): transversal ----------
\begin{tikzpicture}[scale=1.8,line join=round,font=\footnotesize]
  % range Df (line through x)
  \draw[gray!70,thick] (-1.368,-0.728) -- (1.368,0.728);
  \node[gray!55!black,rotate=28,anchor=south] at (-1.10,-0.585) {$\range Df$};
  % the point x
  \fill (0,0) circle (0.9pt);
  \node[anchor=north east] at (0.1,-0.05) {$x$};
  % left-null direction w (normal to range Df)
  \draw[red!75!black,thick,dashed,->] (0,0) -- (-0.563,1.059);
  \node[red!75!black,anchor=east] at (-0.60,0.62) {$w\ (w^{T}Df=0)$};
  % decomposition of b: projection on range Df + w-component
  \draw[gray,dotted] (0,0) -- (1.0155,0.5394);
  \draw[red!75!black,thick,dotted] (1.0155,0.5394) -- (0.570,1.378);
  \draw[gray,thin] (1.095,0.582) -- (1.053,0.661) -- (0.973,0.619); % right angle
  \node[red!75!black,anchor=west] at (0.86,0.95) {$|w^{T}b|>0$};
  % homotopy direction b (transversal: green)
  \draw[green!55!black,very thick,->] (0,0) -- (0.570,1.378);
  \node[blue!45!black,anchor=west] at (0.60,1.32) {$b=f(x_0)-y^{*}$};
  % title + consequence
  \node[anchor=south] at (0.1,1.62) {(a) transversal: $w^{T}b\neq 0$};
  %\node[draw=green!45!black,fill=green!8,rounded corners,align=center,anchor=north]
  %  at (0.15,-0.95) {$b\notin\range Df\ \Rightarrow\ \rank[\,Df\mid b\,]=m$\\
  %                   bounded crossing, $\rho\to0^{+}$ (Lem.~\ref{lem:kernel}(b), Thm.~\ref{thm:crossing})};
\end{tikzpicture}\hfill
% ---------- panel (b): non-transversal ----------
\begin{tikzpicture}[scale=1.8,line join=round,font=\footnotesize]
  \draw[gray!70,thick] (-1.368,-0.728) -- (1.368,0.728);
  \node[gray!55!black,rotate=28,anchor=south] at (-1.10,-0.585) {$\range Df$};
  \fill (0,0) circle (0.9pt);
  \node[anchor=north east] at (0.1,-0.05) {$x$};
  \draw[red!75!black,thick,dashed,->] (0,0) -- (-0.563,1.059);
  \node[red!75!black,anchor=east] at (-0.60,0.62) {$w\ (w^{T}Df=0)$};
  % b lies in range Df: no w-component
  \draw[red,very thick,->] (0,0) -- (1.280,0.681);
  \node[red!75!black,anchor=south] at (0.58,0.50) {$w^{T}b=0$};
  \node[blue!45!black,anchor=north] at (0.95,0.24) {$b=f(x_0)-y^{*}$};
  \node[anchor=south] at (0.1,1.62) {(b) non-transversal: $w^{T}b=0$};
  %\node[draw=red!70!black,fill=red!8,rounded corners,align=center,anchor=north]
  %  at (0.15,-0.95) {$b\in\range Df\ \Rightarrow\ \rank[\,Df\mid b\,]=m-1$\\
  %                   $A^{+}$ undefined, $\|\dot x\|\to\infty$ (Prop.~\ref{prop:failure})};
\end{tikzpicture}
\caption{The augmented Jacobian $A=[\,Df\mid b\,]$ at a rank-$(m{-}1)$ point, drawn in output space ($m=2$). At the singular point $x$, $\range Df$ is a line and the left-null vector $w$ is its normal ($w^TDf=0$). \textbf{(a)}~The singularity is \emph{transversal} when the homotopy direction $b=f(x_0)-y^*$ is not parallel to $\range Df$---it has nonzero projections on both $w$ and $\range Df$, so $w^Tb\neq0$; the extra column then reaches off the range, $\rank A=m$, and the crossing is bounded with $\rho\to0^{+}$ (Lemma~\ref{lem:kernel}(b), Theorem~\ref{thm:crossing}). \textbf{(b)}~When $b$ lies in $\range Df$ ($w^Tb=0$) the augmentation adds nothing, $\rank A=m-1$, and $A^{+}$ together with the continuation tangent diverge (Proposition~\ref{prop:failure}).}
\label{fig:transversality}
\end{figure}
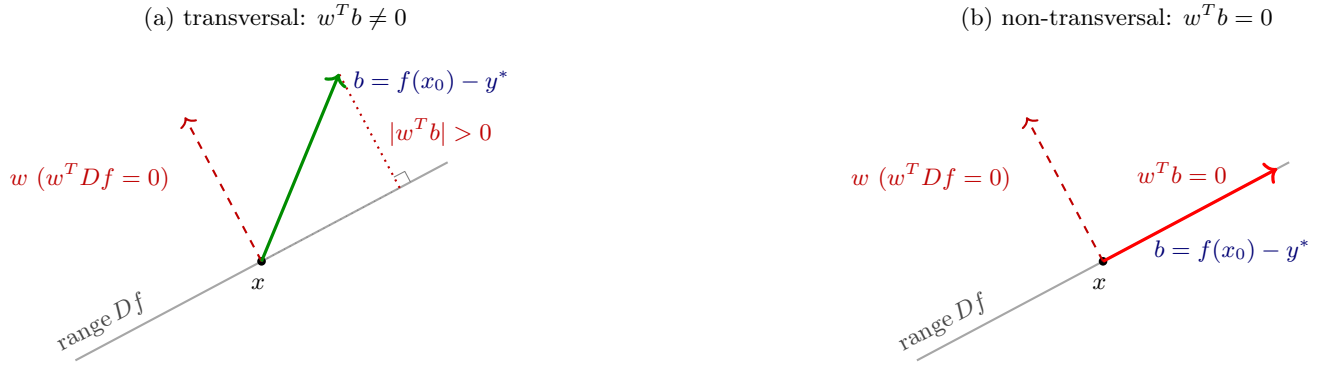

Lemma~\ref{lem:kernel} is the crux. Part (b) says the augmentation \emph{does} keep the tangent well defined through a rank-$(n-1)$ point (Fig.~\ref{fig:transversality}), but at the price that the kernel direction — the only direction through which we can steer the continuation parameter — rotates until its parameter component $\rho$ collapses to zero exactly at the singularity. Section~\ref{sec:crossing} turns this into quantitative statements about tracking.

\section{Nonlinear trajectory tracking}\label{sec:tracking}

We now track a time-varying reference $y^*(t)$. Sampling $y^*$ at instants $t_k = k\,\Delta T$ turns tracking into a sequence of set-point problems, to each of which the Newton homotopy of Section~\ref{sec:wellbehaved} applies. Replacing the per-step parameter $\lambda$ by a continuously increasing $\mu$ through $\lambda = \mu - t + 1$ (so that $\mu$ is synchronized with time, $\mu(t_k) = t_k$, whenever a sample is reached) and letting $\Delta T \to 0$ yields the continuous-time \emph{tracking homotopy}
\begin{equation}\label{hom_tracking}
H(x,\mu, t) = f(x) - f(x_0) + (\mu - t + 1)\,\bigl[\, f(x_0) - y^*(t) \,\bigr],
\end{equation}
anchored at the initial state $x_0 = x(0)$. If $f(x_0)$ already lies on the reference then $\mu(0)=0$; otherwise $\mu(0) = -1$ absorbs the initial transient. Its time derivative is
\begin{equation}\label{diff_hom_tracking}
\dot H = D f(x)\,\dot x + (\dot\mu - 1)\bigl[ f(x_0) - y^*(t) \bigr] - (\mu - t + 1)\,\dot y^*(t).
\end{equation}
Exact tracking, $f(x) = y^*(t)$ together with $Df(x)\,\dot x = \dot y^*(t)$, requires
\begin{equation}\label{tracking_conditions}
\mu = t, \qquad \dot\mu = 1,
\end{equation}
so for a regular output map the parameter $\mu$ is simply time. In general $\mu$ is allowed to lag $t$ in order to negotiate singularities of $f$; the lag $e = t-\mu$ measures the tracking error, as made precise in Section~\ref{sec:accuracy}.

\subsection{Problem solution}

We now solve the tracking homotopy \eqref{hom_tracking}, following the augmented-continuation approach of \cite{Borisevich2013,Borisevich2015} extended to the time-driven setting.

Below is a block-structure viewpoint of the Jacobian $D H(x,\mu,t)$ for the particular homotopy function \eqref{hom_tracking} obtained by splitting the Jacobian into two main blocks:  

- A: the block corresponding to partial derivatives w.r.t. $(x,\mu)$, which is $n \times (n+1)$ matrix.  

- B: the block (column) corresponding to partial derivatives w.r.t. $t$, which is $n \times 1$ column-vector.

In symbols:

\[
\underbrace{A}_{n\times(n+1)}
\;=\;
\begin{pmatrix}
\dfrac{\partial H}{\partial x} & \dfrac{\partial H}{\partial \mu}
\end{pmatrix} = D_{x,\mu} H,
\quad
\underbrace{B}_{n\times1}
\;=\;
\dfrac{\partial H}{\partial t}.
\]

Hence,
\[
DH(x,\mu,t) 
\;=\;
\bigl[A \;\big|\; B\bigr]
\;\in\;\mathbb{R}^{n\times(n+2)}.
\]

We are going to impose a condition similar to the \eqref{full_rank_H} for invertibility of the block $A$ of the Jacobian

\begin{equation}\label{full_rank_A}
\rank A = \rank D_{x,\mu} H = n
\end{equation}

As for the \eqref{full_rank_H} this is indeed a weaker condition compared with $\rank D f = n$ as it could be that the $\rank D_{x,\mu} H = n$ even if the $\rank D f = n - 1$.

Geometrically or numerically \eqref{full_rank_A} implies that, at the point in question, one can (locally) solve

\[
H(x,\mu,t)\;=\;0
\quad\text{for $\,(x,\mu)$\, given a small change in $t$.}
\]

Because if $A$ has rank $n$, it is possible to invert the linear map from $(\delta x,\,\delta\mu)\mapsto A\,(\delta x,\,\delta\mu)$ at least in a one-dimensional sense (there is exactly one missing dimension corresponding to the fact that $A$ is $n\times(n+1)$). 

Concretely:

1. Dimensionality: we have $n$ equations in $\{H=0\}$ but $n+2$ unknowns $(x,\mu,t)$.
  
2. Rank Condition: if $\mathrm{rank}(A)=n$, we can (locally) treat $(x,\mu)$ as implicit functions of $t$.  

If that rank condition fails, we are at singular or critical point with respect to $(x,\mu)$, and special handling may be required to continue solutions in $(x,\mu,t)$-space.

The $B = \partial H/\partial t$ is the $n\times1$ partial derivative of the homotopy with respect to the time $t$.

In a continuation viewpoint, if one wants to track $\bigl(x(t),\,\mu(t)\bigr)$ as $t$ changes, the condition

\[
H\bigl(x(t),\mu(t),t\bigr)\;=\;0
\]

plus the chain rule leads to

\[
\dot H = A \cdot \dot c + B = 0.
\]

where

\[
\dot c(t) = \begin{pmatrix} \dot{x}(t) \\ \dot{\mu}(t) \end{pmatrix}
\]
  
If $\mathrm{rank}(A)=n$, one can solve for $(\dot{x}(t),\,\dot{\mu}(t))$ substituting $\dot{t} = 1$.  

Hence, in differential-algebraic terms

\[
\underbrace{\begin{pmatrix}
\frac{\partial H}{\partial x} & \frac{\partial H}{\partial \mu}
\end{pmatrix}}_{=A}
\begin{pmatrix}
\dot{x}(t)\\[4pt]
\dot{\mu}(t)
\end{pmatrix}
\;+\;
\underbrace{\frac{\partial H}{\partial t}}_{=B}\
\;=\;
0
\]

that gives a solution

\begin{equation}\label{solution_drift}
\dot c_p = -A^\dagger\,B\,
\end{equation}

where $A^\dagger$ is a right-inverse of $A$ (possible because $\rank A=n$).

The geometric interpretation of the result above is:

1. Constraint $H(x,\lambda,t)=0$ defines an $(n+1)$-dimensional manifold in $\mathbb{R}^{n+2}$ (assuming it’s regular). 
 
2. By setting $\dot{t}=1$ a particular slice in the parameter $(x,\lambda)$-space for each increment of $t$ is chosen.  

3. The drift term $\partial H/\partial t$ represents how $H$ changes with time $t$ alone. In order to remain on $H=0$, the $(x,\lambda)$ must counteract that drift. That’s precisely what solving $A \cdot \dot c = -B$ does. Hence, once $A$ has a left inverse, we can cancel out or neutralize the effect of $B$.

It is important to understand that the solution $\dot c_p$ per \eqref{solution_drift} is not unique as the $A$ is a non-square $ n \times (n+1) $ matrix.
Since $ A $ has full row rank by the assumption \eqref{full_rank_A} then the null space of $ A $ has dimension $ (n+1) - n = 1 $, meaning there is a single linearly independent vector in the null space. 

In this case, the general solution can be written as 

\begin{equation}\label{solution}
\dot c = -A^\dagger\,B + \gamma \cdot \ker(A) = \dot c_p + \gamma \cdot \dot c_h
\end{equation}

where:

- $\dot c_p$ is a particular solution of the $A \cdot \dot c = -B$ per \eqref{solution_drift}.

- $\dot c_h$ is a basis vector for the null space (kernel) of $A$ which can be obtained as the last column of the orthogonal matrix $Q$ from the QR decomposition of $A^T$.

- $\gamma$ is a free scalar parameter.

In order to select a unique solution for $\dot c$ we have to impose an additional constraint. Following the arc‐length continuation methods, the additional conditions to uniquely determine the curve $x(t),\,\mu(t))$ are normalization and orientation conditions:

\begin{equation}\label{conditions}
\begin{gathered}
\| \dot{c_h} \| = 1, \\
\det \begin{pmatrix} A \\ \dot{c_h}^T \end{pmatrix} > 0
\end{gathered}
\end{equation}

The solution \eqref{solution} together with conditions \eqref{conditions} determine unique solution trajectory of the homotopy equation $H(x(t),\mu(t),t)=0$ assuming $H(x(0),\mu(0),0)=0$. This result however doesn't guarantee tracking of a trajectory $y^*(t)$ for a particular homotopy function \eqref{hom_tracking}. The next section solves this.

\subsection{Accuracy of tracking}\label{sec:accuracy}
 
Generally, the conditinos \eqref{tracking_conditions} may not be satisfied everywhere since the reference trajectory $y^*(t)$ is not exactly tracked. Particularly in two cases: if the $f(x_0)$ doesn't lie on the $y^*(t)$ and there is an initial transient to reach the $y^*(t)$, or if the trajectory tracking requires overcoming singular points of $f(x)$.

It is expected to minimize a following type of criteria to ensure accuracy of the trajectory tracking
  
\begin{equation}\label{objective}
\int_0^T \| f(x(t)) - y^*(t) \|^2 dt \to \min
\end{equation}

where $\| v \|^2 = v^T v$.

Minimizing the error $e = t - \mu$ implies minimizing \eqref{objective}.

\begin{equation}
\int_0^T \bigl[ t - \mu(t) \bigr]^2 dt \to \min  \implies  \int_0^T \| f(x(t)) - y^*(t) \|^2 dt \to \min
\end{equation}

The \eqref{hom_tracking} can be rearranged as

\begin{equation}
H(x,\mu, t) = f(x) - y^*(t) + (\mu - t) \cdot \Bigl[ f(x_0) - y^*(t) \Bigr] 
\end{equation}

So satisfying $H = 0$ means

\begin{equation}
f(x) - y^*(t) = (t - \mu) \cdot \Bigl[ f(x_0) - y^*(t) \Bigr] 
\end{equation}

thus

\begin{equation}
\int_0^T \| f(x) - y^*(t) \|^2 dt = \int_0^T \bigl[ t - \mu(t) \bigr]^2 \cdot \| f(x_0) - y^*(t) \|^2 dt
\end{equation}

If $y^*(t)$ is bounded around $f(x_0)$ as it lies inside of a ball of radius $M$ centered at $f(x_0)$, i.e. $\|f(x_0) - y^*(t)\| < M$ then

\begin{gather}
\int_0^T \| f(x) - y^*(t) \|^2 dt = \int_0^T \bigl[ t - \mu(t) \bigr]^2 \cdot \| f(x_0) - y^*(t) \|^2 dt \\
< M^2 \int_0^T \bigl[ t - \mu(t) \bigr]^2 dt
\end{gather}

The solution \eqref{solution} has an additional degree of freedom $\gamma$ that can be used to force the trajectory $\mu(t)$ converge to $t$. 

If 

\begin{equation}
e(t) = t - \mu(t)
\end{equation}

and

\begin{equation}
\gamma = k \cdot e(t)
\end{equation}

then for a sufficiently large scalar constant $k > 0$ the error $e(t)$ is bounded $|e(t)| < \epsilon$.

We define the block matrices $A$ and $B$ for the homotopy function \eqref{diff_hom_tracking} as

\begin{equation}\label{AB_blocks}
\begin{gathered}
A = \begin{pmatrix}
D f(x) & f(x_0) - y^*(t)
\end{pmatrix} \\
B = y^*(t) - f(x_0) - (\mu - t + 1) \cdot \dot y^*(t)
\end{gathered}
\end{equation}

Consider the behaviour away from singularities, where $f(x)$ is regular in the sense $\rank D f(x)=n$ and $y^*(t)$ can be tracked exactly.
In this case a particular form of the right inverse $A^\dagger$ of the matrix $A$ is simply:

\begin{equation}
A^\dagger = \begin{pmatrix}
\Bigl[ D f(x) \Bigr]^{-1} \\ 0
\end{pmatrix}
\end{equation}

Thus the dynamics of $\mu(t)$ is fully driven by the homogeneous part solution:

\begin{equation}
\begin{pmatrix} \dot{x}(t) \\ \dot{\mu}(t) \end{pmatrix} = \gamma \cdot \ker(A)
\end{equation}

also subject to additional constraints \eqref{conditions}.

The kernel of the matrix $A$ has the following closed form as the $D f(x)$ is invertible:

\begin{equation}
\ker(A) = \begin{pmatrix}
\Bigl[ - D f(x) \Bigr]^{-1} \cdot (f(x_0) - y^*(t)) \\ 1
\end{pmatrix}
\end{equation}

So as expected, the solution trajectory by the \eqref{solution} for \eqref{diff_hom_tracking} with unity scaling $\gamma = 1$ produced $\dot \mu = 1$.
Now if $\gamma = k \cdot [t - \mu(t)]$ the final dynamic is

\begin{equation}
\dot \mu = k \cdot [t - \mu(t)]
\end{equation}

the solution with the initial condition is

\begin{equation}
\mu(t) = t - \frac{1}{k} + \Bigl(\mu(0)+\frac{1}{k}\Bigr) \cdot e^{-kt} \to t - \frac{1}{k}
\end{equation}

Thus for $t \to \infty$ the $\mu(t)$ follows $t$ up to to $-1/k$ and the asymptotic tracking error is

\begin{equation}
|e(t)| = |t - \mu(t)| \le 1/k
\end{equation}

\begin{remark}[Removing the $O(1/k)$ bias]\label{rem:exactlocking}
The residual $1/k$ is a property of the proportional choice $\gamma = k\,e$, not of the method: holding the steady-state value $\gamma_{ss} = (1-d)/\rho$ (where $d$ is the parameter component of the particular solution and $\rho$ the parameter component of the unit kernel, both computed by the controller at every step) requires $e \ne 0$. Two standard compensators remove the bias (Figure~\ref{fig:lockingarch}). \emph{Drift feedforward}, $\gamma = \operatorname{sat}_M\!\bigl((1-d)/\rho\bigr) + k e$, cancels the drift exactly: in the regular region the locking dynamics becomes $\dot e = -k\rho e$ and $e \to 0$ exponentially with no residual; the saturation (gated off as $\rho \to 0$) is mandatory, since $(1-d)/\rho$ diverges at the singularity --- the authority--boundedness tradeoff of Remark~\ref{rem:tradeoff} reasserting itself, and no compensator can do better there. Under a bounded mismatch $\Delta$ in the computed drift the residual degrades to $\Delta/(k\rho_0)$, mismatch level rather than drift level. \emph{Integral action}, $\gamma = \operatorname{sat}\bigl(k_p e + k_i \int e\bigr)$ with conditional-integration (clamping) anti-windup --- the integrator is frozen whenever the unclamped compensator output is saturated in the direction the error pushes, the industrial standard --- removes constant residuals with no model of $d$ at all; for constant drift and frozen authority the locking error converges to zero exponentially, by a strict quadratic Lyapunov function with cross term for the two-dimensional locking loop. In practice the two are best combined: PI feedback with the saturated drift feedforward added inside the same clamped regulator, so the feedforward carries the bulk of $\gamma_{ss}$, the integrator trims only the feedforward mismatch, and the conditional-integration logic protects both terms through the fold transits with a single saturation. Both statements are machine-checked (Section~\ref{sec:lean}), and both compensators are compared numerically in Section~\ref{sec:lockingsims}. At the singular instant itself all locking laws coincide in effect, since $\rho \to 0$ disables the feedback path; the guarantees of Theorem~\ref{thm:crossing} are unaffected.
\end{remark}

\subsection{Crossing singularities and recovery of tracking}\label{sec:crossing}

The regular-region analysis above breaks down precisely where the method is supposed to earn its keep: at points where $\rank D f(x) = n-1$. There the closed form of the kernel used above involves $[Df(x)]^{-1}$ and diverges. To analyze the crossing we use instead the quantities that stay finite through a transversal singularity, as identified in Lemma~\ref{lem:kernel}.

With $\dot t = 1$, write the continuation flow in the form
\begin{equation}\label{flow}
\dot c = -A^\dagger B + k\, e\, \hat n, \qquad e = t - \mu, \quad k > 0,
\end{equation}
where $A^\dagger = A^T (A A^T)^{-1}$ is the minimum-norm right inverse (well defined whenever $\rank A = n$), and $\hat n$ is the unit vector spanning $\ker A$, oriented so that its last component $\rho = \hat n_{n+1} \ge 0$.

\begin{definition}[Continuation control law]\label{def:controllaw}
The controller carries an internal scalar state $\mu$ and, at each time $t$ (plant state $x$, locking error $e=t-\mu$, gain $k>0$), computes:
\begin{enumerate}
\item[(i)] homotopy data: $b=f(x_0)-y^*(t)$, augmented matrix $A=[\,Df(x)\mid b\,]\in\R^{n\times(n+1)}$, and drift $B=-b-\dot y^*(t)(1-e)$;
\item[(ii)] least-norm particular solution $-A^\dagger B$, with $A^\dagger=A^T(AA^T)^{-1}$;
\item[(iii)] unit kernel $\hat n$: any vector with $A\hat n=0$, $\|\hat n\|=1$ (e.g.\ the last right-singular vector of $A$, or the normalized cross product of the two rows of $A$ when $n=2$), its \emph{sign} fixed by the orientation rule below;
\item[(iv)] flow \eqref{flow}: $\dot c=(\dot x,\dot\mu)=-A^\dagger B+k\,e\,\hat n$; the first $n$ components are the commanded velocity $\dot x$, and $\mu$ is integrated by $\dot\mu$.
\end{enumerate}
Because $\ker A$ is one-dimensional, $\hat n$ is unique up to sign; that sign --- the \emph{orientation} --- is the controller's single behavioural degree of freedom:
\begin{description}
\item[\normalfont(R)\ $\rho\ge0$ (reflecting):] fix the sign by $\rho:=\hat n_{n+1}\ge0$. Away from singularities this coincides with the continuous tangent, but where $\rho$ would change sign (at a fold) it overrides the arc-length orientation \eqref{conditions} and flips $\hat n$, so the tangent is \emph{discontinuous} there --- the reflection, realized as a Filippov switch. It keeps $\dot\mu$ oriented forward and the locking feedback dissipative (Theorem~\ref{thm:reflection}(i)); it is the default and the orientation assumed in the analysis of this section.
\item[\normalfont(C)\ continuity (crossing):] take $\hat n(t)$ to be the \emph{continuous} unit tangent field along the solution curve --- equivalently, fix its sign by the arc-length orientation \eqref{conditions}, $\det\!\begin{pmatrix}A\\ \hat n^{T}\end{pmatrix}>0$, held at constant sign as $t$ evolves --- so the tangent never reverses and $\rho$ passes smoothly through zero; at a fold the state crosses to the opposite sheet (Theorem~\ref{thm:reflection}(ii)).
\end{description}
For the dynamic plant \eqref{eq:sys} replace $Df\to\Lambda(x)$ and $B\to a(x)-b-\dot y^*(t)(1-e)$, and read $\dot x$ as the control $u$ (Section~\ref{sec:dynamic}).
\end{definition}

Splitting $\dot c = (\dot x, \dot \mu)$ and writing $d = \bigl( -A^\dagger B \bigr)_{n+1}$ for the parameter component of the particular solution, the locking error obeys the scalar equation
\begin{equation}\label{error_ode}
\dot e \;=\; 1 - d(x,t) - k\, \rho(x,t)\, e .
\end{equation}
Here $\rho$ is the feedback authority over $\mu$: by Lemma~\ref{lem:kernel} it equals $(1 + \|[Df]^{-1}b\|^2)^{-1/2} \in (0,1]$ away from singularities and drops to $0$ at a transversal singularity.

We make the following standing assumptions along the trajectory:

\begin{itemize}
\item[(A1)] $f \in C^2$ and $y^* \in C^1$ on a neighborhood of the trajectory, and $b(t), \dot y^*(t)$ are bounded.
\item[(A2)] (Global transversality.) Every instant with $\rank Df(x(t)) = n-1$ is a transversal singularity in the sense of Definition~\ref{def:transversal}. Hence $\rank A(t) = n$ for all $t$, and the flow \eqref{flow} has a finite right-hand side everywhere.
\item[(A3)] (Isolated singularity.) The singular instants are isolated at $t^*$: there exist $\delta > 0$ and $\rho_0 > 0$ with $\rho(x(t),t) \ge \rho_0$ for all $t \notin I_\delta = (t^*-\delta, t^*+\delta)$.
\end{itemize}

Under (A1)--(A2) the map $(AA^T)^{-1}$ is bounded along the trajectory (its smallest singular value is bounded below by transversality), so
\[
\bar D \;=\; \sup_t \bigl| 1 - d(x(t),t) \bigr| \;<\; \infty .
\]

\begin{theorem}[Bounded crossing and recovery of tracking]\label{thm:crossing}
Under assumptions (A1)--(A3), the continuation flow \eqref{flow} satisfies:
\begin{enumerate}
\item[1.] (No blow-up.) The right-hand side of \eqref{flow} is finite for all $t$, including at $t^*$; in particular $\|\dot x(t)\| \le \|A^\dagger B\| + k|e|$ stays bounded through the singularity, and near $t^*$ the injected motion $k e\, \hat n$ is directed along $\ker D f(x)$.
\item[2.] (Contraction away from the singularity.) On any interval where $\rho \ge \rho_0$,
\[
|e(t)| \;\le\; |e(t_1)|\, e^{-k \rho_0 (t - t_1)} \;+\; \frac{\bar D}{k \rho_0}.
\]
In particular the output error obeys $\| f(x) - y^* \| = |e|\, \|b\| = O(1/k)$ in the regular regime.
\item[3.] (Bounded excursion during the crossing.) For all $t \in \overline{I_\delta}$,
\[
|e(t)| \;\le\; |e(t^*-\delta)| \;+\; 2\delta\, \bar D ,
\]
so the transient growth of the tracking error while feedback authority is lost is at most $2\delta\bar D$ in magnitude, and the output error stays bounded by $\|b\|_\infty\bigl(|e(t^*-\delta)| + 2\delta\bar D\bigr)$.
\item[4.] (Re-locking.) For $t \ge t^* + \delta$, item~2 applies from the value $e(t^*+\delta)$; hence $e(t)$ re-converges exponentially to the $O(\bar D/(k\rho_0))$ band. The flow returns to $\mu \approx t$ and reference tracking is recovered after the crossing.
\end{enumerate}
\end{theorem}

\begin{proof}
\emph{Item 1.} By Lemma~\ref{lem:kernel}(b) and (A2), $\rank A = n$ at $t^*$, so $A^\dagger = A^T(AA^T)^{-1}$ is bounded and $\hat n$ is a unit vector; thus the right-hand side of \eqref{flow} is finite. As $\rho \to 0$, Lemma~\ref{lem:kernel}(b) gives $\hat n \to (v, 0)$ with $v \in \ker Df(x)$, so the injected term $ke\,\hat n$ aligns with $\ker Df(x) \times \{0\}$. The velocity bound is immediate from \eqref{flow}.

\emph{Item 2.} On the interval, \eqref{error_ode} reads $\dot e = -k\rho\, e + (1-d)$ with $k\rho \ge k\rho_0 > 0$ and $|1-d| \le \bar D$. The comparison principle for the scalar linear equation with nonnegative time-varying rate gives
\[
|e(t)| \le |e(t_1)| e^{-\int_{t_1}^t k\rho\, d\tau} + \int_{t_1}^t e^{-\int_\tau^t k\rho}\,\bar D\, d\tau \le |e(t_1)| e^{-k\rho_0 (t-t_1)} + \frac{\bar D}{k\rho_0}.
\]
The output identity $f(x) - y^* = e\, b$ (from $H = 0$) yields the norm bound.

\emph{Item 3.} Since the orientation is fixed so that $\rho \ge 0$, the feedback term is dissipative. Using \eqref{error_ode},
\[
\frac{d}{dt}\,\tfrac{1}{2} e^2 \;=\; e\,\dot e \;=\; e(1-d) - k\rho\, e^2 \;\le\; |e|\,\bar D ,
\]
hence $\frac{d}{dt}|e| \le \bar D$ wherever $e \ne 0$. Integrating over $\overline{I_\delta}$ (length $2\delta$) gives $|e(t)| \le |e(t^*-\delta)| + 2\delta\bar D$. The output bound follows from $f(x)-y^* = e\, b$.

\emph{Item 4.} Apply item~2 with $t_1 = t^*+\delta$.
\end{proof}

The theorem answers the practical question directly: as long as every singularity met along the way is \emph{transversal}, the parameter $\mu$ is allowed to fall behind time $t$ by a controlled amount during the crossing (at most $2\delta \bar D$), and it provably re-locks to $t$ afterwards, restoring tracking. The condition that decides success is transversality, and its failure is genuinely fatal.

\begin{proposition}[Loss of transversality]\label{prop:failure}
Suppose at some $t^*$ the singularity is non-transversal, i.e. $\rank Df(x) = n-1$ and $w^T b(t^*) = 0$. Then $\rank A(t^*) = n-1 < n$: the matrix $A A^T$ is singular, $A^\dagger$ and the continuation tangent are undefined, and the solution set of $H = 0$ has a fold or bifurcation in $(x,\mu,t)$. Forward-time integration with $\dot t = 1$ generically forces $\|\dot x\| \to \infty$ (the naive continuous-Newton blow-up reappears), so no finite-velocity crossing exists and tracking cannot be recovered by the flow alone.
\end{proposition}

\begin{proof}
If $w^T b = 0$ then $b \in \range Df(x)$, so $\range A = \range Df(x)$ has dimension $n-1$ and $A$ is row-rank deficient; $AA^T$ is singular. The constraint $A\dot c = -B$ is then solvable only if $B \in \range A = \range Df(x)$. When $B \notin \range Df(x)$ there is no finite tangent compatible with $\dot t = 1$, forcing $\|\dot c\| \to \infty$; when $B \in \range Df(x)$ the solution set gains a dimension, which is a bifurcation of the continuation curve.
\end{proof}

\begin{remark}[Reciprocity / crossing criterion]
The transversality condition \eqref{transversality}, $w^T b \ne 0$, is the continuation-flow counterpart of the dynamic reciprocity criterion for crossing Type~2 singularities of parallel mechanisms: the driving direction must have a nonzero component along the uncontrollable mode $w$. Here the driving direction is simply the homotopy direction $b = f(x_0) - y^*$, and the criterion emerges from a purely kinematic (rank) argument rather than from the inverse dynamic model.
\end{remark}

\begin{remark}[Authority--boundedness tradeoff]\label{rem:tradeoff}
The exact residual $1/k$ obtained in Section~\ref{sec:accuracy} corresponds to the \emph{scaled} kernel $(-[Df]^{-1}b, 1)$, whose parameter component is normalized to $1$: this preserves full feedback authority ($\dot e = 1 - ke$) but its state part $-[Df]^{-1}b$ diverges at the singularity. The \emph{unit} kernel used in Theorem~\ref{thm:crossing} keeps $\dot x$ bounded across the singularity at the cost of vanishing authority ($\rho \to 0$) exactly at $t^*$. One cannot have both: bounded state velocity and undiminished parameter authority are incompatible at a rank-$(n-1)$ point. This is the precise sense in which $\mu$ is ``poorly controlled'' near the singularity, and item~3 of the theorem bounds the resulting excursion.
\end{remark}

\begin{remark}[Why $n \ge 2$ is needed to see the effect]
In the scalar toy example ($n = 1$) a rank-$(n-1)$ point means $Df = 0$, the left null vector is $w = 1$, and transversality reads $f(x_0) - y^* \ne 0$, which holds whenever the target differs from the anchor. The null space $\ker Df = \R$ is the whole line, so there is no separate uncontrollable submanifold and the ``wandering'' motion of item~1 is invisible. The genuine difficulty — bounded but off-task motion along $\ker Df$ during the crossing — first appears for $n \ge 2$, e.g. the 2-DOF arm, where $\ker Df$ is a proper subspace.
\end{remark}

\section{Reachability and the reflection--crossing dichotomy}\label{sec:reflection}

Remark~\ref{rem:tradeoff} identifies a genuine choice at the singularity: the one-dimensional kernel of $A$ must serve two ends at once --- advancing the homotopy parameter $\mu$ and steering the state through $\ker Df$ --- and these compete. The choice is made by the \emph{orientation} of the unit kernel vector $\hat n$, and the two admissible orientations produce qualitatively different lifts: \emph{reflection} (the state turns back onto the sheet it arrived on) or \emph{crossing} (it passes to the opposite sheet). This section resolves three questions in order: (a) at a generic fold, what does each orientation do (Theorem~\ref{thm:reflection}); (b) \emph{which} orientation should be used, decided by a \emph{reachability} condition on the reference (Proposition~\ref{prop:reach}); and (c) how large is the transient excursion the state incurs while it negotiates the fold (Proposition~\ref{prop:excursion}). Throughout we take the kinematic case $\dot x = u$, $y = \varphi(x)$, $\varphi:\R^n\to\R^n$ (so $Df = D\varphi$), where the fold structure is cleanest; the dynamic case follows via the substitution $D\varphi \to \Lambda$ of Section~\ref{sec:dynamic}.

\paragraph{The reachability question.} Near a fold the two sheets $S_\pm$ both map onto the same reachable half-space of outputs (Proposition~\ref{prop:foldnf}); the reference is either \emph{feasible} (a preimage exists nearby) or, having crossed the fold's critical value, \emph{infeasible} (no nearby preimage). The engineering choice is then sharp: if the piece of reference beyond the fold is infeasible near the encountered fold, no controller reaches it and the state should ride the boundary and return (reflection); if instead that piece remains reachable but its continuous preimage lies on the opposite sheet, the state must pass through (crossing). We make this precise below, after fixing the fold model.

Before localizing we state precisely which singularities the analysis covers.

\begin{definition}[Fold singularity]\label{def:fold}
Let $\varphi:\R^n\to\R^n$ be $C^\infty$ and let $x_*$ be a critical point with a \emph{corank-one} rank drop, $\rank D\varphi(x_*)=n-1$; let $r$ span the kernel $\ker D\varphi(x_*)$ and $w$ the left null space, $w^TD\varphi(x_*)=0$. The point $x_*$ is a \emph{fold} if it is Whitney-nondegenerate,
\begin{equation}\label{eq:foldnondeg}
w^{T}\,D^{2}\varphi(x_*)(r,r)\ \ne\ 0,
\end{equation}
i.e.\ the second derivative of the lost output component $w^{T}\varphi$ along the uncontrollable direction $r$ does not vanish.
\end{definition}

\begin{proposition}[Whitney fold normal form and its scope]\label{prop:foldnf}
At a fold (Definition~\ref{def:fold}) there exist local coordinates $x\mapsto(\xi,z)\in\R\times\R^{n-1}$ about $x_*$ and $y\mapsto(Y,\zeta)$ about $\varphi(x_*)$ in which
\begin{equation}\label{eq:whitneynf}
Y=y_c+c\,\xi^{2},\qquad \zeta=z,\qquad c\ne0,
\end{equation}
so that, transverse to the critical set $\{\xi=0\}$, a single output coordinate folds quadratically while the remaining $n-1$ stay regular \cite{GolubitskyGuillemin1973}. Folds are the \emph{generic} corank-one singularities of maps $\R^n\to\R^n$: those maps whose singularities are all folds (and, in codimension two, cusps) form an open dense set, and a generic reference --- one outside a thin, perturbation-removable exceptional set --- crosses the fold hypersurface transversally, at isolated instants.
\end{proposition}

\begin{remark}[Scope]
The condition applies verbatim with $D\varphi$ replaced by the decoupling matrix $\Lambda$ (Section~\ref{sec:dynamic}); the arm at full extension and the cubic $x^3-x$ at $\pm1/\sqrt3$ are folds. Degeneracies where \eqref{eq:foldnondeg} fails (the cusp $Y=\xi^3$, or corank-$\ge2$ rank drops) are non-generic and lie beyond the present fold analysis; their treatment is left to future work (Section~\ref{sec:future}).
\end{remark}

Figure~\ref{fig:foldcover} depicts the two orientations: reflection on the incoming sheet $S_-$ under $\rho\ge0$, and crossing to the opposite sheet $S_+$ under the continuity orientation.

\begin{figure}[t]
\centering
% Fold cover TikZ (only the pictures). Wrapped in a figure environment
% (with caption/label) in the main file. Requires: tikz (main preamble).
% ================= (a) reflection =================
\begin{tikzpicture}[x=2.72cm,y=1.7cm,font=\footnotesize,line join=round]
  % background: infeasible strip and critical-value line
  \fill[black!6] (-0.6,-1.9) rectangle (0,1.72);
  \draw[black!30,dotted] (0,-1.9) -- (0,1.72);
  \node[black!45,anchor=south] at (0.1,1.4) {$y=y_c$};
  \node[black!55,align=center] at (-0.33,1.05) {infeasible\\[-1pt]$y<y_c$};
  % the cover (parabola y = xi^2, drawn as (xi^2, xi))
  \draw[gray,thick] plot[variable=\t,domain=-1.2:1.2,samples=90] ({\t*\t},{\t});
  % reflection: two lines parallel to the incoming branch S_-  (in and back)
  \draw[blue!70!black,thick] plot[variable=\t,domain=-0.95:-0.06,samples=40] ({\t*\t+0.05},{\t});
  \draw[blue!70!black,thick] plot[variable=\t,domain=-0.95:-0.06,samples=40] ({\t*\t-0.05},{\t});
  \draw[blue!70!black,->] (0.3525,-0.55) -- (0.2709,-0.47);   % in, toward Sigma_s
  \draw[blue!70!black,->] (0.1709,-0.47) -- (0.2525,-0.55);   % back, away
  \draw[blue!70!black,thick] (0.0536,-0.06) to[out=110,in=70] (-0.0464,-0.06); % U-turn
  % one-sided Filippov fields near Sigma_s (both toward xi=0)
  \draw[red!75!black,->] (0.14,0.30) -- (0.14,0.08);
  \draw[red!75!black,->] (0.14,-0.30) -- (0.14,-0.08);
  \node[red!75!black,anchor=west] at (0.15,0.24) {$\dot\xi|_{0^{+}}<0$};
  \node[red!75!black,anchor=west] at (0.15,-0.24) {$\dot\xi|_{0^{-}}>0$};
  % fold point + labels
  \fill (0,0) circle (1.3pt);
  \node[anchor=east] at (-0.07,0) {$\Sigma_s$ (fold)};
  \node[blue!70!black,anchor=east] at (-0.07,-0.52) {reflect};
  \node[blue!70!black] at (0.98,-0.52) {$S_-$ (incoming)};
  \node[black!45] at (0.98,0.52) {$S_+$ (not entered)};
  % reference: in (v<0) and back
  \draw[black!55,->] (1.05,-1.34) -- (-0.25,-1.34);
  \draw[black!55,->] (-0.25,-1.52) -- (1.05,-1.52);
  \node[black!55,anchor=south] at (0.4,-1.32) {in ($v<0$)};
  \node[black!55,anchor=north] at (0.4,-1.54) {back};
  \node[black!55,anchor=east] at (-0.3,-1.43) {$y^{*}(t)$};
  % axis labels
  \node[anchor=north] at (0.5,-1.98) {output $y=y_c+c\,\xi^{2}$};
  \node[rotate=90,anchor=south] at (-0.76,-0.1) {argument $\xi$};
  \node[anchor=south] at (0.45,1.78) {(a) $\rho\ge0$: reflection (Filippov sliding)};
\end{tikzpicture}\hfill
% ================= (b) crossing =================
\begin{tikzpicture}[x=2.72cm,y=1.7cm,font=\footnotesize,line join=round]
  \fill[black!6] (-0.6,-1.9) rectangle (0,1.72);
  \draw[black!30,dotted] (0,-1.9) -- (0,1.72);
  \node[black!45,anchor=south] at (0.1,1.4) {$y=y_c$};
  \node[black!55,align=center] at (-0.33,1.05) {infeasible\\[-1pt]$y<y_c$};
  % cover = trajectory (green)
  \draw[green!45!black,very thick] plot[variable=\t,domain=-1.2:1.2,samples=90] ({\t*\t},{\t});
  % kernel direction r
  \draw[red!75!black,<->] (0,-0.32) -- (0,0.32);
  \node[red!75!black,anchor=south west] at (-0.45,0.33) {$\ker D\varphi=r$};
  % crossing arrows: in on S_-, out on S_+
  \draw[green!45!black,->] (0.3844,-0.62) -- (0.2304,-0.48);
  \draw[green!45!black,->] (0.2304,0.48) -- (0.3844,0.62);
  \fill (0,0) circle (1.3pt);
  \node[anchor=east] at (-0.07,0) {$\Sigma_s$ (fold)};
  \node[green!45!black,anchor=east] at (-0.07,-0.52) {cross};
  \node[green!45!black] at (0.98,-0.52) {$S_-$ (incoming)};
  \node[green!45!black] at (0.98,0.52) {$S_+$ (exit)};
  \draw[black!55,->] (1.05,-1.34) -- (-0.25,-1.34);
  \draw[black!55,->] (-0.25,-1.52) -- (1.05,-1.52);
  \node[black!55,anchor=south] at (0.4,-1.32) {in ($v<0$)};
  \node[black!55,anchor=north] at (0.4,-1.54) {back};
  \node[black!55,anchor=east] at (-0.3,-1.43) {$y^{*}(t)$};
  \node[anchor=north] at (0.5,-1.98) {output $y=y_c+c\,\xi^{2}$};
  \node[rotate=90,anchor=south] at (-0.76,-0.1) {argument $\xi$};
  \node[anchor=south] at (0.45,1.78) {(b) continuity: crossing ($\mathbb{Z}/2$ deck $\tau:\xi\mapsto-\xi$)};
\end{tikzpicture}
\caption{Fold cover of a min-type fold in Whitney normal form $y=y_c+c\,\xi^2$
(Proposition~\ref{prop:foldnf}): horizontal axis is the output $y$, vertical
axis the fold coordinate $\xi$. The two sheets $S_-=\{\xi<0\}$ and
$S_+=\{\xi>0\}$ meet on the switching surface $\Sigma_s=\{\xi=0\}$ (the fold),
and only $y\ge y_c$ is reachable; the reference dips into the infeasible region
$y<y_c$ ($v<0$) and returns. \textbf{(a)} Under the $\rho\ge0$ orientation the
state slides in and back along the incoming sheet $S_-$ (blue): the one-sided
fields $\dot\xi|_{0^\pm}$ both point toward $\Sigma_s$, so it is attracting and
the state reflects (Lemma~\ref{lem:filippov}, Theorem~\ref{thm:reflection}(i)).
\textbf{(b)} Under the continuity orientation the state crosses $S_-\to S_+$
through the fold, the sheet exchange $\tau:\xi\mapsto-\xi$ generating the
$\mathbb{Z}/2$ deck transformation of $\xi\mapsto\xi^2$
(Theorem~\ref{thm:reflection}(ii,iii)).}
\label{fig:foldcover}
\end{figure}
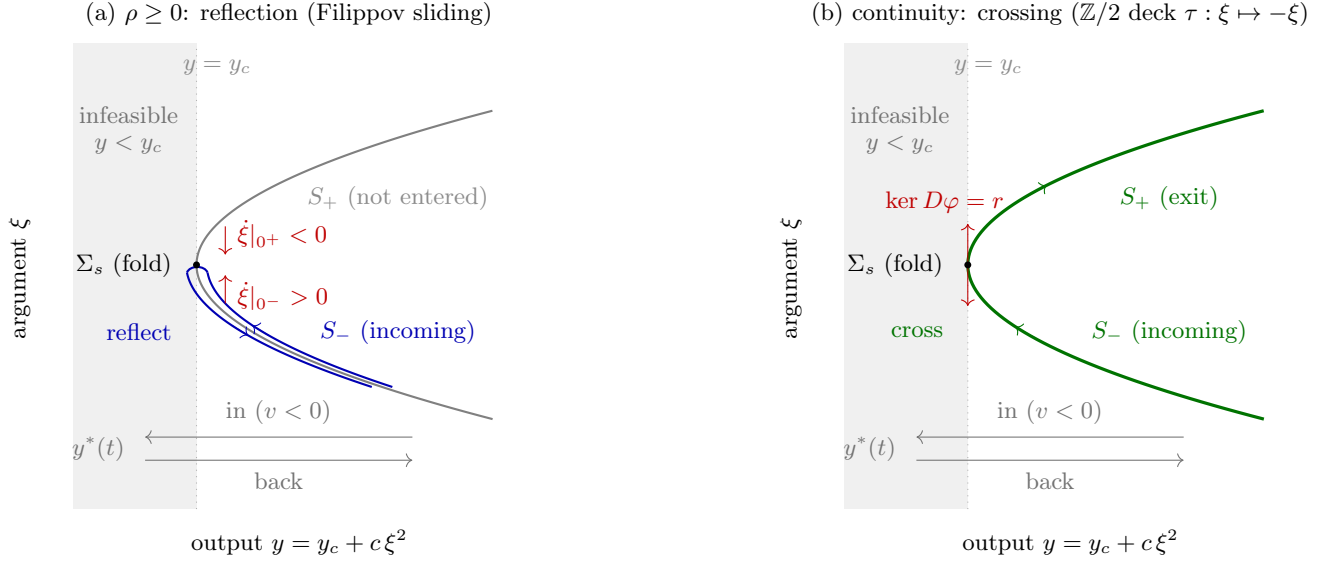

\paragraph{Local model.} By transversality the augmented map is regular and the non-fold coordinates decouple, so the question is one-dimensional transverse to the fold. By Proposition~\ref{prop:foldnf} the output map then takes the Whitney normal form $y = y_c + c\,\xi^2$ (we fix $c>0$, a min-type fold), with critical set $\{\xi=0\}$ and sheets $S_-=\{\xi<0\}$, $S_+=\{\xi>0\}$; this is the local picture of the cubic $x^3-x$ at $\pm 1/\sqrt3$. The reference crosses transversally, $y^*(t)=y_c+v\,(t-t^*)$, and it suffices to analyze the scalar closed loop in $(\xi,e)$, $e=t-\mu$, with $b(t)=\varphi(\xi_0)-y^*(t)$ frozen at $b_*$ near $t^*$. With $g = 2c\xi$ (so the rank drop is $g=0$), $D^2 = g^2 + b^2$, and, in the kinematic case, $B = -b - v(1-e)$, the $\rho\ge0$ orientation gives
\begin{equation}\label{eq:xe_system}
\dot\xi = -\frac{gB}{D^2} - k\,e\,s\,\sign(b), \qquad
\dot e = 1 + \frac{bB}{D^2} - k\,e\,\frac{|g|}{D}, \qquad s=\sign(\xi).
\end{equation}
Since the kernel state component $\hat n_1 = \pm(-b)/D$ has its sign fixed by $\rho = \hat n_2 = |g|/D \ge 0$, and $\rho \propto |g|$ changes character across $\xi=0$, the surface $\{\xi=0\}$ is a switching surface and the closed loop is a Filippov system there.

\begin{proposition}[Sliding sign, leading order]\label{prop:sign}
For the reduced fold model \eqref{eq:xe_system}, the quasi-steady locking error on approach to the singularity ($\dot e = 0$, $g\to0$) satisfies
\begin{equation}\label{eq:esign}
\sign(e) = \sign(g^2 - b\,v) \;\xrightarrow{\,g\to0\,}\; -\sign(b\,v),
\qquad\text{hence}\qquad \sign(e\,b_*) = -\sign(v).
\end{equation}
Consequently the sliding condition $e\,b_* > 0$ holds if and only if $v<0$, i.e. the reference is decreasing through the critical value of the min-type fold --- ``pushing toward the infeasible side''. (For a max-type fold the reachable side and the sign of $v$ both reverse, and the invariant statement is unchanged.)
\end{proposition}

\begin{proof}
Setting $\dot e = 0$ in \eqref{eq:xe_system} away from the fold gives $e = (1 + bB/D^2)\,D/(k|g|)$. Substituting $B \approx -b - v$ (the $ve$ term is $O(1/k)$) and simplifying, $1 + bB/D^2 = (g^2 - bv)/(g^2+b^2)$, whence $\sign(e) = \sign(g^2 - bv)$; letting $g\to0$ yields \eqref{eq:esign}. Moreover on the surface the constraint $H=0$ slaves the error exactly, $e = (y_c - y^*)/b$, which has the same sign, so the conclusion survives the boundary layer.
\end{proof}

\begin{lemma}[Sliding and reflection, Filippov]\label{lem:filippov}
Regard the reduced fold model under the $\rho\ge0$ orientation as a Filippov system with switching surface $\Sigma_s=\{\xi=0\}$, and suppose the reference crosses the critical value $y_c$ of a min-type fold from the reachable side ($v<0$ on approach). Then:
\begin{enumerate}
\item[(a)] the one-sided fields satisfy $\dot\xi|_{0^-} = k\,e\,\sign(b) > 0 > -k\,e\,\sign(b) = \dot\xi|_{0^+}$ on the sliding region $\{e\,b>0\}$ (nonempty by Proposition~\ref{prop:sign}), so $\Sigma_s$ is attracting and reached in finite time --- precisely when $y^*$ reaches $y_c$;
\item[(b)] while $y^*<y_c$ one has $e=(y_c-y^*)/b>0$, so $\dot\xi|_{0^+}<0$ and the state is confined to $\xi\le0$;
\item[(c)] on $\Sigma_s$ the constraint $H=0$ pins $e=(y_c-y^*)/b$ and $y=y_c$: the state rides the singular set;
\item[(d)] sliding terminates when $y^*$ recrosses $y_c$ ($e=0$); for $y^*>y_c$, $\dot\xi|_{0^-}<0$ releases the state into $S_-$. Hence the state returns to the incoming sheet (reflection).
\end{enumerate}
Under the continuity orientation ($\hat n$ aligned with its predecessor) $\Sigma_s$ is not a switching surface, the field is transversal, and the state passes to $S_+$ (crossing).
\end{lemma}

\begin{proof}
Filippov's construction \cite{Filippov1988} gives existence of solutions and, where both one-sided fields point strictly toward $\Sigma_s$, a unique sliding solution. The strict inequalities in (a) hold on $\{eb>0\}$, entered on the reachable-side approach by Proposition~\ref{prop:sign}; finite-time reachability follows because $y^*$ reaches $y_c$ in finite time and, off the singularity, $y=y^*$ forces $\xi\to0$. The value on $\Sigma_s$ in (c) is $H=0$, solvable because transversality keeps $A$ full row rank. Parts (b),(d) are the signs of the one-sided normal velocities under $e>0$, $e<0$.
\end{proof}

\begin{proposition}[Reduction at a generic fold in $\R^n$]\label{prop:reduction}
Let $\varphi:\R^n\to\R^n$ have a generic corank-one fold at $x_*$ (Whitney nondegeneracy $w^TD^2\varphi(r,r)\ne0$, with $r$ spanning $\ker D\varphi$ and $w$ the left null vector). In Whitney coordinates $\varphi(x)=(x_1^2,x_2,\dots,x_n)$, with $A=[D\varphi\mid b]$ and transversality $w^Tb = b_1 \ne 0$, the closed loop splits as a singular block in $(x_1,\mu)$ --- identical to the scalar fold model --- and $n-1$ regular equations $\dot x_i + b_i\dot\mu = -B_i$ ($i\ge2$) that are smooth affine functions of $\dot\mu$. The kernel at $x_1=0$ is $(r,0)$, so the reflection/crossing degree of freedom lies entirely in the $(x_1,\mu)$ block. Proposition~\ref{prop:sign} and Lemma~\ref{lem:filippov} therefore apply verbatim along $r$.
\end{proposition}

\begin{proof}
In the stated coordinates $D\varphi=\operatorname{diag}(2x_1,1,\dots,1)$, so the rows for $i\ge2$ read $\dot x_i + b_i\dot\mu = -B_i$ with unit coefficient on $\dot x_i$, regular in $\dot\mu$; the first row is $2x_1\dot x_1 + b_1\dot\mu = -B_1$, the scalar fold block. At $x_1=0$ the augmented kernel is spanned by $(e_1,0)$, i.e. $r\times\{0\}$, so the orientation degree of freedom acts only on $(x_1,\mu)$.
\end{proof}

\begin{remark}[Non-equivariance of the continuation law]\label{rem:nonequivariance}
The following caveat sharpens how Proposition~\ref{prop:reduction} must be read. The continuation law \eqref{flow} is built from the least-norm inverse $A^\dagger=A^T(AA^T)^{-1}$ and the \emph{unit} kernel vector $\hat n$: both depend on the Euclidean structure of the chosen coordinates and are therefore not natural under diffeomorphism. The Whitney change of coordinates of Proposition~\ref{prop:foldnf} consequently does not conjugate the closed loop into the continuation closed loop of the normal form; it conjugates it into a flow whose particular solution and kernel normalization are distorted by the bounded, invertible Jacobian of the change. The reduction remains valid because everything that decides the dichotomy is invariant: the kernel \emph{direction} $\ker A$ (a line, independent of normalization), transversality, and the sliding conditions of Proposition~\ref{prop:sign} and Lemma~\ref{lem:filippov} are sign conditions, and signs are preserved by diffeomorphism. The \emph{quantitative} constants ($\bar D$, $\rho_0$, push margins), however, transform by the conjugation and are not those of the original coordinates: quantitative bounds should be stated in the original coordinates, as Theorem~\ref{thm:crossing} does, or transported with the conjugation constants.
\end{remark}

\begin{theorem}[Reflection at a generic fold]\label{thm:reflection}
Let $\varphi:\R^n\to\R^n$ have a generic corank-one fold, crossed transversally by the reference from the reachable sheet, and consider the continuation controller \eqref{flow} in the Filippov (ideal, $k\to\infty$) sense. Then
\begin{enumerate}
\item[(i)] the $\rho\ge0$ orientation reflects: the state reaches the singular set, rides it while the reference is infeasible, and returns to the incoming sheet, with bounded control throughout;
\item[(ii)] the continuity orientation crosses to the opposite sheet;
\item[(iii)] the two lifts differ exactly by the sheet-exchange $\tau:\xi\mapsto-\xi$, the generator of the $\mathbb Z/2$ deck transformation of the fold cover $\xi\mapsto\xi^2$.
\end{enumerate}
For finite $k$ the same holds up to an $O(1/k)$ perturbation localized at the singular instant.
\end{theorem}

\begin{proof}
By Proposition~\ref{prop:reduction} the Whitney normal form reduces the singular dynamics to the scalar fold model in $(x_1,\mu)$ along $\ker D\varphi$, the transverse velocities being regular in $\dot\mu$. Proposition~\ref{prop:sign} shows the reachable-side approach enters the sliding region $\{eb>0\}$; Lemma~\ref{lem:filippov} then gives finite-time attraction, confinement, the pinned exit at $y^*=y_c$, and release onto the incoming sheet under $\rho\ge0$ (reflection), and transversal passage under continuity (crossing). Boundedness of the control is Theorem~\ref{thm:crossing}(1). The deck action $\tau$ is the sheet exchange of $\xi\mapsto\xi^2$, applied by continuity and not by $\rho\ge0$.
\end{proof}

\begin{remark}[Uniqueness of sliding; the inclusion branches]\label{rem:leanstrength}
Two parts of Theorem~\ref{thm:reflection} hold in a stronger, machine-checked form (Section~\ref{sec:lean}). First, part~(i) is a uniqueness statement: while the one-sided fields straddle the switching surface, \emph{every} Filippov solution of the convexified closed loop starting on the singular set remains on it --- the state does not merely admit the sliding motion, it is forced into it. Second, part~(iii) is a statement about the inclusion itself: at the release instant the convexified fold model admits precisely two continuations, exit into the incoming sheet or into the opposite sheet, related by the deck exchange $\tau$ and projecting to the same output path. The kernel orientation of Definition~\ref{def:controllaw} is therefore exactly a \emph{selection rule} among branches of the Filippov inclusion, not merely a numerical convention.
\end{remark}

\subsection{Which orientation: a reachability criterion}\label{subsec:reach}

Theorem~\ref{thm:reflection} says what each orientation \emph{does}; it does not say which one should be used. That is decided by whether the reference beyond the fold is reachable, and on which sheet. We keep the reduced min-type model $Y=y_c+c\,\xi^2$ (reachable set $\{Y\ge y_c\}$, sheets $S_\pm$), with the state arriving on the incoming sheet $S_{\mathrm{in}}$ and the reference crossing the critical value transversally at rate $v=\dot Y^*(t^*)$. Recall (Proposition~\ref{prop:sign}) that the sliding/reflection region $\{e\,b>0\}$ is entered exactly when $v<0$, i.e.\ when the reference heads to the \emph{infeasible} side $Y<y_c$.

\begin{proposition}[Reachability and orientation]\label{prop:reach}
Consider the continuation controller at a transversally crossed fold, the reference feasible on approach, in the ideal ($k\to\infty$) sense.
\begin{enumerate}
\item[(i)] \emph{(Infeasible continuation $\Rightarrow$ reflect.)} If the reference continues to the infeasible side of the encountered fold ($v<0$ for the min-type model), then no state realizes it there. The $\rho\ge0$ orientation reflects: it rides $\Sigma_s$ with the output pinned at $y_c$ and output error $|Y^*-y_c|$ (the \emph{fold gap}), returning to $S_{\mathrm{in}}$. Forcing the continuity orientation cannot lower this error --- both sheets lie on the reachable side and neither reaches $Y^*<y_c$ --- and, having no reachable branch to follow, it drives the state off the reachable set with $\rho<0$; the locking equation $\dot e=(1-d)+k|\rho|e$ is then anti-dissipative and the flow diverges (the dynamic form of Proposition~\ref{prop:failure}).
\item[(ii)] \emph{(Reachable continuation on the far sheet $\Rightarrow$ cross.)} If the reference stays feasible through the encounter, both orientations are bounded; reflection keeps $S_{\mathrm{in}}$, crossing moves to the opposite sheet $S_{\mathrm{out}}$. The crossing is \emph{forced} --- reflection saturates and loses tracking --- if and only if the reference's continuation is not reachable on $S_{\mathrm{in}}$ (the incoming sheet's output range terminates before the reference, as at a neighbouring fold); otherwise both sheets track $Y^*$ and the choice is \emph{discretionary}, to be resolved by a secondary criterion, not by the tracking objective.
\end{enumerate}
\end{proposition}

\begin{proof}
(i) The reachable set of the fold is $\{Y\ge y_c\}$ and both sheets map onto it, so $Y^*<y_c$ has no preimage on either sheet; the smallest achievable output error is attained on $\Sigma_s$, where $y=y_c$, giving $|Y^*-y_c|$. By Proposition~\ref{prop:sign}, $v<0$ places the flow in the sliding region $\{eb>0\}$, so Lemma~\ref{lem:filippov} yields reflection under $\rho\ge0$. Under continuity $\Sigma_s$ is not attracting and $\xi$ is driven monotonically; but with the target infeasible on both sheets the constraint $H=0$ cannot be met with $\xi$ bounded, so $\rho<0$ persists and $\dot e=(1-d)+k|\rho|e$ grows without bound. (ii) If a preimage exists on $S_{\mathrm{out}}$, the final clause of Lemma~\ref{lem:filippov} carries the tangent across, giving $y=Y^*$ up to the crossing transient of Proposition~\ref{prop:excursion}; if additionally $Y^*$ leaves the output range of $S_{\mathrm{in}}$ then reflection, confined to $S_{\mathrm{in}}$, cannot realize it and saturates, so crossing is forced; if $Y^*$ is reachable on both sheets, $H=0$ holds on either and the tracking objective is indifferent.
\end{proof}

\begin{proposition}[Excursion during the encounter]\label{prop:excursion}
Let the fold be negotiated on a window $I_\delta=(t^*-\delta,t^*+\delta)$, and recall the output identity $\|y-y^*\|=|e|\,\|b\|$.
\begin{enumerate}
\item[(a)] \emph{(Reflection, additive.)} Under $\rho\ge0$ the feedback is dissipative and Theorem~\ref{thm:crossing}(3) gives $|e(t)|\le|e(t^*-\delta)|+2\delta\bar D$; hence the output excursion is $O(1/k)+O(\delta)$.
\item[(b)] \emph{(Crossing, multiplicative.)} Under continuity let $2\delta_c$ be the sub-window on which $\rho<0$ and $\bar\rho=\max|\rho|$. Then
\begin{equation}\label{eq:excursion}
|e(t)|\ \le\ \bigl(|e(t^*-\delta_c)|+2\delta_c\bar D\bigr)\,e^{\,2k\bar\rho\,\delta_c},
\end{equation}
so the crossing excursion is bounded iff $k\delta_c$ is bounded. For a \emph{crossing-compatible} reference (transversality margin bounded away from zero, so $\xi$ traverses $\Sigma_s$ at the injected speed $\sim k|e|$ and $\delta_c=O(1/k)$) the factor is $e^{O(1)}$ and the excursion is $O(1)$; as the margin $w^Tb\to0$, $\delta_c\to O(1)$ and the bound blows up, recovering the divergence of Proposition~\ref{prop:reach}(i).
\end{enumerate}
\end{proposition}

\begin{proof}
(a) is Theorem~\ref{thm:crossing}(3). (b) On the sub-window where $\rho=-|\rho|<0$ the scalar error obeys $\dot e=(1-d)+k|\rho|e$, a linear equation with nonnegative rate $k|\rho|\le k\bar\rho$ and $|1-d|\le\bar D$; the comparison principle gives \eqref{eq:excursion}. The transit time $\delta_c$ is set by the injected motion: in a neighbourhood $|\xi|<\varepsilon$ of $\Sigma_s$, $\dot\xi$ is dominated by $k\,e\,\hat n_1=O(k|e|)$, so an unobstructed passage onto a reachable $S_{\mathrm{out}}$ takes $\delta_c\sim\varepsilon/(k|e|)=O(1/k)$; when the far sheet is infeasible $\xi$ cannot leave the neighbourhood, $\delta_c=O(1)$, and the factor is $e^{O(k)}$.
\end{proof}

The additive-versus-multiplicative contrast of Proposition~\ref{prop:excursion} is exactly the difference between the gentle reflections and the violent-but-bounded crossings observed numerically (Section~\ref{sec:sims}): reflection costs $O(1/k+\delta)$, a compatible crossing costs $e^{O(1)}$ more, and an incompatible one is unbounded. Figure~\ref{fig:excursion} summarizes the two regimes and the quantities that bound them.

\begin{figure}[htbp]
\centering
% Excursion-bounds TikZ (only the pictures), replicating fig:excursion.
% Adapted surgically from fold_cover_fig.tex. Requires: tikz (main preamble).
% ================= (a) reflection: additive excursion =================
\begin{tikzpicture}[x=2.72cm,y=1.7cm,font=\footnotesize,line join=round]
  % background: infeasible strip and critical-value line
  \fill[black!6] (-0.6,-1.9) rectangle (0,1.72);
  \draw[black!30,dotted] (0,-1.9) -- (0,1.72);
  \node[black!45,anchor=south] at (0,1.52) {$y=y_c$};
  \node[black!55,align=center] at (-0.33,1.05) {infeasible\\[-1pt]$y<y_c$};
  % the cover (parabola y = xi^2, drawn as (xi^2, xi))
  \draw[gray,thick] plot[variable=\t,domain=-1.2:1.2,samples=90] ({\t*\t},{\t});
  % reflection: two lines parallel to the incoming branch S_-  (in and back)
  \draw[blue!70!black,thick] plot[variable=\t,domain=-0.95:-0.06,samples=40] ({\t*\t+0.05},{\t});
  \draw[blue!70!black,thick] plot[variable=\t,domain=-0.95:-0.06,samples=40] ({\t*\t-0.05},{\t});
  \draw[blue!70!black,->] (0.3525,-0.55) -- (0.2709,-0.47);   % in, toward Sigma_s
  \draw[blue!70!black,->] (0.1709,-0.47) -- (0.2525,-0.55);   % back, away
  \draw[blue!70!black,thick] (0.0536,-0.06) to[out=110,in=70] (-0.0464,-0.06); % U-turn
  % one-sided Filippov fields near Sigma_s (both toward xi=0)
  %\draw[red!75!black,->] (0.14,0.30) -- (0.14,0.08);
  %\draw[red!75!black,->] (0.14,-0.30) -- (0.14,-0.08);
  %\node[red!75!black,anchor=west] at (0.15,0.24) {$\dot\xi|_{0^{+}}<0$};
  %\node[red!75!black,anchor=west] at (0.15,-0.24) {$\dot\xi|_{0^{-}}>0$};
  % fold point + labels
  \fill (0,0) circle (1.3pt);
  \node[anchor=south east] at (-0.03,0.04) {$\Sigma_s$};
  % infeasible target y* and the FOLD GAP (on xi=0, from y_c to y*)
  \node[red!75!black] at (-0.42,0) {$\times$};
  \node[red!75!black,anchor=south] at (-0.42,0.05) {$y^{*}$};
  \draw[red!75!black,<->] (-0.02,0) -- (-0.40,0);
  \node[red!75!black,anchor=north] at (-0.21,-0.03) {fold gap};
  \node[blue!70!black,anchor=east] at (-0.07,-0.52) {reflect};
  \node[blue!70!black] at (0.98,-0.52) {$S_-$ (incoming)};
  \node[black!45] at (0.98,0.52) {$S_+$ (not entered)};
  % simplified output law: clamped at the reachable boundary
  \node[blue!70!black,anchor=west,font=\scriptsize] at (0.1,0) {$y(t) \to \max(y^{*},y_c)$};
  % additive-excursion box
  %\node[draw=blue!60!black,fill=blue!8,rounded corners,align=center,text width=4.6cm,font=\scriptsize]
  %     at (0.42,-1.42) {output pinned at $y_c$, error $=|y^{*}-y_c|=|e|\,\|b\|$;\\[1pt]
  %                      $|e|\le|e_0|+2\delta\bar D=O(1/k)+O(\delta)$};
  % axis labels + title
  \node[anchor=north] at (0.5,-1.98) {output $y=y_c+c\,\xi^{2}$};
  \node[rotate=90,anchor=south] at (-0.76,-0.1) {argument $\xi$};
  \node[anchor=south] at (0.45,1.78) {(a) reflection ($\rho\ge0$): additive excursion};
\end{tikzpicture}\hfill
% ================= (b) crossing: multiplicative excursion =================
\begin{tikzpicture}[x=2.72cm,y=1.7cm,font=\footnotesize,line join=round]
  \fill[black!6] (-0.6,-1.9) rectangle (0,1.72);
  \draw[black!30,dotted] (0,-1.9) -- (0,1.72);
  \node[black!45,anchor=south] at (0,1.52) {$y=y_c$};
  \node[black!55,align=center] at (-0.33,1.05) {infeasible\\[-1pt]$y<y_c$};
  % rho<0 transit window (half-width delta_c) near the fold -- drawn before trajectory
  \fill[red!10] (0,-0.30) rectangle (0.52,0.30);
  % cover = trajectory (green)
  \draw[green!45!black,very thick] plot[variable=\t,domain=-1.2:1.2,samples=90] ({\t*\t},{\t});
  % delta_c window: extent along xi and the injected-transit label
  \draw[red!75!black,<->] (0.045,-0.30) -- (0.045,0.30);
  \node[red!75!black,anchor=west] at (0.07,-0.02) {$\rho<0$ for $2\delta_c$};
  \node[red!75!black,anchor=south west] at (-0.7,0.3) {transit $\dot\xi\sim k|e|$};
  % crossing arrows: in on S_-, out on S_+
  \draw[green!45!black,->] (0.3844,-0.62) -- (0.2304,-0.48);
  \draw[green!45!black,->] (0.2304,0.48) -- (0.3844,0.62);
  \fill (0,0) circle (1.3pt);
  \node[anchor=south east] at (-0.03,0.04) {$\Sigma_s$};
  \node[green!45!black,anchor=east] at (-0.07,-0.52) {cross};
  \node[green!45!black] at (0.98,-0.52) {$S_-$ (incoming)};
  \node[green!45!black] at (0.98,0.52) {$S_+$ (exit)};
  % simplified output law: reference tracked on either sheet
  \node[green!45!black,anchor=west,font=\scriptsize] at (0.05,0.95) {$y(t) \to y^{*}(t)$};
  % multiplicative-excursion box
  %\node[draw=green!45!black,fill=green!8,rounded corners,align=center,text width=4.6cm,font=\scriptsize]
  %     at (0.42,-1.42) {$|e|\le(|e_0|+2\delta_c\bar D)\,e^{2k\bar\rho\delta_c}$;\\[1pt]
  %                      bounded iff $k\delta_c=O(1)$ (else diverges, Prop.)};
  \node[anchor=north] at (0.5,-1.98) {output $y=y_c+c\,\xi^{2}$};
  \node[rotate=90,anchor=south] at (-0.76,-0.1) {argument $\xi$};
  \node[anchor=south] at (0.45,1.78) {(b) continuity: crossing --- multiplicative excursion};
\end{tikzpicture}
\caption{The reachability dichotomy and the excursion bounds at a min-type fold
$y=y_c+c\,\xi^2$ (output $y$ horizontal, fold coordinate $\xi$ vertical; the two
sheets $S_\pm=\{\pm\xi>0\}$ meet on $\Sigma_s=\{\xi=0\}$, and only $y\ge y_c$ is
reachable). \textbf{(a) Reflection} ($\rho\ge0$, Proposition~\ref{prop:reach}(i)):
the reference dips to the infeasible side $y<y_c$ (shaded), where no state exists;
the flow rides $\Sigma_s$ with the output pinned at $y_c$ and returns on the
incoming sheet $S_-$. The tracking error is then the \emph{fold gap}
$|y^*-y_c|=|e|\,\|b\|$ (red double-arrow), and the locking error obeys the
\emph{additive} bound $|e|\le|e_0|+2\delta\bar D=O(1/k)+O(\delta)$
(Proposition~\ref{prop:excursion}(a)): the excursion grows at most linearly in the
time spent at the singularity and shrinks with the gain. \textbf{(b) Crossing}
(continuity, Proposition~\ref{prop:reach}(ii)): the destination is reachable on the
far sheet $S_+$, and the state passes through the fold. On the sub-window of
half-width $\delta_c$ around $\Sigma_s$ the parameter authority reverses sign
($\rho<0$, shaded), and the locking error obeys the \emph{multiplicative} bound
$|e|\le(|e_0|+2\delta_c\bar D)\,e^{2k\bar\rho\delta_c}$
(Proposition~\ref{prop:excursion}(b)); the state traverses this window at the
injected speed $\dot\xi\sim k|e|$, so $\delta_c=O(1/k)$ and the factor is
$e^{O(1)}$ for a crossing-compatible reference, but diverges as the transversality
margin collapses ($k\delta_c\to\infty$), recovering the failure of
Proposition~\ref{prop:reach}(i).}
\label{fig:excursion}
\end{figure}
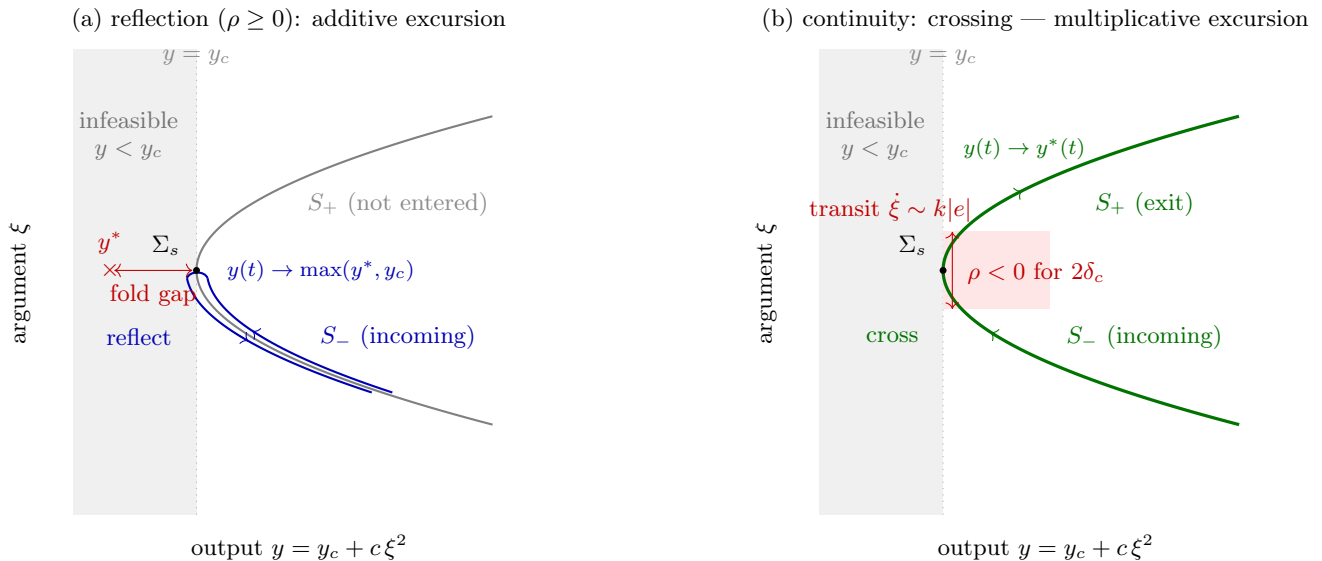

\subsection{Choosing the orientation in practice}\label{sec:orientation}

Propositions~\ref{prop:reach}--\ref{prop:excursion} turn the orientation into a decision rule rather than a stylistic choice; because $\ker A$ is one-dimensional, the single degree of freedom is spent either on advancing/locking $\mu$ or on steering through $\ker Df$, and the two uses are mutually exclusive (Remark~\ref{rem:tradeoff}). We summarize.

\emph{Default: $\rho\ge0$ (reflection).} The greedy, minimal-reconfiguration policy keeps $\mu$ monotone and the feedback dissipative, so tracking stays bounded and re-locks (Theorem~\ref{thm:crossing}). Whenever the incoming sheet continues to serve the reference --- generic in the redundant case, where both assemblies realize the same $y^*(t)$ --- it is preferable: additive excursion (Proposition~\ref{prop:excursion}(a)) and no reconfiguration.

\emph{Forced crossing: the branch dies.} If the reference drives the incoming sheet to the boundary of its reachable range (a neighbouring fold), maintaining $\rho\ge0$ would require $\dot\mu<0$; the flow stalls or diverges (Proposition~\ref{prop:reach}(i)/\ref{prop:failure}). A crossing is then necessary, and by Proposition~\ref{prop:excursion}(b) it is bounded only along a crossing-compatible reference. A practical detector follows from the error dynamics: if $|e|$ grows despite the $\rho\ge0$ orientation, the branch is dying and a crossing is being forced.

\emph{Discretionary crossing: a secondary criterion.} When both sheets track $y^*$ (Proposition~\ref{prop:reach}(ii), reachable-on-both case), the tracking objective is indifferent, and any preference for a working-mode change must come from outside it --- joint limits, collision avoidance, manipulability, actuator effort, or a commanded posture. This is the price of the greedy default: $\rho\ge0$ is locally but not globally optimal.

These cases suggest a two-level controller: the inner continuation flow \eqref{flow} realizes tracking with the $\rho\ge0$ guarantees, while a supervisor monitors branch viability (through the growth of $|e|$) and any secondary criterion, overriding the orientation to cross --- ideally after shaping the reference to be crossing-compatible --- only when Proposition~\ref{prop:reach} dictates. Formalizing this supervisor, and a second-degree-of-freedom augmentation that would let locking and branch selection proceed simultaneously, is left for future work.

\paragraph{Verification on scalar folds.} As a bridge from the theory to the experiments, we verify Propositions~\ref{prop:reach}--\ref{prop:excursion} directly on scalar folds, where each regime is isolated: the min-type normal form $y=x^2$ (reachable $y\ge0$, genuinely infeasible below) for reflection, and the cubic $y=x^3-x$ (surjective, so outputs beyond a fold live on the far sheet) for crossing. Figure~\ref{fig:excursion-data} collects the measurements. The reflection run pins the output at $y_c$ over the infeasible arc with error equal to the fold gap ($0.500$, matching $\max|y^*-y_c|$ to three digits) and a locking error that stays additive; forcing the continuity orientation into that same infeasible arc diverges ($x\to-49.7$), the failure of Proposition~\ref{prop:reach}(i). The crossing run passes through both folds with a locking-error transient confined to the $\rho<0$ window, and the gain sweep confirms the multiplicative bound of Proposition~\ref{prop:excursion}(b) holds at every $k$ (peak $|e|=4.19\le38.0$ at $k=100$), up to the point where $k\delta_c$ grows enough that the bound---and the run---diverge. In plain terms, the two constants appearing in these bounds have a simple reading: $\delta$ (and its crossing counterpart $\delta_c$) is \emph{how long the trajectory lingers near the singularity} --- the width of the window over which feedback authority is weak --- while $\bar D$ bounds \emph{how fast the locking error can drift} during that window, once the stabilizing feedback has gone limp. Reflection multiplies the two additively (drift rate $\times$ time adrift), so its excursion is mild; crossing lets the drift compound while $\rho<0$, so its excursion grows exponentially in the same product $k\delta_c$.

\begin{figure}[htbp]
\centering
% this figure is generated by fig_excursion_data.py
\includegraphics[width=0.98\textwidth]{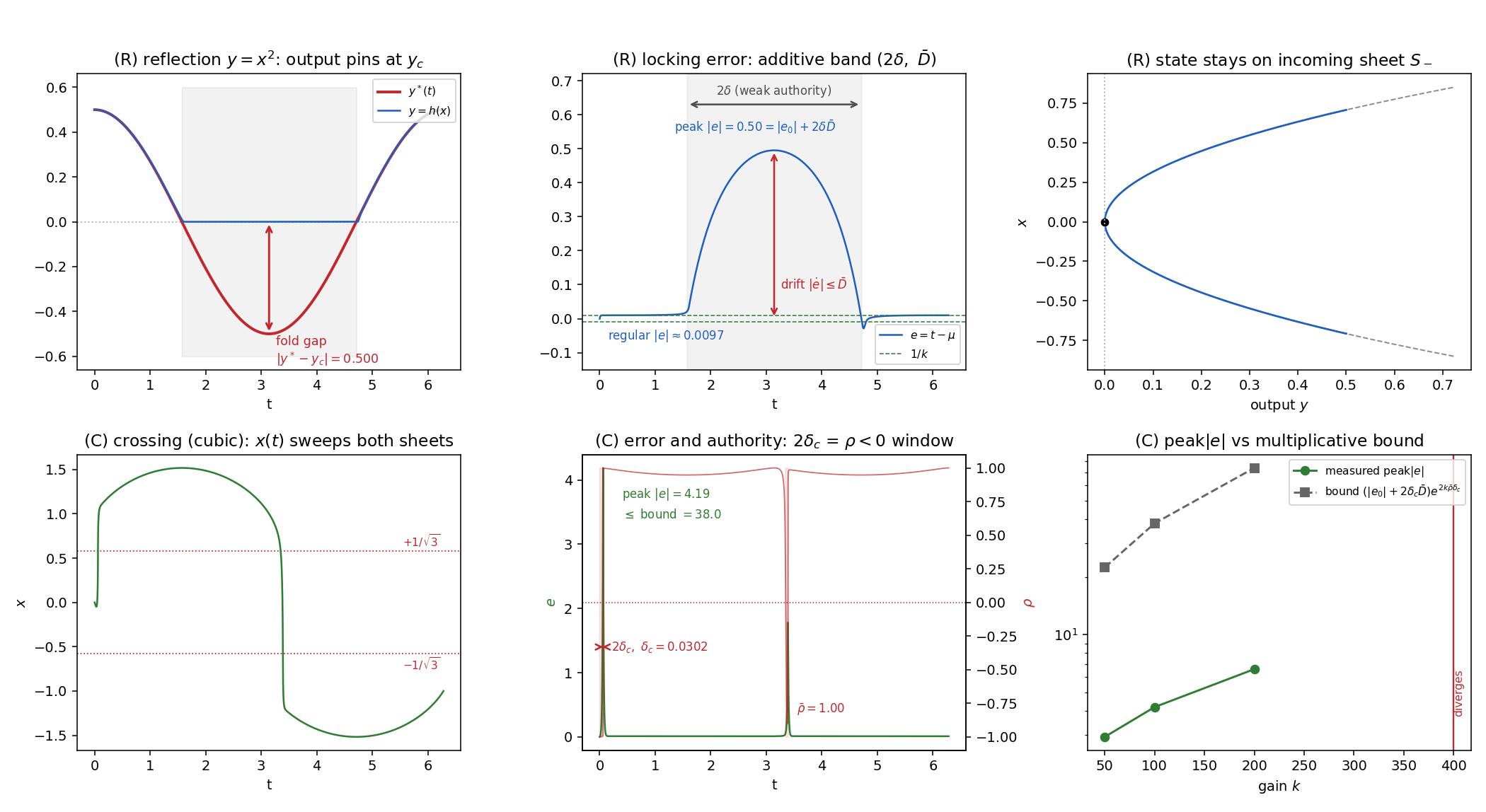}
\caption{Numerical verification of the reachability and excursion propositions on
scalar folds (all quantities measured from the continuation flow, $k=100$ unless
swept). \textbf{Top row --- reflection}, min-type fold $y=x^2$ under $\rho\ge0$,
reference $y^*=0.5\cos t$ dipping to the infeasible side $y<y_c=0$ (shaded):
\emph{(R1)} the output $y$ (blue) pins at $y_c$ while $y^*$ (red) is infeasible, the
red double-arrow marking the fold gap $|y^*-y_c|=0.500$; \emph{(R2)} the locking
error $e=t-\mu$ stays in the $\pm1/k$ band away from the fold ($|e|\approx9.7\times
10^{-3}$) and its excursion is additive (peak $0.50$), never explosive;
\emph{(R3)} the phase view $(y,x)$ shows the state confined to the incoming sheet
$S_-$ (reflection). \textbf{Bottom row --- crossing}, cubic $y=x^3-x$ under the
continuity orientation, $y^*=2\sin t$: \emph{(C1)} the state $x(t)$ sweeps through
both folds $\pm1/\sqrt3$ (dotted), a genuine branch change; \emph{(C2)} the locking
error $e$ (green) and the parameter authority $\rho$ (red, right axis) with the
$\rho<0$ transit window shaded --- here $\bar\rho=1$, $\delta_c$ gives $k\delta_c=3.0$,
and the measured peak $|e|=4.19$ respects the multiplicative bound
$(|e_0|+2\delta_c\bar D)e^{2k\bar\rho\delta_c}=38.0$; \emph{(C3)} the gain sweep
$k\in\{50,100,200,400\}$ plots the measured peak $|e|$ against that bound (log
scale): the bound holds throughout, and at $k=400$ the growth of $k\delta_c$ drives
both the bound and the run to divergence, the finite-$k$ realization of
Proposition~\ref{prop:reach}(i).}
\label{fig:excursion-data}
\end{figure}

\section{Generalization to control-affine systems of relative degree one}\label{sec:dynamic}

Sections~\ref{sec:wellbehaved}--\ref{sec:tracking} treated the output map as an object whose argument is directly assignable, $\dot x = u$: this is the kinematic, velocity-controlled case. We now return to the dynamic control-affine system of the introduction,
\begin{equation}\label{affine_sys}
\dot x = f(x) + g(x)\,u, \qquad y = h(x), \qquad x\in\R^n,\; u,y\in\R^m,
\end{equation}
of relative degree one, and show that the entire development --- Lemma~\ref{lem:kernel}, the transversality condition of Definition~\ref{def:transversal}, Theorem~\ref{thm:crossing}, Proposition~\ref{prop:failure}, and the reflection/crossing dichotomy --- carries over verbatim once the output-map Jacobian $Df$ is replaced by the \emph{decoupling matrix} of the system.

\subsection{Output-level homotopy and the decoupling matrix}

For a relative-degree-one system the input enters the first derivative of the output,
\begin{equation}\label{output_deriv}
\dot y = D h(x)\,[f(x)+g(x)u] = a(x) + \Lambda(x)\,u,
\end{equation}
with the decoupling matrix and output drift
\begin{equation}
\Lambda(x) := D h(x)\,g(x) \in \R^{m\times m}, \qquad a(x) := D h(x)\,f(x) \in \R^m .
\end{equation}
Feedback linearization inverts $\Lambda$; the \emph{control singularity} --- the loss of relative degree at which linearization fails --- is precisely $\rank \Lambda(x) < m$. This is the dynamic counterpart of the Jacobian singularity $\rank D f < n$ of the kinematic problem.

Impose the same homotopy, now on the output,
\begin{equation}\label{hom_dynamic}
H(x,\mu,t) = h(x) - y^*(t) + (\mu - t)\,b(t), \qquad b(t) := h(x_0) - y^*(t),
\end{equation}
so that $H = 0$ reproduces the output-error identity $h(x) - y^* = (t-\mu)\,b = e\,b$, with $e = t-\mu$. Differentiating along the closed loop and using \eqref{output_deriv},
\begin{equation}
\dot H = \Lambda(x)\,u + b\,\dot\mu + B, \qquad B := a(x) - b - \dot y^*(1-e),
\end{equation}
which is the linear relation $A\,c + B = 0$ in the unknowns $c = (u,\dot\mu)$, with the augmented matrix
\begin{equation}\label{A_dynamic}
A(x,t) = \bigl[\, \Lambda(x) \;\big|\; b(t) \,\bigr] \in \R^{m\times(m+1)} .
\end{equation}
This is structurally identical to the kinematic augmented system: the sole change is that the affine term $B$ now carries the drift $a(x) = D h\,f$, while $A = [\Lambda \mid b]$ has exactly the block form analyzed in Lemma~\ref{lem:kernel}. The controller is therefore
\begin{equation}\label{flow_dynamic}
c = \begin{pmatrix} u \\ \dot\mu \end{pmatrix} = -A^\dagger B + k\,e\,\hat n, \qquad \rho = \hat n_{m+1} \ge 0,
\end{equation}
where now the first $m$ components are the physical control $u$ applied to the plant \eqref{affine_sys} and $\mu$ is a scalar controller state. Setting $f \equiv 0$, $g \equiv I$ gives $\Lambda = D h = D f$ and $a = 0$, recovering the kinematic controller of Section~\ref{sec:tracking} exactly; the correspondence is $D f \leftrightarrow \Lambda = D h\,g$, $\dot x \leftrightarrow u$, with the drift $a = D h\,f$ absorbed into $B$.

Because $A = [\Lambda \mid b]$ has the same block structure, Lemma~\ref{lem:kernel} applies with $D f$ replaced by $\Lambda$: away from control singularities $\rho = (1 + \|\Lambda^{-1} b\|^2)^{-1/2} \in (0,1]$, and at a rank-$(m-1)$ point the kernel collapses to $\ker\Lambda \times \{0\}$ with $\rho \to 0$, provided transversality holds.

\begin{definition}[Dynamic transversality]\label{def:dyntransversal}
A control-singular point ($\rank\Lambda(x) = m-1$) is \emph{transversal} if $w^T b \ne 0$, where $w \ne 0$ spans the left null space of the decoupling matrix, $w^T\Lambda(x) = 0$; equivalently $b = h(x_0) - y^* \notin \range\Lambda(x)$. The vector $w$ is the output direction along which the input instantaneously loses authority.
\end{definition}

With this substitution, Theorem~\ref{thm:crossing} (bounded crossing and re-locking), Proposition~\ref{prop:failure} (failure at loss of transversality), and the reflection-versus-crossing analysis hold verbatim: their proofs use only the block structure of $A$ and the boundedness of $A^\dagger$ and $B$, and $B$ remains bounded wherever $a(x)$ is bounded. In particular the output error obeys $\|h(x) - y^*\| = |e|\,\|b\| = O(1/k)$ away from control singularities, and $\mu$ re-locks after a transversal crossing.

\subsection{Zero dynamics and boundedness of the state}

One genuinely new ingredient separates the dynamic case from the kinematic one. When $n > m$ the controller regulates only the $m$-dimensional output while the plant carries $n-m$ internal states, and boundedness of the control $u$ --- guaranteed exactly as the boundedness of $\dot x$ in Theorem~\ref{thm:crossing}(1) --- no longer implies boundedness of the full state. The internal (zero) dynamics must be stable. On a domain where relative degree one holds, choose coordinates $(y,z)$ with $z \in \R^{n-m}$ complementing the output, in which \eqref{affine_sys} reads
\begin{equation}
\dot y = a(x) + \Lambda(x)\,u, \qquad \dot z = q(y,z).
\end{equation}

\begin{assumption}[Minimum phase]\label{as:minphase}
The internal dynamics $\dot z = q(y,z)$ is input-to-state stable with respect to $y$, and the reference $y^*(t)$ together with $\dot y^*(t)$ is bounded.
\end{assumption}

\begin{proposition}[Dynamic tracking through control singularities]\label{prop:dyntracking}
Consider the relative-degree-one system \eqref{affine_sys} under the transversality condition of Definition~\ref{def:dyntransversal} along the trajectory, isolated control singularities, and the minimum-phase Assumption~\ref{as:minphase}. Then the dynamic continuation controller \eqref{flow_dynamic} produces a bounded control $u$, keeps the full state $(y,z)$ bounded, and achieves output tracking $\|h(x(t)) - y^*(t)\| = O(1/k)$ away from the singularities, with the bounded-excursion and re-locking guarantees of Theorem~\ref{thm:crossing} at each transversal control singularity.
\end{proposition}

\begin{proof}
Transversality gives $\rank A = m$ throughout, so $A^\dagger = A^T(AA^T)^{-1}$ and the unit kernel $\hat n$ are bounded, and the control $u = (-A^\dagger B + k e\,\hat n)_{1:m}$ is bounded. The scalar error equation \eqref{error_ode}, with $D f$ replaced by $\Lambda$, yields $|e| = O(1/k)$ where $\rho \ge \rho_0$ and the bounded excursion of Theorem~\ref{thm:crossing}(3) across a transversal crossing; hence $\|y - y^*\| = |e|\,\|b\|$ is bounded and $O(1/k)$ in the regular regime. Boundedness of $y = h(x)$ and of $y^*$ then bounds $z$ through the ISS property of Assumption~\ref{as:minphase}, hence bounds $x$; bounded $x$ keeps $a(x)$, $\Lambda(x)$, $b(t)$ bounded, which closes the argument.
\end{proof}

\begin{remark}
The scalar motivating example of Section~\ref{sec:motivating} ($\dot x = u$, $y = x(x^2-1)$) is the case $m = n = 1$, $f = 0$, $g = 1$, with $\Lambda(x) = D h(x) = 3x^2-1$ vanishing at $x = \pm 1/\sqrt3$; there are no internal states and Assumption~\ref{as:minphase} is vacuous. The generalization thus subsumes both the motivating example and the kinematic arm, and extends the guarantees to genuine control singularities --- points of relative-degree loss --- of MIMO affine systems, provided the system is minimum phase.
\end{remark}

\section{Arbitrary relative degree via filtered-error reduction}\label{sec:reldeg}

The relative-degree-one theory of Section~\ref{sec:dynamic} extends to any relative degree through a standard reduction. Consider first a SISO system of relative degree $r$, so the input enters at the $r$-th derivative of the output,
\begin{equation}\label{rdr}
y^{(r)} = L_f^r h(x) + L_g L_f^{r-1} h(x)\,u = \alpha_r(x) + \beta(x)\,u,
\end{equation}
with the high-frequency gain $\beta(x) = L_g L_f^{r-1} h(x)$. The control singularity --- the loss of relative degree at which feedback linearization $u = (v - \alpha_r)/\beta$ diverges --- is $\beta(x) = 0$, the direct generalization of $\Lambda$ dropping rank in Section~\ref{sec:dynamic}.

Introduce the \emph{filtered tracking error}
\begin{equation}\label{sliding}
\sigma \;=\; \Bigl(\tfrac{d}{dt} + \lambda\Bigr)^{r-1}\,(y - y^*), \qquad \lambda > 0 .
\end{equation}
Its highest-order term is $(y-y^*)^{(r-1)}$, which is free of $u$, so $\sigma$ has relative degree exactly one:
\begin{equation}\label{sigma_dot}
\dot\sigma = a_\sigma(x,t) + \beta(x)\,u,
\end{equation}
with the \emph{same} gain $\beta(x)$ and a $u$-free drift $a_\sigma$ collecting $\alpha_r$, the reference derivatives, and the lower-order filter terms. Equation \eqref{sigma_dot} is precisely the relative-degree-one structure of Section~\ref{sec:dynamic} with scalar decoupling $\beta$. Applying the continuation controller to regulate $\sigma \to 0$ --- augmented matrix $A = [\beta \mid b_\sigma]$, transversality $b_\sigma \ne 0$ (equivalently $b_\sigma \notin \range\beta$) at $\beta = 0$, and the flow \eqref{flow_dynamic} --- yields, by Theorem~\ref{thm:crossing}, a bounded control and $|\sigma| = O(1/k)$ with re-locking through the singularity. Since \eqref{sliding} is a Hurwitz filter of order $r-1$ driven by $\sigma$, the output error inherits the bound
\begin{equation}
\Bigl(\tfrac{d}{dt}+\lambda\Bigr)^{r-1}(y - y^*) = \sigma \;\Longrightarrow\; |y - y^*| = O(1/k),
\end{equation}
the constant depending on the filter pole $\lambda$.

The MIMO case is identical. For a vector relative degree $\{r_1,\dots,r_m\}$ with decoupling matrix $\Lambda(x)$ of entries $L_{g_j} L_f^{r_i-1} h_i$, set $\sigma_i = (\tfrac{d}{dt}+\lambda_i)^{r_i-1}(y_i - y_i^*)$; then $\dot\sigma = a_\sigma + \Lambda(x)\,u$ has relative degree one, the augmented matrix is $A = [\Lambda \mid b_\sigma] \in \R^{m\times(m+1)}$, and every result of Sections~\ref{sec:tracking}--\ref{sec:dynamic} applies with $D f$ replaced by $\Lambda$.

\begin{corollary}[Arbitrary relative degree]\label{cor:rdr}
Let the system have a uniform (vector) relative degree with decoupling matrix $\Lambda(x)$, be minimum phase (input-to-state stable zero dynamics of dimension $n - \sum_i r_i$), and let the reference be bounded with bounded derivatives. If the filtered-error anchor $b_\sigma$ is transversal to $\Lambda$ at the isolated decoupling singularities ($\rank\Lambda = m-1$), then the continuation controller applied to $\sigma$ produces a bounded control $u$, keeps the full state bounded, drives $\sigma$ to an $O(1/k)$ band with the bounded-excursion and re-locking behaviour of Theorem~\ref{thm:crossing} at each singularity, and hence achieves bounded output tracking $\|y - y^*\| = O(1/k)$ --- passing through decoupling singularities where feedback linearization diverges.
\end{corollary}

Three costs of the reduction should be stated plainly. First, it requires a \emph{known}, uniform relative degree and the Lie derivatives up to order $r$, so the controller depends on more of the model and is correspondingly more sensitive to modelling error, exactly as exact feedback linearization is. Second, minimum-phaseness must now be certified on the larger $n - \sum_i r_i$ zero dynamics. Third, the relative-degree normal form itself degenerates at $\beta = 0$; the controller is unaffected because it operates in the original $x$-coordinates and only needs $\alpha_r, \beta$ (equivalently $a_\sigma, \Lambda$) as smooth functions of $x$, but the internal-boundedness argument must then be carried out in $x$-space via the minimum-phase assumption rather than in normal-form coordinates. Transient peaking of the filtered-error chain, familiar from high-gain designs, is also possible for aggressive $\lambda$.

\subsection{Relation to feedback linearization}\label{sec:fl}

We close the theory by placing the controller relative to exact feedback linearization: away from singularities it \emph{is} feedback linearization up to an $O(1/k)$ term, and its only added value is what it does \emph{at} the singularities.

\begin{proposition}[Consistency with feedback linearization]\label{prop:consistency}
In the regular region ($\Lambda$ invertible), at the locking steady state $\dot\mu = 1$, the continuation control \eqref{flow_dynamic} equals
\begin{equation}\label{u_reduction}
u = \Lambda^{-1}\bigl(\dot y^*(1-e) - a\bigr) = \Lambda^{-1}\bigl(\dot y^* - a\bigr) + O(1/k),
\end{equation}
i.e. it coincides with input--output feedback linearization (feedforward $\Lambda^{-1}(\dot y^* - a)$) up to an $O(1/k)$ tracking correction, and converges to it exactly as $k \to \infty$.
\end{proposition}

\begin{proof}
Take the scalar relative-degree-one case, $A = [\Lambda \mid b]$, $B = a - b - \dot y^*(1-e)$. The minimum-norm right inverse gives the particular solution component $u_p = -\Lambda B/(\Lambda^2+b^2)$ and $\dot\mu_p = -bB/(\Lambda^2+b^2)$, and the oriented unit kernel is $\hat n = (-b,\Lambda)/\sqrt{\Lambda^2+b^2}$ with $\rho = \Lambda/\sqrt{\Lambda^2+b^2}$. Imposing $\dot\mu = \dot\mu_p + k e\,\rho = 1$ fixes $k e = (1-\dot\mu_p)/\rho$. Substituting into $u = u_p + k e\,\hat n_1$ and simplifying,
\[
u = u_p - \frac{b}{\sqrt{\Lambda^2+b^2}}\cdot\frac{1-\dot\mu_p}{\rho}
  = -\frac{(\Lambda^2+b^2)(B+b)}{\Lambda(\Lambda^2+b^2)} = -\frac{B+b}{\Lambda} .
\]
Since $B + b = a - \dot y^*(1-e)$, this is $u = \Lambda^{-1}(\dot y^*(1-e) - a)$. The locking error satisfies $e = O(1/k)$ by Theorem~\ref{thm:crossing}(2), giving \eqref{u_reduction}. The vector and higher-relative-degree cases follow identically, the latter on the filtered variable $\sigma$: away from $\beta = 0$, $u = \beta^{-1}(-a_\sigma) + O(1/k)$.
\end{proof}

Using the output identity $y - y^* = e\,b$, the $O(1/k)$ term in \eqref{u_reduction} is exactly a proportional feedback of the tracking error, and the locking dynamics $\dot e = 1 - \dot\mu$, contracting at rate proportional to $k\rho$, plays the role of the linear outer loop of a feedback-linearizing tracker. The continuation controller is therefore a \emph{strict extension} of feedback linearization: identical to it (to $O(1/k)$) wherever $\Lambda$ is invertible, and bounded --- by Theorem~\ref{thm:crossing} and Proposition~\ref{prop:dyntracking} --- where $\Lambda$ is not.

\section{Multi-parameter continuation and adaptive control}\label{sec:adaptive}

The augmentation need not stop at one column. Introducing $p \ge 1$ continuation parameters enlarges the augmented Jacobian to
\begin{equation}\label{eq:multiaug}
A \;=\; [\,\Lambda \mid b_1\,\cdots\,b_p\,] \in \R^{m\times(m+p)},
\end{equation}
with a $p$-dimensional kernel at regular points. This section collects the structural results for \eqref{eq:multiaug} --- all machine-checked (Section~\ref{sec:lean}) --- and reads them as statements about adaptive control, where some of the columns are parameter sensitivities rather than homotopy directions. The numerical illustration is in Section~\ref{sec:adaptivesim}.

\begin{proposition}[Exact repair criterion]\label{prop:multirepair}
The augmented matrix \eqref{eq:multiaug} has full row rank $m$ if and only if the classes of $b_1,\dots,b_p$ in the cokernel $\R^m/\range\Lambda$ span it. In particular, at a corank-$p$ rank drop, $p$ directions with linearly independent cokernel classes repair the drop; and whenever the classes are linearly independent the kernel collapses to
\[
\ker A \;=\; \ker\Lambda \times \{0\},
\]
all $p$ parameter authorities vanishing simultaneously at the singularity. For $p = 1$ the criterion at a corank-one point is exactly the transversality of Definition~\ref{def:transversal}.
\end{proposition}

\begin{proof}
$\range A = \range\Lambda + \operatorname{span}\{b_1,\dots,b_p\}$, so $A$ is onto iff the added directions cover a complement of $\range\Lambda$, i.e.\ iff their cokernel classes span; the count follows since $p$ independent classes span an $r$-dimensional cokernel exactly when $r \le p$. If the classes are independent and $(v,c) \in \ker A$, then $\sum_i c_i\,[b_i] = [-\Lambda v] = 0$ in the cokernel forces $c = 0$, and then $\Lambda v = 0$. For $p = 1$ a single class is independent iff it is nonzero, i.e.\ iff $b \notin \range\Lambda$.
\end{proof}

\begin{theorem}[Residual authority]\label{thm:residual}
Let $\pi : (v, c) \mapsto c$ denote the projection onto the parameter block. Then $\pi(\ker A)$ is the kernel of the composite map $\bar B : c \mapsto \bigl[\sum_i c_i b_i\bigr] \in \R^m/\range\Lambda$: a rate vector $c$ extends to a kernel direction $(v,c) \in \ker A$ exactly when $\sum_i c_i b_i$ falls back into $\range\Lambda$. Consequently, if $A$ has full row rank at a corank-$r$ point,
\[
\dim \pi(\ker A) \;=\; p - r :
\]
each unit of corank consumed by the repair pins one combination of the parameter rates, and the remaining $p-r$ combinations stay free. At a regular point ($r = 0$) nothing is pinned, and for $p = r = 1$ the statement is the authority collapse $\rho = 0$ of Lemma~\ref{lem:kernel}(b).
\end{theorem}

\begin{proof}
$(v,c) \in \ker A$ iff $\sum_i c_i b_i = -\Lambda v$, so $c \in \pi(\ker A)$ iff $\sum_i c_i b_i \in \range\Lambda$, i.e.\ iff $\bar B c = 0$. If $A$ has full row rank then $\bar B$ is onto (its range is the span of the cokernel classes, which is everything by Proposition~\ref{prop:multirepair}), and rank--nullity gives $\dim\ker\bar B = p - \dim(\R^m/\range\Lambda) = p - r$.
\end{proof}

There is a precise precedent in numerical algebraic geometry: a generic parameter homotopy avoids the discriminant (singular) locus with probability one, meeting it only at isolated points, the added parameters lifting the solution path off the singular set (the ``gamma trick''). The parameter $\mu$ and anchor $b$ play exactly this role here, and the multi-parameter augmentation \eqref{eq:multiaug} is the systematic version: $p$ parameters transversally repair singularities of corank up to $p$, strictly enlarging the admissible singular set beyond the corank-one case treated above, while the same $p$-dimensional null space supplies the freedom for the auxiliary control of Section~\ref{sec:aux}.

\paragraph{Reading for adaptive control.} Let the plant depend on unknown parameters $\theta \in \R^q$ with online estimate $\hat\theta$, and run the continuation in the \emph{joint} variable $(x, \mu, \hat\theta)$: the homotopy map acquires the sensitivity columns $s_j = \partial H/\partial\hat\theta_j$, and the augmented Jacobian is exactly \eqref{eq:multiaug} with the columns split as $[\Lambda \mid b \mid s_1 \cdots s_q]$ --- homotopy direction and sensitivities entering on the same footing, as repair directions. Proposition~\ref{prop:multirepair} gives the exact condition for the joint tangent to stay well defined: the classes of $b, s_1, \dots, s_q$ must span the cokernel, so parameter mismatch that is visible in the output can substitute for continuation authority. The kernel of \eqref{eq:multiaug} consists precisely of the joint motions that leave the output untouched, so estimator updates projected onto $\ker A$ are \emph{exactly} output-invariant, where classical adaptive designs obtain this only asymptotically through swapping arguments. For identification with nonconvex parameterization (see the survey \cite{Borisevich2013}) the reading is sharpest: the spurious equilibria that obstruct gradient identification lie on the singular manifolds of the regressor, and these are transversality failures in the parameter block of \eqref{eq:multiaug}. Instead of \emph{avoiding} the regressor's singular manifolds --- the usual approach --- the continuation \emph{crosses} them transversally with the kernel orientation policy and bounded update rates, exactly as the controller crosses decoupling singularities. Persistent excitation becomes the requirement that the repair criterion hold uniformly along the trajectory, and Theorem~\ref{thm:residual} quantifies how much adaptation authority survives a crossing: a corank-$r$ singularity pins $r$ combinations of the joint rates and leaves $p - r$ free. A full adaptive design --- convergence under persistent excitation, mismatch robustness of the tracking bounds --- is beyond the present scope, but the frozen-parameter estimates of Theorem~\ref{thm:crossing} transfer verbatim to the parameter block, and Section~\ref{sec:adaptivesim} demonstrates the crossing on the smallest nonconvex identification problem.

\section{Machine-checked formalization}\label{sec:lean}

Before turning to the numerical studies we report a complementary form of validation: the core results of Sections~\ref{sec:augmented}--\ref{sec:reflection} and~\ref{sec:adaptive} have been formalized and verified in the Lean~4 proof assistant over the Mathlib mathematical library. Every formalized statement is proved completely, with no admitted goals, and each theorem was audited to depend only on the three standard axioms of Lean's classical foundation. This section records what was formalized and where the boundaries of the verification lie; the development is available from the author.

\paragraph{Scope.} The formalized material comprises: the kernel structure lemma (Lemma~\ref{lem:kernel}), both in coordinate-free form and in the matrix form of Definition~\ref{def:transversal}, including the equivalence of $b \notin \range\Lambda$ with $w^Tb \ne 0$ at a corank-one point; the failure statement (Proposition~\ref{prop:failure}); the authority--boundedness tradeoff (Remark~\ref{rem:tradeoff}), promoted to a theorem; the no-blow-up statement of Theorem~\ref{thm:crossing}(1), in the form: at a transversal singularity the Gram matrix $AA^T$ of the augmented matrix is positive definite, $A^\dagger = A^T(AA^T)^{-1}$ is a right inverse, and $A^\dagger y$ is norm-minimal among all solutions of $Az = y$ --- the least-norm property that names the inverse, here established by proof, which bounds the commanded velocity by every feasible one; the contraction and bounded-excursion estimates of Theorem~\ref{thm:crossing}(2)--(4), obtained from a one-sided Gr\"onwall comparison on $|e|$ valid for either sign of the rate; the sliding-sign computation (Proposition~\ref{prop:sign}); the consistency computation (Proposition~\ref{prop:consistency}); the exact-locking statements of Remark~\ref{rem:exactlocking} (exponential convergence of the locking error to zero under drift feedforward, and under PI locking via a strict quadratic Lyapunov function for the two-dimensional locking loop); and the sliding-mode content of Section~\ref{sec:reflection} in the scalar fold model --- a minimal constructive theory of Filippov solutions (defined in integral-selection form), the two-field convexification, finite-time attraction, confinement, the uniqueness of the sliding motion (Remark~\ref{rem:leanstrength}), release, the multi-parameter repair and residual-authority statements of Section~\ref{sec:adaptive} (Proposition~\ref{prop:multirepair} and Theorem~\ref{thm:residual}, including the exact computation of the kernel's parameter-block projection), and both the reflection and crossing lifts together with their deck relation. The kernel localization of Proposition~\ref{prop:reduction} --- at the Whitney fold point the augmented kernel is spanned by $(e_1,0)$ and transversality reduces to $b_1\ne0$ --- is likewise verified, subject to the caveat of Remark~\ref{rem:nonequivariance}.

\paragraph{Boundaries.} The uniqueness statement through the release instant is verified in full: every Filippov solution of the fold model through the singular point departs the singular set at the release instant and coincides thereafter with the unique classical release arc (by a Gr\"onwall argument on the integral form), so the reflection is not merely \emph{a} behaviour of the convexified closed loop but \emph{the} behaviour. One ingredient remains hypothesized rather than proved: the Whitney normal form (Proposition~\ref{prop:foldnf}) enters as the choice of coordinates, exactly as in the proofs above; it is checkable instance by instance, and every example of Sections~\ref{sec:sims}--\ref{sec:converter} verifies it explicitly. None of the numerical claims of the following sections depends on it.

\section{Numerical simulations}\label{sec:sims}

We validate the theory on a progression of systems: the scalar cubic of Section~\ref{sec:motivating}, a 2-DOF kinematic arm (which reflects at maximum reach and, under the opposite orientation, crosses the elbow branch), and dynamic plants of relative degree one and two. Unless stated otherwise each run uses the $\rho\ge0$ orientation and gain $k=100$, and each measured quantity is matched to the corresponding statement of Theorem~\ref{thm:crossing}. Throughout, the continuation control stays bounded where feedback linearization diverges by many orders of magnitude.

\paragraph{Constraint stabilization.} All runs integrate the continuation flow \eqref{flow} in real time with a fixed-step RK4 scheme, with one numerical safeguard worth stating. The flow only \emph{conserves} the homotopy constraint: it enforces $\dot H = 0$, so that $H\equiv 0$ holds identically along the exact trajectory once $H(x_0,\mu(0),0)=0$. But $\{H=0\}$ is a marginally stable invariant with no restoring action, so the local error the integrator commits at the stiff singular crossings accumulates as a slow drift off the manifold, corrupting the output identity $f(x)-y^* = e\,b$. We suppress this purely numerical artifact by a Baumgarte stabilization: in the particular solution the drift $B$ is replaced by $B + c_H\,H$, which turns the constraint dynamics into $\dot H = -c_H\,H$ and makes $\{H=0\}$ attracting. Along the exact solution $H\equiv 0$, so the added term vanishes identically there and the continuous-time guarantees of Theorem~\ref{thm:crossing} are untouched; it acts solely as an invariant-manifold stabilizer, as is standard for index-one differential-algebraic systems. The gain $c_H$ is of order $k$, and its magnitude matters only in proportion to the number and violence of the crossings a trajectory makes: the toy example below, whose reference sweeps far past both folds and changes branch repeatedly, is the extreme case that genuinely requires it, whereas a single gentle reflection injects almost no drift and needs only a small $c_H$.

\paragraph{Kernel computation and orientation (the behavioural choice).} The unit kernel $\hat n$ of the augmented matrix $A$ is computed directly: for a two-output map ($A\in\R^{2\times3}$) as the normalized cross product of the two rows of $A$, and in general (e.g.\ the $3\times4$ converter map) as the last right-singular vector of $A$ from an SVD. Since $\ker A$ is one-dimensional only the \emph{sign} of $\hat n$ is free, and by Definition~\ref{def:controllaw} that sign is the design choice that selects reflection versus crossing. Writing $\hat n_k$ for the unit kernel evaluated at integration step $k$, the two orientations are enforced as per-step sign selections:
\begin{itemize}
\item[\textbf{(R)}] $\rho\ge0$ (reflecting): take $\hat n_k$ with $(\hat n_k)_{n+1}\ge0$ (flip the computed kernel when its parameter component is negative). Used for the arm at maximum reach, the dynamic relative-degree one and two plants, and the converter.
\item[\textbf{(C)}] continuity (crossing): take $\hat n_k$ with $\hat n_k^{T}\hat n_{k-1}\ge0$, i.e.\ carry the tangent over from the previous step. This is the discrete realization of the continuous arc-length orientation \eqref{conditions} of Definition~\ref{def:controllaw}: as $\Delta t\to0$ the selection $\hat n_k^{T}\hat n_{k-1}\ge0$ reproduces the continuous unit tangent field. Used for the toy example and the elbow-flip run of Section~\ref{sec:arm}.
\end{itemize}
This sign selection is the entire mechanism behind the reflection/crossing dichotomy: the same flow \eqref{flow}, integrated with rule (R) or (C) and everything else (reference, gain, plant) held fixed, produces the qualitatively different outcomes reported below. Accordingly we state the orientation used for each experiment.

\subsection{Toy example}

We track the reference $y^*(t) = -6\cos t$ for the scalar system \eqref{toy_example}, $h(x) = x(x^2-1)$, at gain $k=100$. The amplitude far exceeds the fold values $h(\pm 1/\sqrt3) = \mp 2/(3\sqrt3) \approx \mp 0.385$, so tracking forces the state to sweep the full range $x\in[-2,2]$ and to cross \emph{both} decoupling singularities $x = \pm 1/\sqrt3$ twice per period. For $|y^*| > 0.385$ the demanded output is realized only on the outer branches of the cubic, not on the branch the state currently occupies, so tracking \emph{requires} passing through each fold; we therefore use the continuity (arc-length) orientation of $\hat n$. This is the crossing case of the dichotomy: unlike the arm at maximum reach or the converter at resonance --- where nothing exists beyond the fold, so the $\rho\ge0$ orientation correctly reflects and rides the boundary --- here the levels the reference demands (up to $|y^*|=6$) genuinely live on the far sheet, and reflection would merely saturate the output at the fold value $\pm0.385$ and fail to track. Crossing is thus not optional but forced (the reflection-versus-crossing dichotomy of Theorem~\ref{thm:reflection}).

The continuation negotiates both folds with bounded input. The commanded control $|u| = |\dot x|$ peaks at $\approx 33$, whereas the continuous-Newton law $\dot x = \dot y^*/Dh$ diverges as $1/|Dh|$ and demands $\sim 4\times 10^{4}$ at the crossings (here $\min_t |Dh| \approx 1.5\times 10^{-4}$). Away from the folds the locking error stays at $|e| \approx 1/k$ and $x$ follows the reference; at each fold the exact inverse is genuinely discontinuous---the state must traverse the non-invertible interval while $y^*$ passes through the fold value---so $y$ shows a brief transient and $\mu$ lags $t$ by a bounded amount (peak $|e|\approx 0.3$) before re-locking to $\approx 10^{-2}$. Under the Baumgarte stabilization the constraint identity is held to $\max_t|H|\approx 4\times 10^{-11}$ through the repeated crossings.

\begin{center}
\ifpdf %if using PDFTeX in PDF mode
  % this figure is generated by sim_toy.py
  \includegraphics[width=1\textwidth]{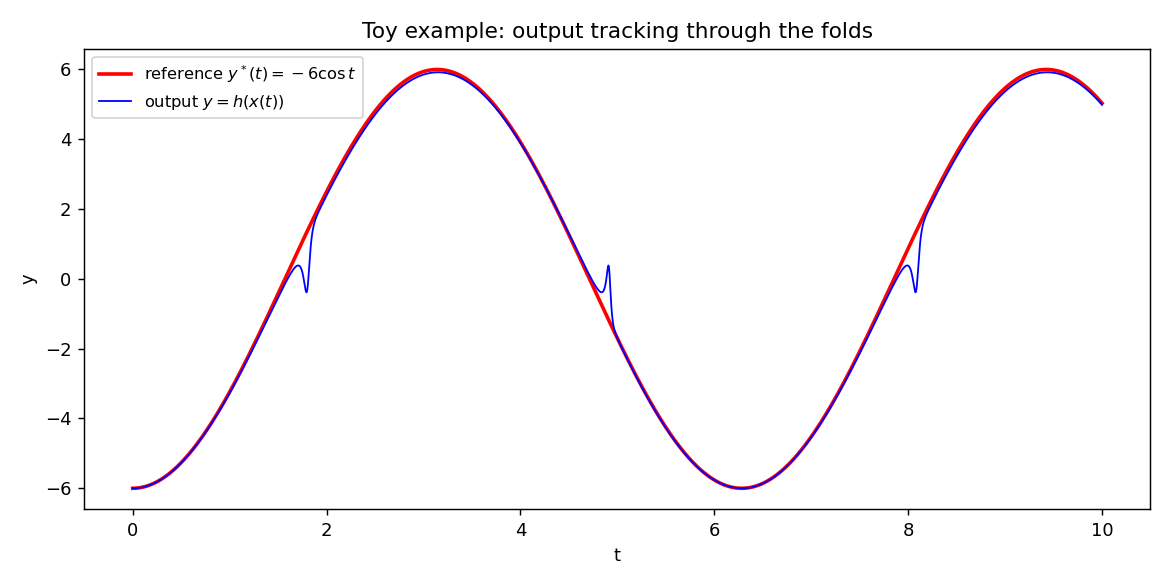}
\fi

Figure 1. Output $y = h(x(t))$ (blue) tracking the reference $y^*(t) = -6\cos t$ (red); the brief departures coincide with the fold crossings, where the exact inverse is discontinuous.
\end{center}

\begin{center}
\ifpdf %if using PDFTeX in PDF mode
  % this figure is generated by sim_toy.py
  \includegraphics[width=1\textwidth]{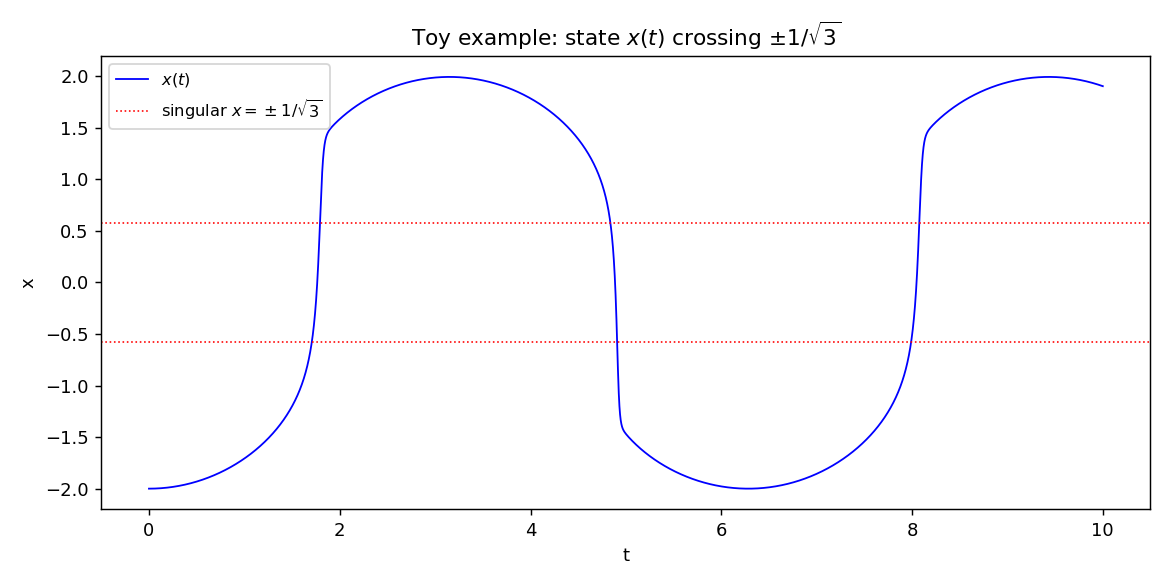}
\fi

Figure 2. State $x(t)$ sweeping $[-2,2]$ and crossing both decoupling singularities $x = \pm 1/\sqrt3$ (dotted), where $Dh = 3x^2 - 1$ vanishes.
\end{center}

\begin{center}
\ifpdf %if using PDFTeX in PDF mode
  % this figure is generated by sim_toy.py
  \includegraphics[width=1\textwidth]{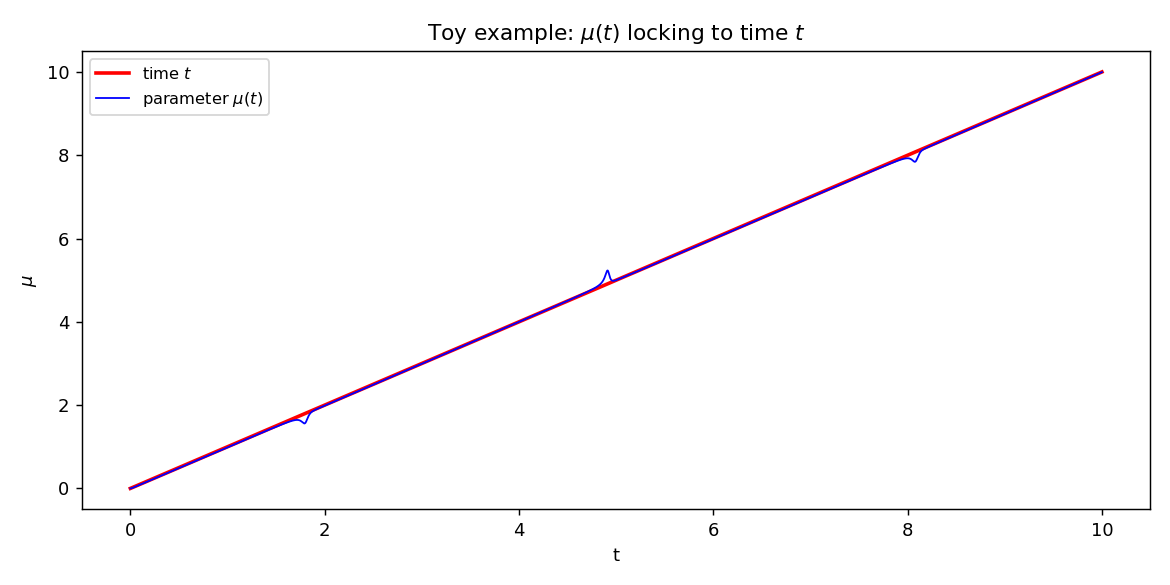}
\fi

Figure 3. Parameter $\mu(t)$ (blue) versus time $t$ (red): $\mu$ locks to $t$ away from the folds and lags briefly at each crossing before re-locking.
\end{center}

\subsection{Removing the $O(1/k)$ bias: locking compensators}\label{sec:lockingsims}

We compare the three locking laws of Remark~\ref{rem:exactlocking} --- the architecture of the full loop, with both compensators, is shown in Figure~\ref{fig:lockingarch} --- on the same cubic plant and reference, at the deliberately moderate gain $k = 20$ so that the proportional band $1/k = 5\times 10^{-2}$ is clearly visible, with all else identical (continuity orientation, both folds crossed twice). The PI gains are $k_p = k$, $k_i = k^2/4$ --- critical damping of the second-order locking loop, whose damping ratio at frozen authority is $\zeta = k_p\sqrt{\rho}/(2\sqrt{k_i})$ --- and the feedforward saturation is $M = 25$, gated off for $|\rho| < 0.2$. The PI output is clamped to $|\gamma| \le 25$ with conditional-integration anti-windup: the unclamped output $k_p e + k_i q$ is computed first, the integrator advances only if that output is not saturated in the direction the error pushes, and the clamped value is applied. This handles the fold crossings automatically, with no reference to the controller-internal authority $\rho$: through a transit the locking error grows (the anti-dissipative window of Proposition~\ref{prop:excursion}(b), during which $\rho < 0$ momentarily reverses the loop gain of the integral path), the output saturates at once, and integration freezes until the crossing completes --- without this clamping the integrator winds at each transit and the (then underdamped) loop rings visibly after every crossing. Figure~4 reports the outcome. In the regular region the proportional law settles at the predicted band (median $|e| = 4.6\times 10^{-2} \approx 1/k$); drift feedforward is exact there (median $|e| = 2.3\times 10^{-7}$, the integrator floor), and PI removes the bias without any model of the drift (median $|e| = 7.7\times 10^{-3}$, a factor $6$ below the proportional band, the residual reflecting the time-varying drift that the integrator tracks with finite bandwidth --- larger $k_i$ shrinks it at the cost of damping). All three laws pass the four fold crossings with bounded control; the feedforward variant pays for its regular-region exactness with the largest transients at the gate handoff ($\max|u| = 45$ against $14$ for P and $17$ for PI), so clamped PI is the practical default when the drift model is uncertain --- consistent with industrial practice.

\begin{figure}[htbp]
\centering
% Architecture of the continuation controller with the locking compensator
% (only the pictures; wrapped in a figure environment in the main file).
% Requires: tikz with libraries positioning, arrows.meta (main preamble).
% ---------- main loop ----------
\begin{tikzpicture}[font=\footnotesize, >=Latex,
  block/.style={draw,rounded corners=1pt,align=center,inner sep=3pt,minimum height=7mm},
  sum/.style={draw,circle,inner sep=0.5pt,minimum size=3.6mm},
  node distance=7mm and 8mm]
  % ---- main row ----
  \node[block] (asm) {homotopy data\\[-1pt]
    $b=f(x_0)-y^*(t)$\\[-1pt] $B=a(x)-b-\dot y^*(1-e)$};
  \node[block, right=of asm] (dec) {augmented decomposition\\[-1pt]
    $A=[\,\Lambda(x)\mid b\,]$,\quad $-A^\dagger B$,\quad $\hat n$ (orientation)\\[-1pt]
    $d=(-A^\dagger B)_{m+1}$,\quad $\rho=\hat n_{m+1}$};
  \node[sum, right=of dec] (sm) {$+$};
  \node[block, right=11mm of sm] (plant) {plant\\[-1pt] $\dot x=f(x)+g(x)u$};
  % branch point on the output wire (the flow splits into u and mu-dot)
  \coordinate (br) at ($(sm.east)!0.55!(plant.west)$);
  \fill (br) circle (1pt);
  % ---- locking row (right to left: integrator, error, compensator) ----
  \node[block] (integ) at ($(br)+(0,-20mm)$) {$\int$};
  \node[sum, left=1.06cm of integ] (esum) {$+$};
  \node[block, left=0.68cm of esum] (gam) {locking compensator\\[-1pt] $\gamma(e;d,\rho)$};
  % ---- main path ----
  \draw[->] ($(asm.west)+(-7mm,0)$) -- (asm.west) node[pos=0,anchor=east] {$y^*,\dot y^*$};
  \draw[->] (asm) -- (dec) node[midway,above] {$b,B$};
  \draw[->] (dec) -- (sm) node[midway,above] {$-A^{\dagger}B$};
  \draw[->] (sm) -- (plant) node[pos=0.8,above] {$u$};
  % state feedback over the top
  \draw[->] (plant.east) -- ++(5mm,0) -- ++(0,13mm) -| (dec.north)
    node[pos=0.25,above] {$x$};
  \draw[->] ($(plant.east)+(5mm,13mm)$) -| (asm.north);
  % ---- locking loop ----
  \draw[->] (br) -- (integ.north) node[midway,right] {$\dot\mu$};
  \draw[->] (integ) -- (esum) node[midway,above] {$\mu$} node[pos=0.85,below] {$-$};
  \draw[->] ($(esum.north)+(0,5mm)$) -- (esum.north)
    node[anchor=south, at start] {clock $t$} node[pos=0.8,right] {$+$};
  \draw[->] (esum) -- (gam) node[midway,above] {$e$};
  \draw[->] ($(gam.north)+(6mm,0)$) -- ++(0,0.78) -| (sm.south) node[above, pos=0.39] {$\gamma\,\hat n$};
  \draw[->] (dec.south)  -| ($(gam.north)+(-6mm,0)$)
    node[below right, pos=0.67] {\scriptsize $d,\rho$};
\end{tikzpicture}

\medskip

% ---------- compensator variants ----------
\begin{tikzpicture}[font=\footnotesize, >=Latex,
  block/.style={draw,rounded corners=1pt,align=center,inner sep=3pt,minimum height=6mm},
  sum/.style={draw,circle,inner sep=0.5pt,minimum size=3.6mm},
  node distance=5mm and 7mm]
  \node[block] (kp) {$k$};
  \node[block] (ff) at (-0.4,1.2) {$\mathrm{sat}_M\!\left(\dfrac{1-d}{\rho}\right)$\\[-1pt]
    \scriptsize gated: $|\rho|\ge\rho_{\min}$};
  \node[sum, right=8mm of kp] (s) {$+$};
  \draw[->] ($(kp.west)+(-6mm,0)$) -- (kp.west) node[pos=0,anchor=east] {$e$};
  \draw[->] ($(ff.west)+(-6mm,0)$) -- (ff.west) node[pos=0,anchor=east] {$d,\rho$};
  \draw[->] (kp) -- (s);
  \draw[->] (ff.east) -| (s.north);
  \draw[->] (s.east) -- ++(6mm,0) node[anchor=west] {$\gamma$};
  \node[anchor=north,align=center] at ($(kp.south)+(4mm,-3mm)$)
    {(a) drift feedforward: exact locking,\\ $\dot e=-k\rho e$ in the regular region};
\end{tikzpicture}\hspace{14mm}%
\begin{tikzpicture}[font=\footnotesize, >=Latex,
  block/.style={draw,rounded corners=1pt,align=center,inner sep=3pt,minimum height=6mm},
  sum/.style={draw,circle,inner sep=0.5pt,minimum size=3.6mm},
  node distance=5mm and 7mm]
  \node[block] (kp) {$k_p$};
  \node[block, above=of kp] (int) {$\int$\ \scriptsize conditional};
  \node[block, right=5mm of int] (ki) {$k_i$};
  \node[sum, right=16mm of kp] (s) {$+$};
  \node[block, right=0.58cm of s] (sat) {$\mathrm{sat}$};
  \coordinate (ein) at ($(kp.west)+(-6mm,0)$);
  \draw[->] (-1.4,0) -- (kp.west) node[anchor=east, at start] {$e$};
  \draw[->] ($(kp.west)+(-10mm,0)$) |- (int.west);
  \draw[->] (int) -- (ki);
  \draw[->] (kp) -- (s);
  \draw[->] (ki.east) -| (s.north);
  \draw[->] (s) -- (sat);
  \draw[->] (sat.east) -- ++(5mm,0) node[anchor=west] {$\gamma$};
  \draw[->,dashed] (sat.north) -- ++(0,0.19) -| ($(int.south)+(5mm,0)$)
    node[above, at start] {\scriptsize freeze if saturated};
  \node[anchor=north,align=center] at ($(kp.south)+(9mm,-3mm)$)
    {(b) clamped PI: exact for constant drift,\\ no model of $d$ required};
\end{tikzpicture}
\caption{Architecture of the continuation controller. Top: the main loop --- the homotopy data $(b, B)$ and the augmented decomposition ($-A^\dagger B$, unit kernel $\hat n$ with the orientation policy, drift $d$, authority $\rho$) produce the flow $\dot c = -A^\dagger B + \gamma\,\hat n$; the first $m$ components drive the plant, the last integrates to the internal parameter $\mu$, and the locking error $e = t - \mu$ closes the loop through the locking compensator $\gamma(e; d, \rho)$. Bottom: the two exact-locking variants of Remark~\ref{rem:exactlocking} --- (a) drift feedforward with saturation, gated off where the authority collapses; (b) clamped PI with conditional-integration anti-windup --- the integrator is frozen whenever the unclamped output is saturated in the direction the error pushes, which in particular freezes it automatically through the fold transits. The proportional law of the main text is the special case $\gamma = k e$.}
\label{fig:lockingarch}
\end{figure}

\begin{center}
\ifpdf
  % this figure is generated by sim_locking_comparison.py
  \includegraphics[width=1\textwidth]{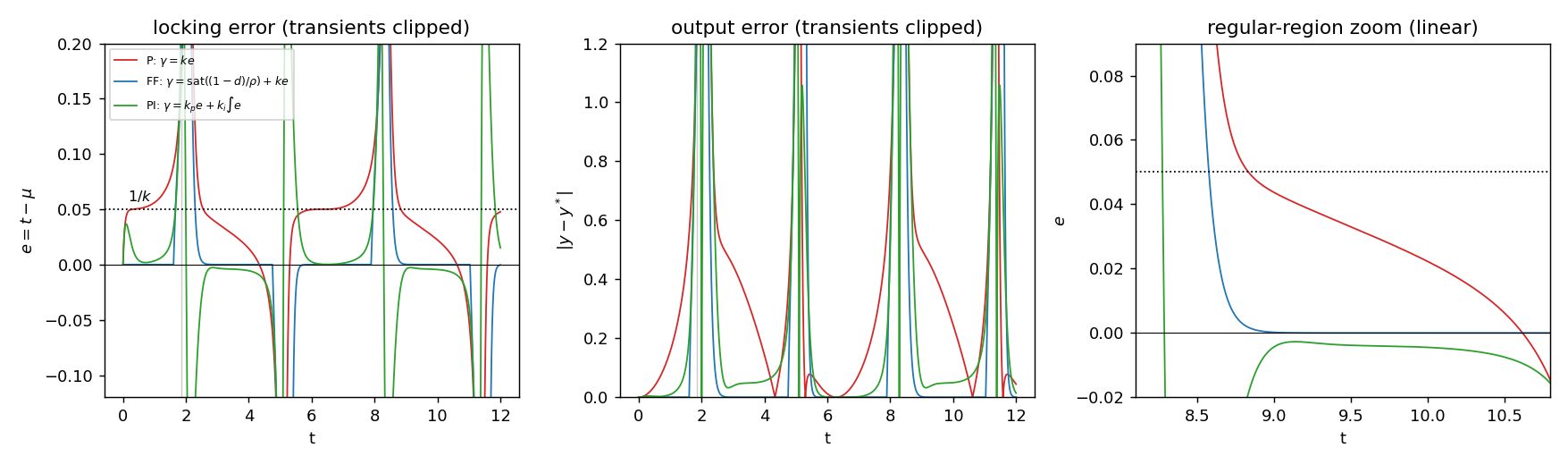}
\fi

Figure 4. Locking-law comparison on the toy example at $k=20$ (vertical grey lines mark the fold crossings; the crossing transients, up to $|e| = 1.3$ and $|y-y^*| = 3.6$, are clipped to keep the regular-region behaviour readable). Left: signed locking error $e = t-\mu$ --- the proportional law (red) hovers on the $1/k$ band (dotted), feedforward (blue) and PI (green) remove the bias; after each transit the PI trace settles through zero once as the integral state moves to its new drift-compensating value. Middle: output error $|y - y^*|$, inheriting the same ordering through the identity $y - y^* = e\,b$. Right: regular-region zoom --- the proportional offset against the zero-mean compensated laws.
\end{center}

\subsection{Adaptive identification through a singular regressor manifold}\label{sec:adaptivesim}

We illustrate the adaptive-control reading of Section~\ref{sec:adaptive} on the smallest nonconvex identification problem: model output $\hat y = g(\hat\theta) = \hat\theta^3 - 3\hat\theta$, measurement $y_m = g(\theta^*)$ with $\theta^* = 2.2$, initial guess $\hat\theta_0 = -2$. The regressor $g'(\hat\theta) = 3\hat\theta^2 - 3$ vanishes on the manifold $\hat\theta = \pm 1$, and every path from $\hat\theta_0$ to $\theta^*$ crosses both of its sheets. The gradient (MIT-rule) estimator $\dot{\hat\theta} = -\gamma\, g'(\hat\theta)\,(g(\hat\theta) - y_m)$ converges to $\hat\theta = -1$: a spurious equilibrium lying \emph{on} the singular manifold, where the residual is far from zero ($|g(\hat\theta) - y_m| = 2.05$) but the gradient direction has died --- the nonconvex-parameterization obstruction of \cite{Borisevich2013} in its purest form. The normalized-gradient (Newton) estimator's demand $|e/g'|$ exceeds $10^4$ along the path, so it can succeed only by avoiding the manifold. The continuation estimator instead runs the homotopy $H(\hat\theta, \mu) = g(\hat\theta) - (1-\mu)\,g(\hat\theta_0) - \mu\, y_m$ jointly in $(\hat\theta, \mu)$, with augmented map $A = [\,g'(\hat\theta) \mid b\,]$, $b = g(\hat\theta_0) - y_m \ne 0$: transversality holds identically on the manifold, and the kernel orientation (continuity) policy crosses both sheets --- $g'$ changes sign twice along the path, with $\min|g'| = 1.3\times10^{-4}$ --- converging to $\hat\theta = \theta^*$ with residual $\sim 10^{-12}$ while $\mu$ makes the characteristic bounded excursion and re-locks to its schedule. The residual-authority count of Theorem~\ref{thm:residual} is visible directly: wherever $|g'| < 10^{-2}$ the kernel's parameter component obeys $|\rho| \lesssim |g'|/|b|$, the parameter block pinned at the crossing ($p - r = 0$ here, with $p = r = 1$) and free elsewhere. Figure~5 reports the run.

\begin{center}
\ifpdf
  % this figure is generated by sim_adaptive_toy.py
  \includegraphics[width=1\textwidth]{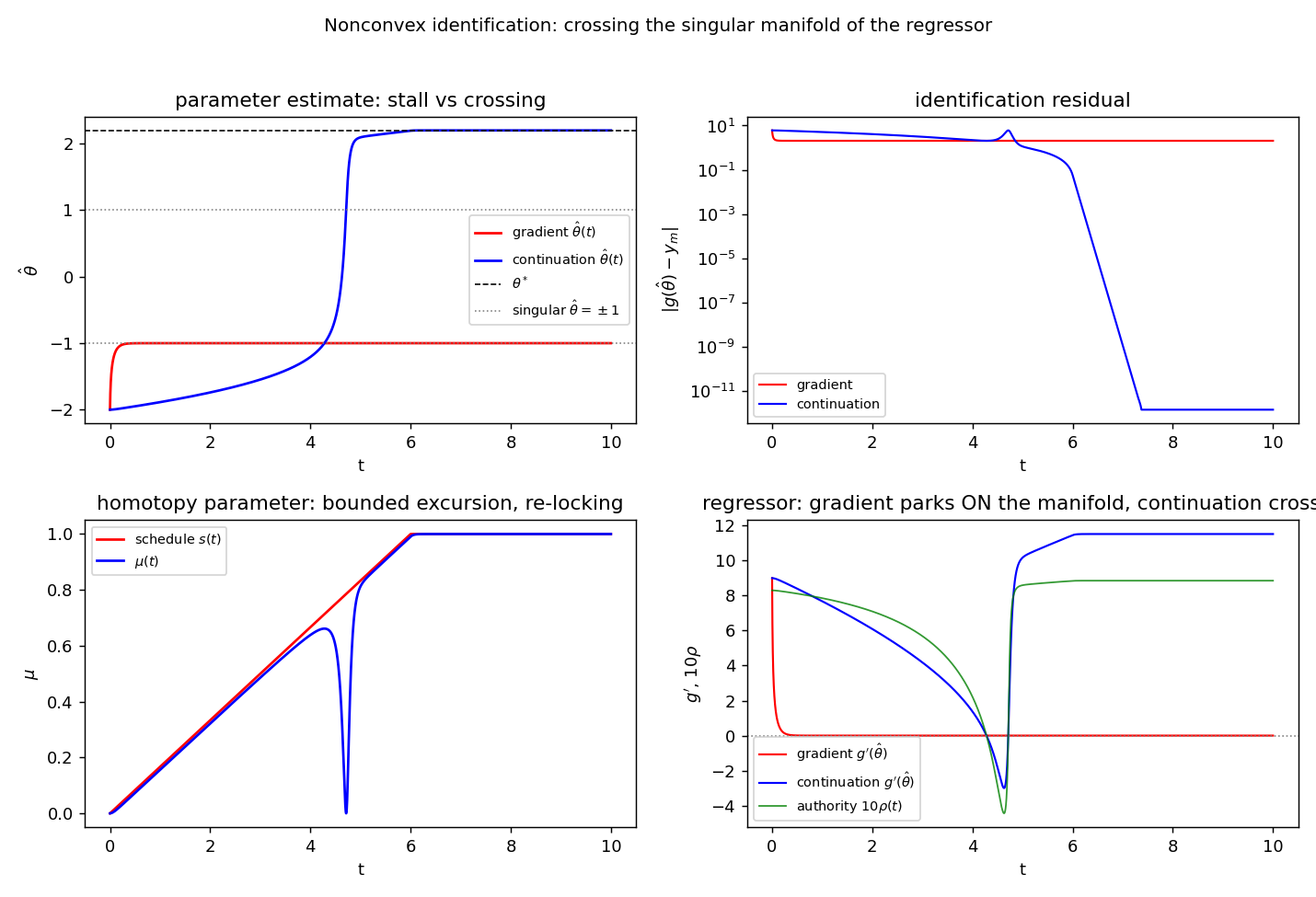}
\fi

Figure 5. Nonconvex identification $g(\hat\theta) = \hat\theta^3 - 3\hat\theta$, $\theta^* = 2.2$: gradient estimator (red) versus joint continuation (blue). Top left: $\hat\theta(t)$ --- the gradient stalls exactly on the singular manifold $\hat\theta = -1$ (dotted), the continuation crosses both sheets and converges to $\theta^*$ (dashed). Top right: identification residual. Bottom left: the homotopy parameter's bounded excursion and re-locking to the schedule. Bottom right: the regressor $g'(\hat\theta)$ along both paths --- the gradient's regressor dies on the manifold, the continuation's passes through zero twice; the authority $\rho$ (green) collapses exactly at the crossings, as Theorem~\ref{thm:residual} requires.
\end{center}

\subsection{2-DOF planar robotic arm}\label{sec:arm}

The singular positions of a 2-DOF planar robotic arm occur when the robot's Jacobian matrix becomes singular. This leads to a loss of controllability or a situation where certain end-effector motions are not possible.

The end-effector position of the robot is given by forward kinematics equations:
\[
x = L_1 \cos(q_1) + L_2 \cos(q_1 + q_2),
\]
\[
y = L_1 \sin(q_1) + L_2 \sin(q_1 + q_2),
\]
where:
- \( q_1, q_2 \): Joint angles.
- \( L_1, L_2 \): Link lengths.

The Jacobian matrix relates the joint velocities \( \dot{\mathbf{q}} = [\dot{q}_1, \dot{q}_2]^T \) to the end-effector velocities \( \dot{\mathbf{x}} = [\dot{x}, \dot{y}]^T \):
\[
\dot{\mathbf{x}} = \mathbf{J}(\mathbf{q}) \dot{\mathbf{q}},
\]

where the Jacobian matrix \( \mathbf{J}(\mathbf{q}) \) is:
\[
\mathbf{J}(\mathbf{q}) =
\begin{bmatrix}
-\left( L_1 \sin(q_1) + L_2 \sin(q_1 + q_2) \right) & -L_2 \sin(q_1 + q_2) \\
\left( L_1 \cos(q_1) + L_2 \cos(q_1 + q_2) \right) & L_2 \cos(q_1 + q_2)
\end{bmatrix}.
\]

Singularities occur when the determinant of the Jacobian matrix is zero:
\[
\det(\mathbf{J}(\mathbf{q})) = 0.
\]

The determinant becomes zero when:
\[
\sin(q_2) = 0 \quad \implies \quad q_2 = 0 \, \text{or} \, q_2 = \pi.
\]

The physical interpretation of the singularities consists of two cases: when \( q_2 = 0 \) the two links are fully extended in a straight line and the arm is at its maximum reach, or when \( q_2 = \pi \) the two links are folded back into a straight line and the arm is at its minimum reach, fully retracted.

Consider the reference circular trajectory centered at \( (x_c, y_c) \) with radius \( r \), assuming the target circle lies within the reachable workspace of the arm:

\[
  x_d(t) = x_c + r \cos(\omega t), \\
  y_d(t) = y_c + r \sin(\omega t),
\]

where \( \omega \) is the angular velocity of the circular motion and \( t \) is time.

\subsubsection{The continuation controller for the arm}

The arm is a velocity-controlled kinematic system: the joint rates $\dot q = (\dot q_1, \dot q_2)$ are the manipulated inputs and the end-effector position plays the role of the output map, $f(q) = (x,y)$, with $D f(q) = \mathbf{J}(q)$. Applying the construction of Section~\ref{sec:tracking}, we anchor the homotopy at the initial configuration $q_0$ (chosen on the reference, $f(q_0) = y^*(0)$), form the homotopy direction $b(t) = f(q_0) - y^*(t)$, and assemble the augmented Jacobian
\[
A(q,t) = \bigl[\, \mathbf{J}(q) \;\big|\; b(t) \,\bigr] \in \R^{2\times 3},
\qquad
B(q,\mu,t) = -\,b(t) - \dot y^*(t)\,(1-e), \quad e = t-\mu .
\]
The controller integrates, in real time with $\dot t = 1$, the flow \eqref{flow}
\[
\begin{pmatrix}\dot q \\ \dot\mu\end{pmatrix}
= -A^\dagger B + k\, e\, \hat n,
\qquad
A^\dagger = A^T (A A^T)^{-1},
\]
where $\hat n$ is the unit generator of $\ker A$ oriented so that its parameter component $\rho = \hat n_3 \ge 0$, and $\gamma = k\,e = k\,(t-\mu)$ is the parameter-locking feedback. The commanded joint rates are the first two components of the flow; $\mu$ is an internal state.

The behaviour is governed entirely by Lemma~\ref{lem:kernel}. Away from the singular set ($\sin q_2 \ne 0$) the Jacobian has rank $2$, the kernel has a nonzero parameter component $\rho\in(0,1]$, the feedback drives $e \to O(1/k)$, and the arm tracks the circle. As the arm approaches full extension $q_2 \to 0$, the Jacobian drops to rank~$1$; its left null direction is $w = (J_{21},\, -J_{11})$. Provided the \emph{transversality} condition of Definition~\ref{def:transversal} holds, $w^T b \ne 0$, the augmented matrix $A$ remains rank~$2$ and the flow stays well defined through the singularity, but the kernel rotates until $\rho \to 0$: the loss of authority over $\mu$ predicted by Lemma~\ref{lem:kernel}(b).

\subsubsection{Simulation results}

We take $L_1 = L_2 = 1$, so the reachable annulus has outer radius $2$. The reference circle is centered at $(x_c,y_c) = (1.2,\,0)$ with radius $r = 0.85$; its far point $(2.05,\,0)$ lies $0.05$ beyond maximum reach, so once per revolution the target is unreachable over a short arc and the arm is driven onto full extension $q_2 = 0$ --- a genuine rank-$1$ singularity. The trajectory is run for two revolutions ($\omega = 2\pi/10$, $T = 20$~s) so the singularity is met twice, at $t \approx 5$ and $t \approx 15$, with gain $k = 100$ and a fixed-step RK4 integrator ($\Delta t = 5\times 10^{-4}$). This exact reference and gain are reused in the next subsection under a different kernel orientation, so that the contrast reported there is attributable to the \emph{policy} alone and not to any change of trajectory. Table~\ref{tab:arm} collects the measured quantities and matches each to the corresponding statement of Theorem~\ref{thm:crossing}.

\begin{table}[h]
\centering
\begin{tabular}{l l l}
\hline
Quantity & Measured & Predicted / interpretation \\
\hline
$\min_t \sigma_{\min}(\mathbf{J})$ & $2.0\times 10^{-9}$ & genuine singularity is reached \\
$\min_t \rho$ (authority) & $1.2\times 10^{-9}$ & $\rho\to 0$ at singularity, Lemma~\ref{lem:kernel}(b) \\
$w^T b$ at singularity & $1.63$ & transversal, Def.~\ref{def:transversal} (not Prop.~\ref{prop:failure}) \\
$\max_t \|\dot q\|$ (this method) & $2.9$ & bounded velocity, Thm.~\ref{thm:crossing}(1) \\
$\max_t \|\mathbf{J}^{+}\dot y^*\|$ (naive) & $2.2\times 10^{7}$ & resolved-rate blow-up $\sim 1/\sigma_{\min}$ \\
$|e|$ in regular region & $7.5\times 10^{-3}$ & $\approx 1/k = 10^{-2}$, Thm.~\ref{thm:crossing}(2) \\
$|e|$ peak at crossing & $2.9\times 10^{-2}$ & bounded excursion, Thm.~\ref{thm:crossing}(3) \\
$|e|$ final (re-lock) & $9.8\times 10^{-3}$ & recovery, Thm.~\ref{thm:crossing}(4) \\
output error, regular / singular arc & $7.0\times 10^{-3}$ / $5.0\times 10^{-2}$ & singular value $\approx$ reach gap $0.05$ \\
\hline
\end{tabular}
\caption{2-DOF arm at the maximum-reach singularity, $\rho\ge0$ orientation, $r=0.85$, $k=100$.}
\label{tab:arm}
\end{table}

\begin{center}
\ifpdf
  % this figure is generated by sim_2dof.py
  \includegraphics[width=1\textwidth]{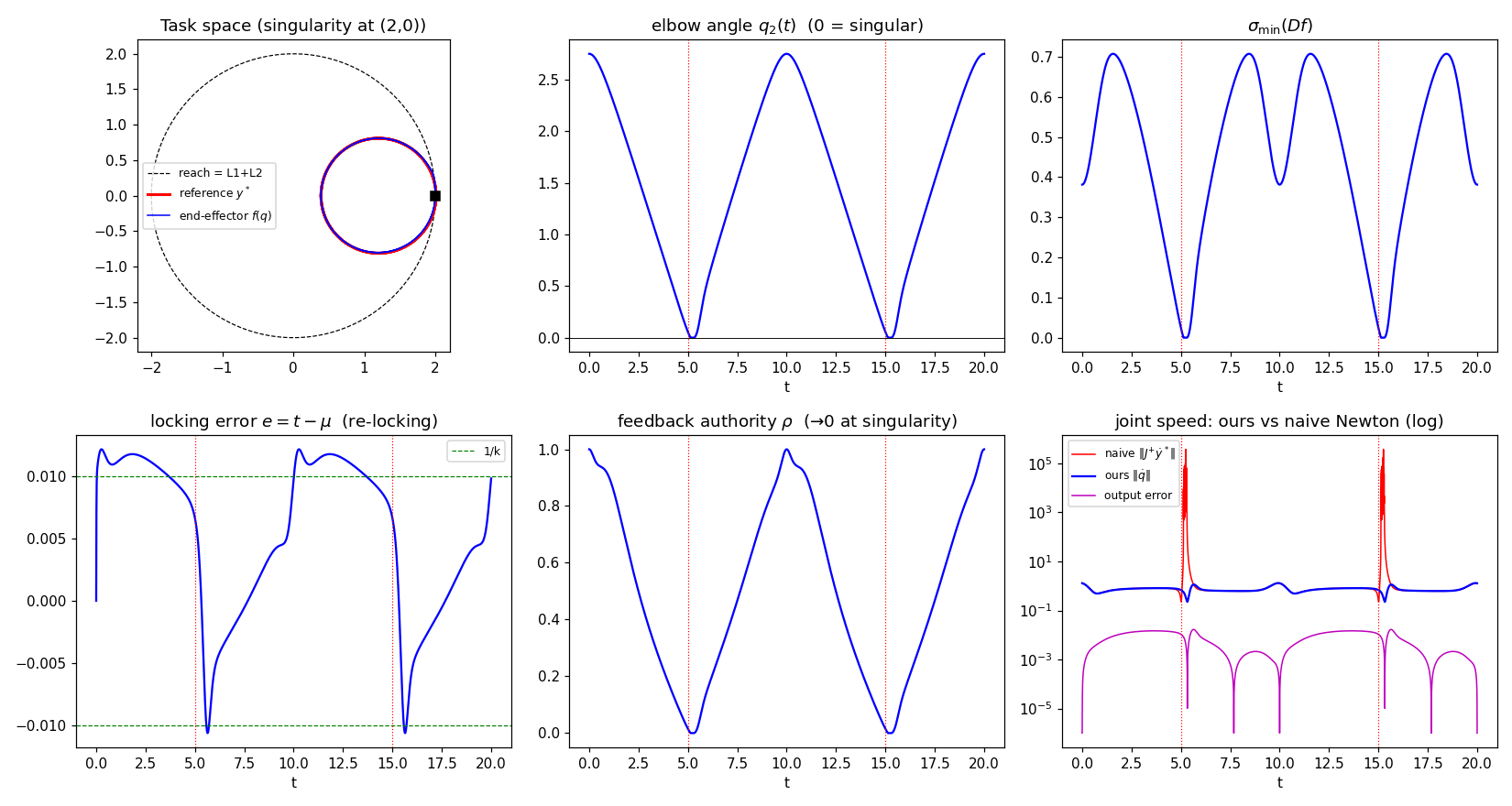}
\fi

Figure 6. Continuation tracking of the 2-DOF arm through the maximum-reach singularity (dotted red lines mark the singular instants $t\approx 5,15$). Top row: task-space path against the reference and the reach circle; elbow angle $q_2(t)$ touching $0$; smallest singular value $\sigma_{\min}(\mathbf{J})$ collapsing at the crossing. Bottom row: locking error $e = t-\mu$ inside the $\pm 1/k$ band with a bounded excursion at the crossing; feedback authority $\rho \to 0$ exactly at the singularity; joint speed of the continuation flow (bounded) against the naive resolved-rate demand $\|\mathbf{J}^{+}\dot y^*\|$ (log scale, spiking to $\sim 2\times 10^7$). This figure is the $\rho\ge0$ (reflecting) policy on the $r=0.85$ reference.
\end{center}

The results confirm the mechanism analyzed in Section~\ref{sec:crossing}. The augmentation converts a hard singularity of $\mathbf{J}$ into a regular point of $A$ precisely because the homotopy direction $b = f(q_0) - y^*$ retains a component along the uncontrollable mode $w$ ($w^T b = 1.63$). At the singular instant the parameter feedback goes limp ($\rho \to 0$), so $\mu$ cannot be actively corrected; but because the crossing is transversal, $\mu$ drifts by only about $3\times 10^{-2}$ and re-locks to time as soon as $\sigma_{\min}(\mathbf{J})$ recovers. During the unreachable arc the end effector rides the reach boundary with output error equal to the reach gap; by Remark~\ref{rem:leanstrength} this pinning is the \emph{unique} Filippov behaviour of the ideal closed loop while the one-sided fields straddle the singular surface --- forced dynamics, not an artifact of the initialization or of the integrator. Crucially, the joint velocities stay below $3$~rad/s throughout, whereas the naive resolved-rate law $\dot q = \mathbf{J}^{+}\dot y^*$ evaluated along the same path demands $\sim 2\times 10^{7}$~rad/s at the crossing --- seven orders of magnitude larger. Unlike damped least squares, this boundedness is obtained from the geometry of the augmented system rather than from an added damping term, and it costs no accuracy away from the singularity, where the tracking error sits at the $O(1/k)$ level.

We emphasize that this example exercises a \emph{boundary} (maximum-reach) singularity, where $q_2$ touches zero and returns on the same elbow branch; it validates parts (1)--(4) of Theorem~\ref{thm:crossing} but does not perform a full elbow reversal. Whether the flow can instead \emph{cross} the branch ($q_2$ passing from positive to negative) is examined next.

\subsubsection{Reflection versus crossing}

We now illustrate the dichotomy of Section~\ref{sec:reflection} empirically, using the \emph{same} reference circle and gain ($r=0.85$, far point $(2.05,0)$, $k=100$) and changing \emph{only} the orientation of $\hat n$, so the contrast is attributable to the policy alone. Since both elbow assemblies realize the same end-effector position, the reference does not fix the branch; the orientation does (Proposition~\ref{prop:reach}). Table~\ref{tab:flip} contrasts the two runs: the $\rho\ge0$ orientation reflects (stays $q_2\ge0$) with the additive excursion of Proposition~\ref{prop:excursion}(a) --- joint speed $2.9$~rad/s, $|e|$ re-locking to $9.8\times10^{-3}$; the continuity orientation crosses ($q_2:+\pi\to-\pi$) with the multiplicative excursion of Proposition~\ref{prop:excursion}(b) --- $\rho$ dips to $-0.51$, joint speed spikes to $\sim240$~rad/s, and locking is transiently lost before recovering. This is the discretionary-crossing case (both branches feasible), so which run is "correct" is set by a secondary criterion, not by tracking.

\begin{table}[h]
\centering
\begin{tabular}{l c c}
\hline
 & $\rho \ge 0$ (Thm.~\ref{thm:crossing}) & continuity (arc-length) \\
\hline
elbow behaviour & reflects (stays $q_2 \ge 0$) & \textbf{crosses} $+\pi \to -\pi$ \\
$\min \sigma_{\min}(\mathbf{J})$ & $2\times 10^{-9}$ & $8\times 10^{-5}$ \\
$\rho$ at singularity & $\to 0^{+}$ (stays $\ge 0$) & passes through $0$, dips to $-0.51$ \\
max $\|\dot q\|$ & $2.9$ & $240$ \\
$\mu$-locking & preserved, $|e|\!\to\!10^{-2}$ & lost then recovered \\
\hline
\end{tabular}
\caption{Reflection versus crossing at the maximum-reach singularity, same reference and gain $k=100$, differing only in the kernel orientation.}
\label{tab:flip}
\end{table}

\begin{center}
\ifpdf
  % this figure is generated by sim_2dof_flip.py
  \includegraphics[width=1\textwidth]{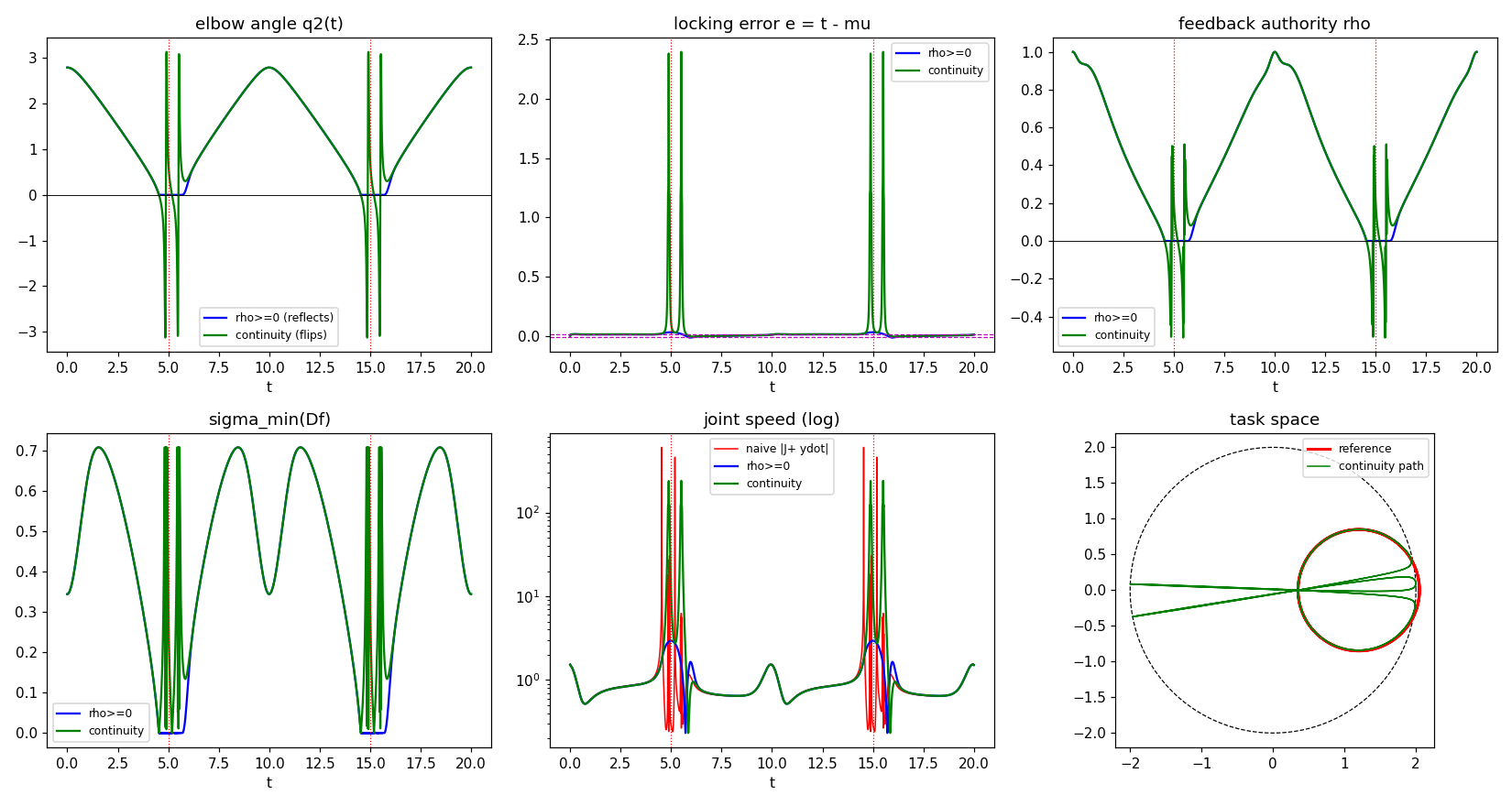}
\fi

Figure 7. Reflection ($\rho\ge0$, blue) versus branch crossing (continuity, green) at the maximum-reach singularity ($t\approx 5,15$). The reflecting orientation keeps $q_2 \ge 0$ and $\rho \ge 0$ with bounded joint speed; the continuity orientation flips $q_2$ across zero but drives $\rho$ negative and produces a large velocity transient.
\end{center}

The run confirms Proposition~\ref{prop:excursion} directly: the crossing transient is largest exactly where the transversality margin is smallest ($|w^Tb|$ collapses toward $0.02$ near the flip), i.e.\ where $\delta_c$ fails to shrink and the multiplicative factor $e^{2k\bar\rho\delta_c}$ grows --- the finite-$k$ shadow of the divergence in Proposition~\ref{prop:reach}(i). By the guidance of Section~\ref{sec:orientation} this is a discretionary crossing (both branches feasible); the practical detector (growth of $|e|$) is what a supervisor would monitor to decide when a crossing is instead forced.

\subsection{Dynamic relative-degree-one example}

We finally test the dynamic generalization of Section~\ref{sec:dynamic} on a genuinely dynamic plant, comparing the continuation controller against exact feedback linearization at a control (decoupling) singularity. Consider the relative-degree-one system with $n = 2$, $m = 1$,
\begin{equation}\label{dyn_example}
\dot x_1 = -x_1 + x_2 + u, \qquad \dot x_2 = -x_2 + \tfrac12 x_1, \qquad y = h(x) = x_1^3 - x_1 .
\end{equation}
Here $D h = (3x_1^2 - 1,\,0)$, $g = (1,0)^T$, so the decoupling matrix is the scalar $\Lambda(x) = 3x_1^2 - 1$ and the output drift is $a(x) = (3x_1^2-1)(-x_1+x_2)$. The control singularity $\Lambda = 0$ occurs at $x_1 = \pm 1/\sqrt3 \approx \pm 0.577$, where feedback linearization $u = (v - a)/\Lambda$ diverges. The internal state $x_2$ obeys $\dot x_2 = -x_2 + \tfrac12 x_1$, which is input-to-state stable with respect to $x_1$, so the system is minimum phase and Assumption~\ref{as:minphase} holds.

The reference $y^*(t) = 0.45\sin t$ pokes just past the fold value $0.385$ attained at the singularity, so the output is driven onto $\Lambda = 0$ once per half-cycle. We run the continuation controller \eqref{flow_dynamic} with $\rho \ge 0$ orientation, $k = 100$, against exact feedback linearization with a stabilizing outer loop $v = \dot y^* - \lambda(y - y^*)$, $\lambda = 5$. Table~\ref{tab:dyn} reports the outcome.

\begin{table}[h]
\centering
\begin{tabular}{l c c}
\hline
Quantity & Continuation & Feedback linearization \\
\hline
$\min_t |\Lambda|$ reached & \multicolumn{2}{c}{$1.0\times 10^{-8}$ (singularity is met)} \\
control effort $\max_t |u|$ & $15$ (bounded) & $1.3\times 10^{7}$ demand (diverges) \\
output tracking $|y-y^*|$, regular & $9.8\times 10^{-3}\;(\approx 1/k)$ & fails ($\gg 1$) \\
output tracking, at fold & $6.5\times 10^{-2}$ (fold gap) & --- \\
locking error $|e|$, regular / final & $9.9\times 10^{-3}$ / $8.4\times 10^{-3}$ & --- \\
internal state $\max_t |x_2|$ & $0.22$ (bounded) & --- \\
\hline
\end{tabular}
\caption{Dynamic relative-degree-one system \eqref{dyn_example} at the decoupling singularity $\Lambda = 0$, $k = 100$.}
\label{tab:dyn}
\end{table}

\begin{center}
\ifpdf
  % this figure is generated by sim_dynamic_rd1.py
  \includegraphics[width=1\textwidth]{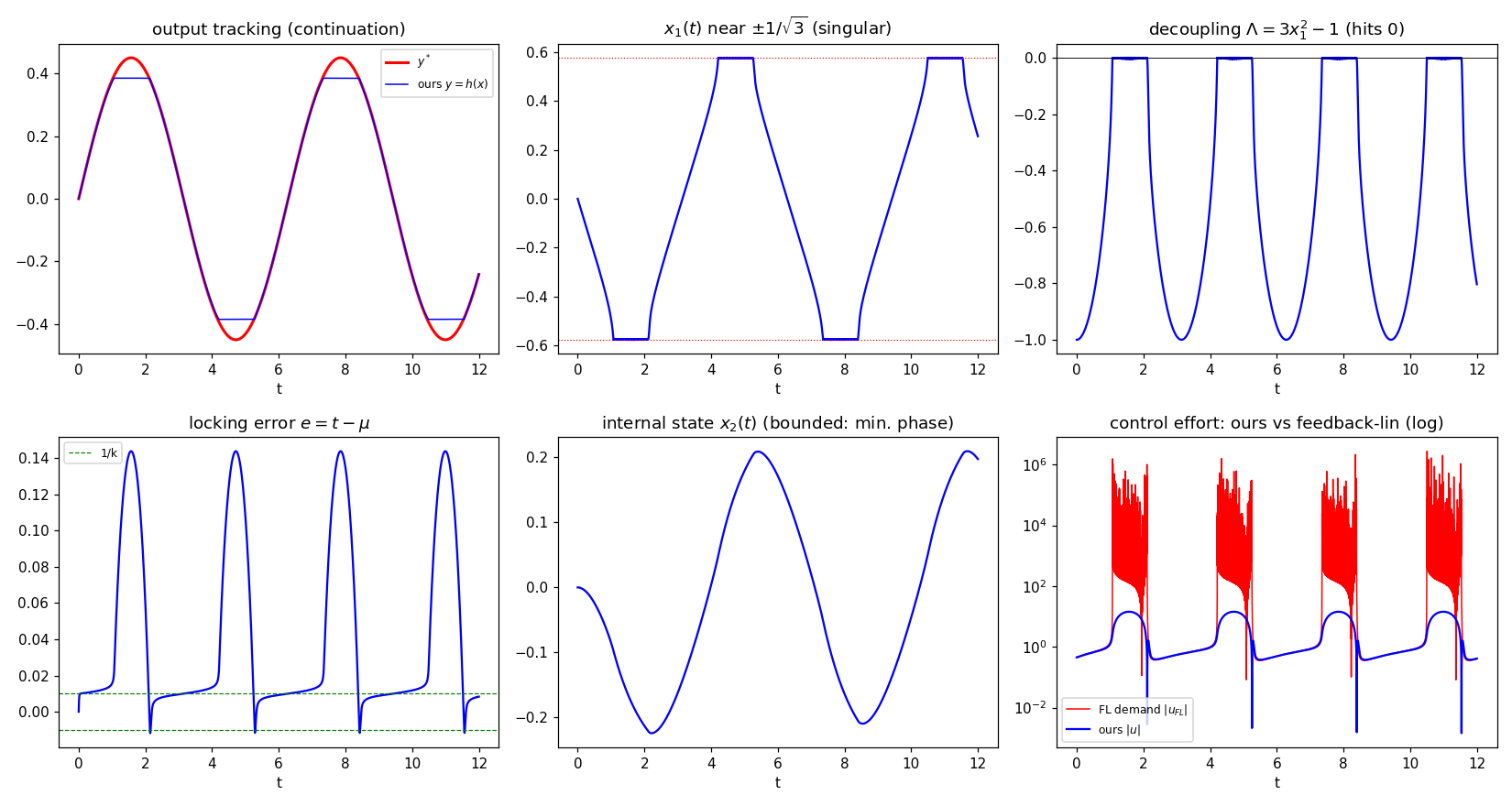}
\fi

Figure 8. Dynamic relative-degree-one tracking through the control singularity $\Lambda = 0$. Top row: output $y = h(x)$ against the reference; the state $x_1(t)$ grazing $\pm 1/\sqrt3$; the decoupling term $\Lambda = 3x_1^2 - 1$ collapsing to zero. Bottom row: locking error $e = t - \mu$ within the $\pm 1/k$ band; the internal state $x_2(t)$ bounded (minimum phase); control effort of the continuation flow (bounded, $\le 15$) against the feedback-linearization demand $|u_{FL}| = |(v-a)/\Lambda|$ (log scale, diverging to $\sim 10^7$).
\end{center}

The experiment realizes Proposition~\ref{prop:dyntracking}. At the decoupling singularity feedback linearization is undefined --- its control demand reaches $1.3\times 10^7$ --- whereas the continuation controller passes through with bounded effort $|u| \le 15$, tracks the output to $O(1/k)$ away from the singularity, and keeps the internal state bounded, confirming that minimum-phaseness carries the boundedness of $u$ over to boundedness of the full state. The residual output error at the fold, $6.5\times 10^{-2}$, is precisely the amount by which the reference exceeds the value reachable on the incoming branch --- the scalar analogue of the reach gap of the arm --- and $\mu$ re-locks to time once $\Lambda$ recovers.

Two caveats mirror the kinematic study. First, this is again the reflecting regime: the state remains on the incoming branch $|x_1| \le 1/\sqrt3$, which is exactly the behaviour guaranteed under the $\rho \ge 0$ orientation; tracking a reference that demands the far branch would require the crossing orientation and the supervisory logic of Section~\ref{sec:orientation}. Second, the output identity $h(x) - y^* = e\,b$ was observed to hold up to a residual $\max_t|H| \approx 1.5\times 10^{-2}$ that halves as the integration step is halved: it is a numerical artifact of the orientation kink at $\rho = 0$ (where $u$ is discontinuous) rather than a property of the controller, and it motivates a smoothed orientation near the singularity as a practical refinement.

\subsection{Relative-degree-two example}\label{sec:rd2}

To test the reduction of Section~\ref{sec:reldeg} we take a genuinely second-order plant ($n=3$, $m=1$, relative degree two, one internal state),
\begin{equation}\label{rd2_example}
\dot x_1 = x_2, \quad \dot x_2 = \alpha(x) + \beta(x)\,u, \quad \dot z = -z + \tfrac12 x_1, \quad y = x_1,
\end{equation}
with $\alpha(x) = -x_1 - x_2 + 0.3\,z$ and $\beta(x) = 3x_1^2 - 1$. The decoupling term $\beta$ vanishes at $x_1 = \pm 1/\sqrt3$, where the relative degree is lost and feedback linearization $u = (v-\alpha)/\beta$ diverges. The internal state $z$ satisfies $\dot z = -z + \tfrac12 x_1$, ISS in $x_1$, so the system is minimum phase. We form the filtered error $\sigma = (\dot y - \dot y^*) + \lambda(y - y^*)$ with $\lambda = 5$ (relative degree one), and regulate it with the continuation controller ($k=100$), against exact feedback linearization with Hurwitz error dynamics ($c_0 = 25$, $c_1 = 10$). The reference $y^*(t) = 0.7\sin t$ drives $y = x_1$ back and forth across $\pm 1/\sqrt3$.

\begin{table}[h]
\centering
\begin{tabular}{l c c}
\hline
Quantity & Continuation & Feedback linearization \\
\hline
$\min_t |\beta|$ reached & \multicolumn{2}{c}{$2.6\times 10^{-6}$ (singularity crossed)} \\
control effort $\max_t |u|$ & $59$ (bounded) & $2.7\times 10^{5}$ demand (diverges) \\
output tracking $|y-y^*|$, regular / max & $1.8\times 10^{-3}$ / $7.4\times 10^{-3}$ & fails \\
filtered error $|\sigma|$, regular & $6.1\times 10^{-3}$ ($\to 0$) & --- \\
locking error $|e|$, regular / final & $8.8\times 10^{-3}$ / $1.5\times 10^{-2}$ & --- \\
internal state $\max_t |z|$ & $0.26$ (bounded) & --- \\
identity residual $\max_t |H|$ & $2.5\times 10^{-11}$ & --- \\
\hline
\end{tabular}
\caption{Relative-degree-two system \eqref{rd2_example} crossing the decoupling singularity $\beta = 0$ via the filtered-error reduction, $k = 100$, $\lambda = 5$.}
\label{tab:rd2}
\end{table}

\begin{center}
\ifpdf
  % this figure is generated by sim_dynamic_rd2.py
  \includegraphics[width=1\textwidth]{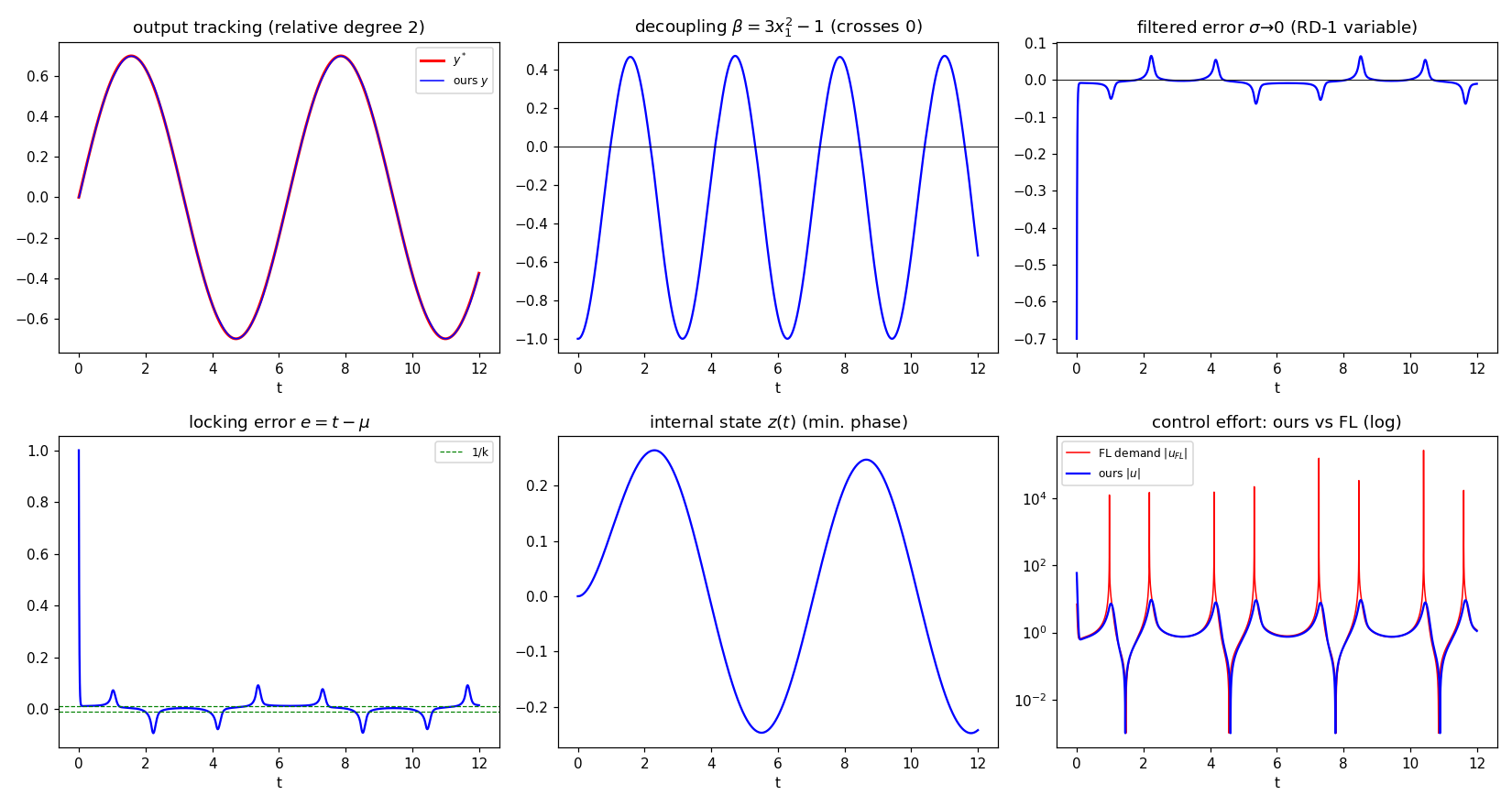}
\fi

Figure 9. Relative-degree-two tracking through the control singularity $\beta = 0$ via the filtered-error reduction. Top row: output $y$ against the reference; the decoupling term $\beta = 3x_1^2 - 1$ crossing zero; the filtered error $\sigma$ (the relative-degree-one variable) driven to zero. Bottom row: locking error $e = t-\mu$ within the $\pm 1/k$ band; the internal state $z(t)$ bounded (minimum phase); control effort of the continuation flow (bounded, $\le 59$) against the feedback-linearization demand (log scale, diverging to $\sim 3\times 10^5$).
\end{center}

The result confirms Corollary~\ref{cor:rdr}. Unlike the reflecting behaviour of the first-order example, here the output map $y = x_1$ is monotone, so the state genuinely \emph{crosses} the singularity ($|y|$ reaches $0.70 > 1/\sqrt3$): the momentum $x_2$ carries $x_1$ through $\beta = 0$, where the control is instantaneously ineffective, while the continuation flow keeps $u$ bounded ($\le 59$, against a feedback-linearization demand of $2.7\times 10^5$). The filtered error $\sigma$ is regulated to $O(1/k)$, so the output tracks to $1.8\times 10^{-3}$; the internal state stays bounded, confirming minimum-phaseness; and the identity residual is essentially exact ($10^{-11}$), because the constant anchor $b_\sigma = \sigma_0$ dominates the augmented kernel and suppresses the orientation kink that afflicted the first-order run. This is the cleanest of the three experiments: a genuine passage through a control singularity of a relative-degree-two system, with bounded effort and accurate tracking, where feedback linearization is undefined.

\section{Application: DC/DC resonant power converter}\label{sec:converter}

The preceding examples were designed to expose the mechanism in the cleanest possible settings. We close with a genuine engineering plant whose feedforward inversion is singular in two physically distinct ways, and we show that the continuation controller of Sections~\ref{sec:tracking}--\ref{sec:reldeg} inverts it through both. The plant is the dual-bridge series resonant DC/DC converter (DB SRC), whose first-harmonic (FHA) model and closed-form feedforward inversion were developed in \cite{Borisevich2020SRC}. Instead of a dynamic (relative-degree) model we use the converter's steady-state input--output characteristic directly: this is the kinematic / root-finding case of Sections~\ref{sec:wellbehaved}--\ref{sec:tracking}, in which the manipulated quantities are the modulation parameters and the ``output map'' $f$ is the steady-state characteristic. The controller commands the \emph{rates} of the modulation parameters ($\dot u$), exactly as it commanded joint rates for the arm of Section~\ref{sec:arm}.

\subsection{Plant model}

The DB SRC drives a series $LC$ tank (leakage inductance $L$, resonant capacitor $C$, transformer ratio $n$) from two actively controlled full bridges. In switching-frequency angle $t' = \omega t$ the modulation is described by four variables: primary duty $d$, secondary shorting time $s$, phase shift $\beta$, and switching frequency $\omega$. The voltage gain is $G = n V_{out}/V_{in}$, the tank resonant frequency $\omega_0 = 1/\sqrt{LC}$, characteristic impedance $Z_0 = \sqrt{L/C}$, and normalized frequency $F_n = \omega/\omega_0$.

Following \cite{Borisevich2020SRC}, the buck and boost regimes are merged into a single \emph{aggregated} modulation variable $q \in [0, 2\pi]$,
\begin{equation}\label{eq:pe_q}
(d,s) = \begin{cases} (q,\,0), & q \le \pi \quad\text{(buck: duty $d=q$, no shorting)},\\[2pt] (\pi,\,q-\pi), & q > \pi \quad\text{(boost: full duty, shorting $s=q-\pi$)}. \end{cases}
\end{equation}
This aggregation is itself a hand-built homotopy that stitches the two modes continuously across the transition $d=\pi$; our construction generalizes it. Taking the secondary shorting offset $s_{add}=0$, the modulation reduces to the three quantities $u=(q,\beta,\omega)$.

The FHA steady-state characteristic maps $u$ to two commutation angles $(\sigma,\delta)$ (the ZVS timing quantities) and the output transconductance $W = I_{out}/V_{in}$:
\begin{align}
A(q,\beta) &= 4\sin d + 4G\bigl[\sin(\beta+s) + \sin\beta\bigr], \label{eq:pe_A}\\
B(q,\beta) &= 4 - 4\cos d - 4G\bigl[\cos(\beta+s) + \cos\beta\bigr], \label{eq:pe_B}\\
\sigma &= \operatorname{atan2}(B, A), \qquad \delta = \beta - \sigma, \label{eq:pe_sigdel}\\
W &= \frac{n}{2\pi^2 Z(\omega)}\sqrt{A^2+B^2}\,\bigl[\cos(\delta+s) + \cos\delta\bigr], \label{eq:pe_W}
\end{align}
with $(d,s)$ given by \eqref{eq:pe_q} and the tank impedance
\begin{equation}\label{eq:pe_Z}
Z(\omega) = \sqrt{R^2 + \bigl(\omega L - 1/\omega C\bigr)^2}, \qquad Z(\omega)\ \text{minimized at}\ \omega = \omega_0 .
\end{equation}
The output vector is $y = (\sigma, \delta, W)$ and the control is $u = (q,\beta,\omega)$. Experiment~1 below fixes $\omega$ and works with the $2\times 2$ sub-problem $u=(q,\beta)\mapsto y=(\sigma,\delta)$.

\subsection{Two decoupling singularities and reduction to the framework}

The model has a decisive structural feature. From \eqref{eq:pe_sigdel}, $\delta = \beta - \sigma$, hence $\sigma+\delta=\beta$ identically; and $\sigma,\delta$ do not depend on $\omega$. The $3\times 3$ Jacobian is therefore block-triangular,
\begin{equation}\label{eq:pe_jac}
J \;=\; \frac{\partial(\sigma,\delta,W)}{\partial(q,\beta,\omega)}
\;=\;
\begin{pmatrix}
\sigma_q & \sigma_\beta & 0 \\
-\sigma_q & 1-\sigma_\beta & 0 \\
W_q & W_\beta & W_\omega
\end{pmatrix},
\qquad
\det J \;=\; \frac{\partial\sigma}{\partial q}\,\frac{\partial W}{\partial\omega},
\end{equation}
using $\det\!\begin{pmatrix}\sigma_q & \sigma_\beta\\ -\sigma_q & 1-\sigma_\beta\end{pmatrix} = \sigma_q$. The determinant factorizes into two independent mechanisms, giving two decoupling singularities:
\begin{itemize}
\item[(S1)] \emph{Buck/boost fold}, $\partial\sigma/\partial q = 0$. On the synchronous-rectification line $\delta=0$ it sits at $\cos\sigma = G$, i.e.\ $\sigma = \arccos G$, which coincides with the mode transition $q=\pi$. There the closed-form inverse $s_{\min} = \arccos(2\cos\sigma/G - \cos\delta) - \delta$ has $dq/d\sigma \to \infty$: the analytical inversion breaks down.
\item[(S2)] \emph{Resonance}, $\partial W/\partial\omega = 0$. Since $W \propto 1/Z(\omega)$ and $Z$ is minimized at $\omega=\omega_0$ \eqref{eq:pe_Z}, $W$ attains its maximum and the frequency input loses all authority: the entire $\omega$-column of $J$ vanishes.
\end{itemize}

The reduction to the continuation framework is immediate. With $f$ the steady-state map \eqref{eq:pe_A}--\eqref{eq:pe_W} and $D f = J$, the controller forms the augmented matrix $A = \bigl[\, J(u)\;\big|\;b(t)\,\bigr]$, $b(t) = f(u_0) - y^*(t)$, and integrates the flow \eqref{flow} in real time with $\dot t = 1$,
\begin{equation}\label{eq:pe_flow}
\begin{pmatrix}\dot u \\ \dot\mu\end{pmatrix} = -A^\dagger B + k\,e\,\hat n, \qquad e = t-\mu,\quad \rho = \hat n_{\text{last}} \ge 0 .
\end{equation}
This is precisely the kinematic controller of Section~\ref{sec:tracking} with joint rates replaced by modulation-parameter rates; the dynamic generalization of Section~\ref{sec:dynamic} is not required because the manipulated variables enter the steady-state map directly. Lemma~\ref{lem:kernel}, Definition~\ref{def:transversal}, Theorem~\ref{thm:crossing} and Proposition~\ref{prop:failure} apply verbatim with $Df=J$.

\subsection{Experiment 1: crossing the buck/boost singularity, $(q,\beta)\to(\sigma,\delta)$}

We fix $\omega$ and invert $u=(q,\beta)\mapsto y=(\sigma,\delta)$ at $G=0.8$, $k=100$. The reference holds $\delta^*=0$ (synchronous rectification) and sweeps $\sigma^*(t) = \sigma_x + a\cos(\omega_r t)$ with $\sigma_x = \arccos G = 0.6435$, $a=0.28$, period $10$~s, so the target crosses the transition (S1) four times over $T=20$~s and $q$ is driven back and forth across $\pi$. Crucially, the homotopy is anchored at an \emph{extremum}, $\sigma^*(0)=\sigma_x + a$, not at $\sigma_x$: this keeps $b = f(u_0)-y^*$ with a nonzero component along the uncontrollable mode $w$ at the crossing. Anchoring at $\sigma_x$ makes $b\to 0$ exactly at the singular instant, i.e.\ a \emph{non-transversal} crossing in the sense of Proposition~\ref{prop:failure}; in that (mistaken) setup the tracking error grew and the commanded rate jumped an order of magnitude --- the observable signature of lost transversality. Table~\ref{tab:pe1} reports the transversal run.

\begin{table}[h]
\centering
\begin{tabular}{l l l}
\hline
Quantity & Measured & Predicted / interpretation \\
\hline
$\min_t \sigma_{\min}(J)$ & $1.2\times 10^{-4}$ & buck/boost fold (S1) is reached \\
$\min_t \rho$ & $4.2\times 10^{-4}$ & $\rho\to 0$ at singularity, Lemma~\ref{lem:kernel}(b) \\
$|w^T b|$ at singularity & $0.27$ & transversal, Def.~\ref{def:transversal} \\
regular-region $|e|=|t-\mu|$ & $8.5\times 10^{-3}$ & $\approx 1/k$, Thm.~\ref{thm:crossing}(2) \\
$|e|$ peak at crossing & $3.3\times 10^{-2}$ & bounded excursion, Thm.~\ref{thm:crossing}(3) \\
$|e|$ final (re-lock) & $1.0\times 10^{-2}$ & recovery, Thm.~\ref{thm:crossing}(4) \\
$\max_t \|\dot u\|$ (ours) & $3.3$ & bounded, Thm.~\ref{thm:crossing}(1) \\
$\max_t \|\tfrac{d}{dt}(q,\beta)\|$ (invert\_ds) & $7.3\times 10^{2}$ & analytical inverse rate diverges \\
regular tracking error & $\le 1.1\times 10^{-2}$ & $O(1/k)$ off the singularity \\
\hline
\end{tabular}
\caption{Experiment 1: $(q,\beta)\to(\sigma,\delta)$ across the buck/boost fold, $G=0.8$, $k=100$.}
\label{tab:pe1}
\end{table}

\begin{center}
\ifpdf
  % this figure is generated by sim_continuation_qbeta.py
  \includegraphics[width=1\textwidth]{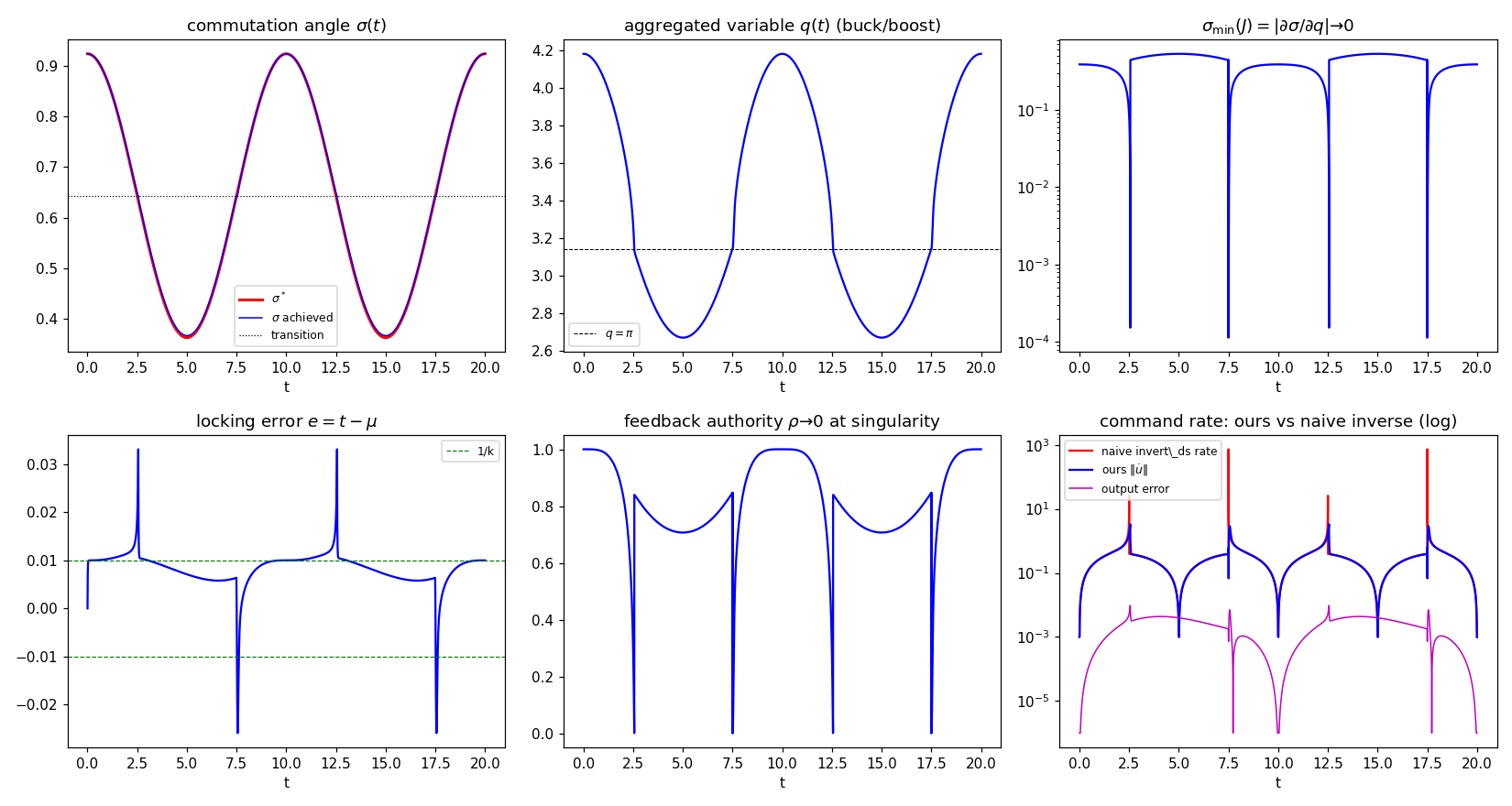}
\fi

Figure 10. Continuation inversion across the buck/boost singularity. Top: $\sigma(t)$ tracking the reference through the transition $\sigma_x=\arccos G$; the aggregated variable $q(t)$ crossing $\pi$; the collapse of $\sigma_{\min}(J)=|\partial\sigma/\partial q|$. Bottom: locking error $e=t-\mu$ in the $\pm 1/k$ band; feedback authority $\rho\to 0$ at the crossing; commanded rate of the continuation flow (bounded) against the analytical inverse \texttt{invert\_ds} whose rate spikes to $\sim 7\times 10^2$.
\end{center}

This is a genuine \emph{crossing} (the output moves monotonically in $\sigma$ as $q$ passes through $\pi$), analogous to the relative-degree-two example of Section~\ref{sec:rd2}. The aggregated variable $q$ is continuous through the mode change, but the map \emph{from the output to $q$} has infinite slope at the transition --- exactly the fold $\partial\sigma/\partial q = 0$ --- and it is there that the closed-form inverse's commanded rate diverges while the continuation flow crosses with bounded rate.

\subsection{Experiment 2: square commutation-plus-power inversion, $(q,\beta,\omega)\to(\sigma,\delta,W)$}

We now use all three controls to regulate all three outputs. The converter parameters are those of the source paper, $L=31~\mu$H, $C=8.2$~nF, $n=2.2$ ($f_0 = 315.7$~kHz, $Z_0 = 61.5~\Omega$), with tank damping $R = Z_0/6$ (quality factor $\approx 6$) and $F_{n,0}=1.3$. The reference sweeps $\sigma^*$ across the transition as in Experiment~1 while simultaneously demanding a $\pm 12\%$ power swing, $W^*(t) = W_0\,(1 + 0.12\sin\omega_r t)$. Because the run stays above resonance, $\partial W/\partial\omega \ne 0$ throughout, so by \eqref{eq:pe_jac} the only singularity met is the buck/boost fold (S1) --- now embedded in a square $3\times 3$ system where $\omega$ moves only to hold power. Table~\ref{tab:pe2}.

\begin{table}[h]
\centering
\begin{tabular}{l l l}
\hline
Quantity & Measured & Interpretation \\
\hline
$q$ crosses $\pi$; $F_n$ range & yes; $[1.16,1.34]$ & fold met, $\omega$ stays above resonance \\
$\min_t \sigma_{\min}(J)$ & $1.4\times 10^{-4}$ & fold (S1) reached \\
$\min_t \rho$ & $4.8\times 10^{-4}$ & authority $\to 0$, Lemma~\ref{lem:kernel}(b) \\
$|w^T b|$ at singularity & $0.27$ & transversal, Def.~\ref{def:transversal} \\
regular $\|y-y^*\|$ (max) & $1.4\times 10^{-2}$ & $O(1/k)$; $W$ rel.\ error $4\times 10^{-3}$ \\
$|e|$ regular / peak / final & $8.6\!\cdot\!10^{-3}$ / $3.3\!\cdot\!10^{-2}$ / $1.0\!\cdot\!10^{-2}$ & Thm.~\ref{thm:crossing}(2)--(4) \\
$\max_t \|\dot u\|$ (ours) & $3.3$ & bounded, Thm.~\ref{thm:crossing}(1) \\
$\max_t \|J^{-1}\dot y^*\|$ (naive FL) & $1.3\times 10^{3}$ & feedback linearization diverges \\
\hline
\end{tabular}
\caption{Experiment 2: square $(q,\beta,\omega)\to(\sigma,\delta,W)$ across the buck/boost fold while regulating power, $k=100$.}
\label{tab:pe2}
\end{table}

\begin{center}
\ifpdf
  % this figure is generated by sim_continuation_qbw.py
  \includegraphics[width=1\textwidth]{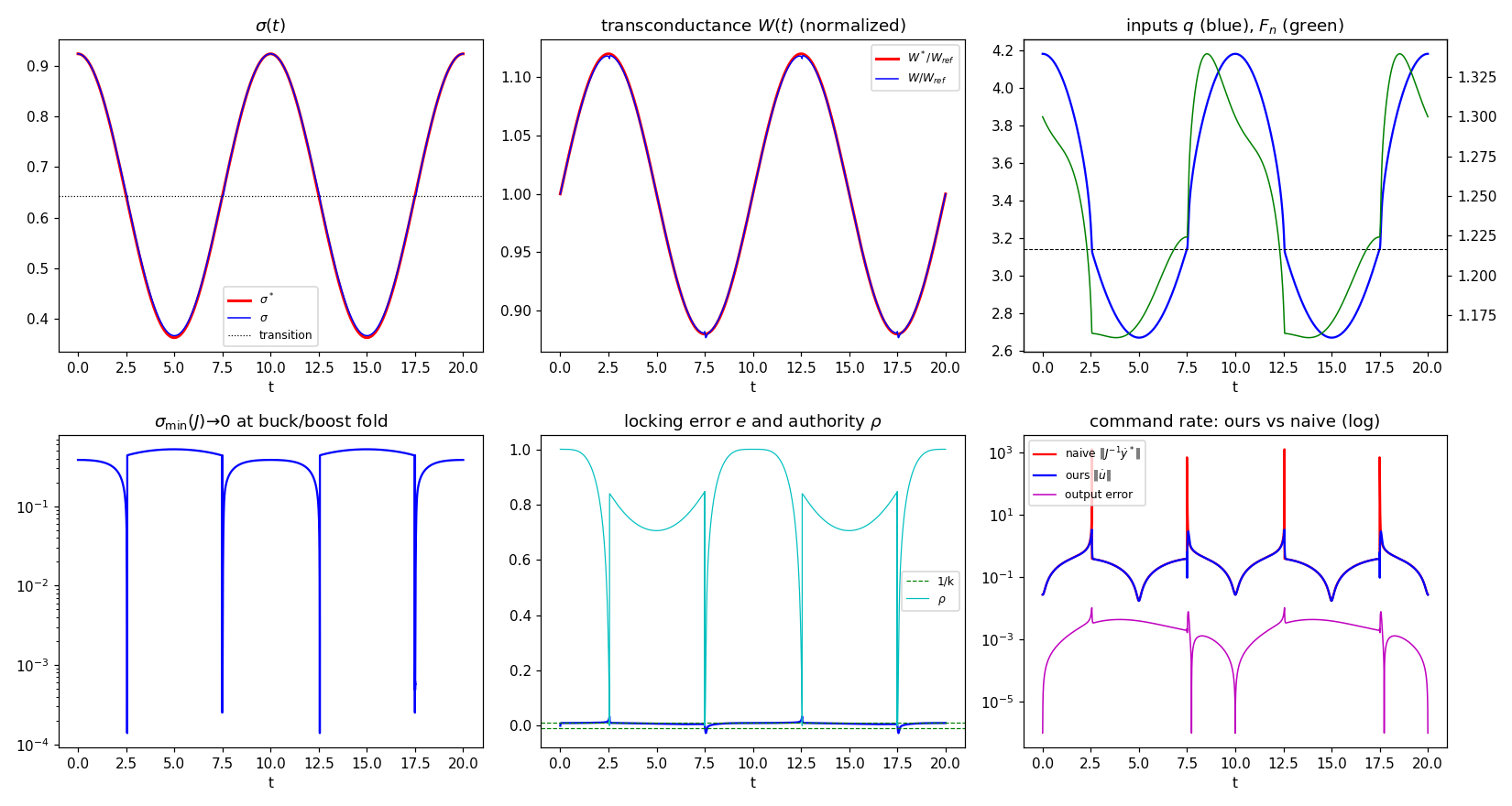}
\fi

Figure 11. Square commutation-plus-power inversion. The controller holds $\sigma$ and the normalized power $W/W_{ref}$ on their references while $q$ crosses $\pi$ and $F_n$ moves only to track power; $\sigma_{\min}(J)$ collapses at the fold; the locking error stays in the $\pm 1/k$ band and re-locks; the continuation rate (bounded) contrasts with the feedback-linearization demand $\|J^{-1}\dot y^*\|\sim 10^3$.
\end{center}

The three outputs are held to $O(1/k)$ throughout the crossing. Since $\det J = (\partial\sigma/\partial q)(\partial W/\partial\omega)$ and the second factor is bounded away from zero above resonance, the rank drop is entirely attributable to the fold factor, and the augmented construction repairs it exactly as in the scalar-output case.

\subsection{Experiment 3: the resonance singularity, $\omega\to\omega_0$}

To exercise the second mechanism (S2) we isolate it: hold $\sigma^*=0.5$ (in buck, away from the transition) and $\delta^*=0$, and ramp the power demand from its $F_n=1.3$ value up through and $5\%$ past the resonant peak $W_{\text{peak}} = W(\omega_0)$. As $W^*$ rises the controller lowers $F_n$ toward $1$; at resonance $\partial W/\partial\omega \to 0$ and the frequency input goes limp. Table~\ref{tab:pe3}.

\begin{table}[h]
\centering
\begin{tabular}{l l l}
\hline
Quantity & Measured & Interpretation \\
\hline
$F_n$ reached & $0.9998$ & resonance $\omega_0$ effectively attained \\
$\min_t |\partial W/\partial\omega|$ & $7.3\times 10^{-8}$ & frequency authority collapses (S2) \\
$\min_t \sigma_{\min}(J)$ & $1.5\times 10^{-8}$ & genuine rank drop \\
$\min_t \rho$ & $2.9\times 10^{-8}$ & $\rho\to 0$, Lemma~\ref{lem:kernel}(b) \\
$|w^T b|$ near resonance & $0.50$ & transversal, Def.~\ref{def:transversal} \\
$W$ achieved max vs peak & $3.338$ vs $3.338$ & saturates at reach limit (reach gap) \\
$\max_t \|\dot u\|$ (ours) & $0.88$ & bounded, Thm.~\ref{thm:crossing}(1) \\
$\max_t \|J^{-1}\dot y^*\|$ (naive FL) & $2.9\times 10^{5}$ & diverges as $1/(\partial W/\partial\omega)$ \\
$|e|$ regular / final & $1.0\times 10^{-2}$ / $2.9\times 10^{-3}$ & $\approx 1/k$, then re-lock \\
\hline
\end{tabular}
\caption{Experiment 3: reflection at the resonance singularity $\partial W/\partial\omega=0$, $k=100$. Power normalized by $W_{ref}=W(F_n{=}1.3)$.}
\label{tab:pe3}
\end{table}

\begin{center}
\ifpdf
  % this figure is generated by sim_continuation_resonance.py
  \includegraphics[width=1\textwidth]{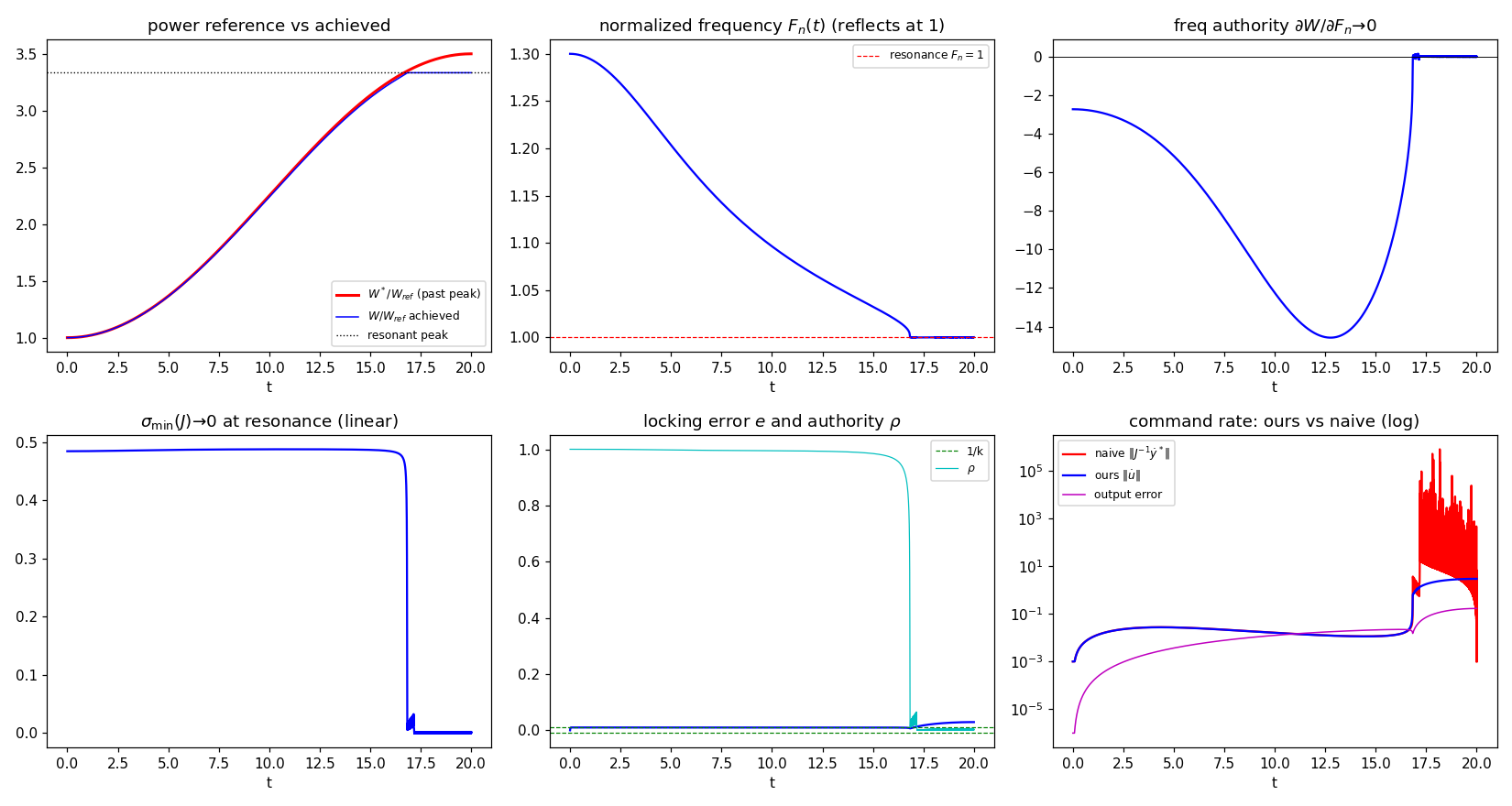}
\fi

Figure 12. Reflection at resonance. Top: the power reference is pushed past the resonant peak (dotted) but $W$ saturates \emph{exactly} at the peak; $F_n(t)$ descends to $1$ and reflects rather than crossing to $F_n<1$; the frequency authority $\partial W/\partial\omega$ collapses to zero. Bottom: $\sigma_{\min}(J)\to 0$; locking error and $\rho$; the bounded continuation rate against the feedback-linearization demand diverging to $\sim 3\times 10^5$.
\end{center}

Unlike Experiments~1--2 this is a \emph{reflection}, not a crossing: $W$ saturates precisely at its resonant maximum --- the reach-gap analogue of the arm at full extension (Section~\ref{sec:arm}) --- and $F_n$ reflects at $1$ under the $\rho\ge 0$ orientation instead of crossing to $F_n<1$. By Remark~\ref{rem:leanstrength}, riding the resonant power ceiling is the unique Filippov behaviour of the ideal closed loop here, so the observed saturation is forced dynamics rather than a numerical artifact. This is physically decisive. The sub-resonant region $F_n<1$ is capacitive, where zero-voltage switching is lost; the $\rho\ge 0$ orientation makes the controller ride the resonant power ceiling and \emph{refuse to leave the soft-switching region}, with a bounded frequency slew, whereas a naive power inverter would command an unbounded downward sweep through resonance (demand $\sim 3\times 10^5$). Deliberately crossing into $F_n<1$ would require the continuity orientation of Section~\ref{sec:orientation} together with a crossing-compatible reference, at the cost of transiting hard switching --- precisely the reflection-versus-crossing dichotomy established for the arm.

\subsection{Discussion}

The converter exhibits, in a single plant, both faces of Lemma~\ref{lem:kernel}: a transversal \emph{crossing} at the buck/boost fold and a transversal \emph{reflection} at the resonance reach-boundary, mirroring respectively the relative-degree-two crossing of Section~\ref{sec:rd2} and the maximum-reach reflection of the arm. The factorization $\det J = (\partial\sigma/\partial q)(\partial W/\partial\omega)$ shows the two singular mechanisms are independent --- the fold is a property of the $(q,\beta)\to(\sigma,\delta)$ sub-map, resonance a property of the $\omega\to W$ channel --- and the augmented construction repairs each without special-casing. Two engineering lessons follow directly from the theory. First, transversality is a \emph{reference-shaping} condition: anchoring the homotopy away from the mode boundary keeps $w^T b \ne 0$, and violating it reproduces the Proposition~\ref{prop:failure} failure. Second, the $\rho \ge 0$ orientation encodes the physically safe operating choice --- stay above resonance, on the incoming mode --- so the boundedness guarantee of Theorem~\ref{thm:crossing} and the correct engineering behaviour coincide.

Finally, this is the \emph{static-inversion} application (Sections~\ref{sec:wellbehaved}--\ref{sec:tracking}): the controller wraps the converter's steady-state characteristic and outperforms the closed-form and feedback-linearizing inverses at both singularities. Extending it to a genuine \emph{dynamic} averaged converter model --- where the relative-degree-one theory of Section~\ref{sec:dynamic} applies with $\Lambda = Dh\,g$, and where minimum-phaseness of the output-filter and load zero dynamics must be certified (and can fail in deep boost) --- is the natural next step.

\section{Future directions}\label{sec:future}

Three directions stand out. The first is largely worked out and near-term; the second builds on the multi-parameter theory of Section~\ref{sec:adaptive}; the third is a conjecture that we state carefully, as it concerns the ultimate scope of the approach. A single design quantity ties them together: the number $p$ of continuation parameters, which sets the dimension of the null space of the augmented Jacobian \eqref{eq:multiaug} and is spent jointly on outer-loop control, adaptation, and transversal repair of singularities (Theorem~\ref{thm:residual}).

\subsection{Arbitrary auxiliary controller in the null space}\label{sec:aux}

The proportional locking law $\gamma = k\,e$ is only the simplest choice. The null-space coefficient $\gamma$ is a virtual control for the scalar locking-error subsystem
\begin{equation}
\dot e = 1 - d(x,t) - \rho(x,t)\,\gamma,
\end{equation}
so \emph{any} stabilizing compensator may be substituted. Integral action, $\gamma = k_p e + k_i\!\int e$, drives $e \to 0$ exactly and removes the $O(1/k)$ bias of Theorem~\ref{thm:crossing}, giving asymptotically exact tracking away from singularities --- this case, together with the drift-feedforward variant, is now developed, machine-checked, and validated in Remark~\ref{rem:exactlocking} and Sections~\ref{sec:lean} and~\ref{sec:lockingsims}; lead--lag, $H_\infty$, adaptive, or nonlinear designs shape the transient and reject disturbances in the drift $d$ and gain $\rho$. This recovers, in the present tracking setting, the ``auxiliary controller as a drift term'' construction of the antecedent setpoint work, and it exposes a clean separation of concerns: singularity handling is intrinsic (through the augmentation and the authority $\rho$), while outer-loop performance is a free design in $\gamma$. The one structural limit is that $\rho \to 0$ at the singular instant caps every auxiliary controller identically there --- the bounded excursion of Theorem~\ref{thm:crossing}(3) is set by the geometry, not by the outer loop.

\subsection{Adaptive designs and higher-order degeneracies}

The structural theory of the multi-parameter augmentation --- the exact repair criterion, the kernel collapse, and the residual-authority count --- is established in Section~\ref{sec:adaptive}, and Section~\ref{sec:adaptivesim} demonstrates the resulting crossing of a regressor's singular manifold on a nonconvex identification problem. What remains open on this front is the full adaptive design: convergence of the joint continuation under persistent excitation (the repair criterion of Proposition~\ref{prop:multirepair} holding with uniform margin), mismatch robustness of the tracking bounds, and the interaction between the estimator's locking law and the fold transits.

A complementary direction concerns \emph{higher-order} (non-fold) degeneracies that remain corank one but violate the Whitney nondegeneracy \eqref{eq:foldnondeg} --- the cusp $Y=\xi^3$ and its successors in the Thom--Boardman hierarchy. There the local model is no longer quadratic, so the reflection/crossing dichotomy of Theorem~\ref{thm:reflection} acquires additional branches; characterizing the continuation flow through such points, and validating it on a Whitney cusp, is left to future work.

\subsection{Beyond a fixed morphism: flowing between linearizations (conjecture)}

Feedback linearization is a global object: it requires a linearizing morphism --- a diffeomorphism together with static or dynamic feedback --- whose existence is an integrability condition (involutivity in the static case, flatness in the dynamic case). The continuation controller, by contrast, requires only \emph{pointwise, path-local} regularity of the augmented map along the trajectory, $\rank[\Lambda \mid b] = m$, which is strictly weaker. We therefore conjecture that the homotopy viewpoint --- a continuous deformation \emph{between} linearizing morphisms, rather than commitment to a single one --- can regulate references and operating regions for which no global linearizing morphism exists, provided a regular augmented solution path exists in a suitably prolonged (jet) space. Concretely, building the homotopy on a dynamic extension that augments $u$ and its derivatives would let continuation replace exact inversion within the flatness/dynamic-linearization machinery, with the pointwise rank condition taking the place of global integrability.

Two boundaries must be respected, and they delimit the conjecture. First, the approach does nothing for the \emph{non-minimum-phase} obstruction: it regulates the output along a path but does not stabilize unstable zero dynamics. Second, it cannot manufacture \emph{controllability}: transversal (removable) rank drops are repairable, but a genuine loss of control authority in a required direction is not. Within these limits, characterizing exactly which non-feedback-linearizable systems admit a regular augmented path --- the true ``width'' of the class reachable by flowing between morphisms --- is the central open problem this work points to.

\section{Conclusion}\label{sec:conclusion}

We presented a homotopy-continuation controller that tracks references through isolated decoupling singularities of square, relative-degree-one systems, and, via the filtered-error reduction, of any uniform relative degree. Replacing the partial inverse $\Lambda^{-1}$ by the least-norm solution of the augmented map $A = [\Lambda \mid b]$ turns a rank-$(m{-}1)$ singularity into a regular point of $A$ whenever the transversality condition $w^Tb \ne 0$ holds; the resulting flow crosses with bounded control and $O(1/k)$ tracking error and re-locks afterwards (Theorem~\ref{thm:crossing}), and its reflection-versus-crossing behaviour at a fold is fixed by the kernel orientation, the two lifts differing by the $\mathbb Z/2$ monodromy of the fold cover (Theorem~\ref{thm:reflection}).

Two statements delimit the contribution. The method does \emph{not} enlarge the structural class of regulable systems: it presupposes a well-defined relative degree, requires minimum phase for internal boundedness, and depends on the model through the same Lie derivatives as feedback linearization --- indeed, off the singular set it \emph{is} feedback linearization up to $O(1/k)$ (Proposition~\ref{prop:consistency}). What it \emph{does} enlarge is the admissible operating envelope of a given system: it regulates along trajectories that traverse decoupling singularities at which both static and dynamic feedback linearization produce unbounded control, the dynamic-systems analogue of what damped least squares provides over the plain Jacobian inverse. This was demonstrated on the arm, the dynamic relative-degree one and two plants, and the resonant converter, with feedback-linearization demand reaching $10^5$--$10^7$ while the continuation control stayed bounded. The gain is genuine but bounded: at the singular instant the uncontrollable direction $w$ ($w^T\Lambda = 0$) is uncontrollable for any controller, so the method fails gracefully, with a bounded and recoverable error, rather than diverging. Lifting these limits --- adaptive designs on the multi-parameter theory of Section~\ref{sec:adaptive}, higher-order degeneracies, and the open question of which non-feedback-linearizable systems admit a regular augmented path --- is the programme set out in Section~\ref{sec:future}.

\end{document}